%% file: paper.tex
\documentclass[a4paper,USenglish,cleveref]{lipics-v2021}

\hideLIPIcs
\input{fig/macros.tex}
\input{unicode}

\nolinenumbers

\newcommand{\psfrage}[1]{{\color{blue}{\sf[PS: #1]}}} %
\newcommand{\hpfrage}[1]{{\color{violet}\sf[HP: #1]}} %
\newcommand{\gpfrage}[1]{{\color{teal}\sf[GP: #1]}} %
\newcommand{\shfrage}[1]{{\color{brown}\sf[SH: #1]}} %
\newcommand{\swfrage}[1]{{\color{red}\sf[SW: #1]}} %
\renewcommand{\psfrage}[1]{ } \renewcommand{\hpfrage}[1]{ } \renewcommand{\gpfrage}[1]{ } \renewcommand{\shfrage}[1]{ } \renewcommand{\swfrage}[1]{ }

\newcommand{\myparagraph}[1]{\vspace{-3.8mm}\subparagraph*{#1.}}

\def\lipicsLabel#1{\textcolor{lipicsGray}{\sffamily\bfseries\upshape\mathversion{bold}#1}}

\input{pgf_header}

\usepackage{dsfont}
\usepackage{mathtools}
\usepackage{pgfplots}
\usepackage{caption}
\usepackage{subcaption}
\usepackage{hyperref}
\usepackage{booktabs}
\usepackage[nocompress]{cite}
\usepackage{listings}
\usepackage{algorithm}
\usepackage{csquotes}
\usepackage[noend]{algpseudocode}
\usepackage{placeins}
\usepackage{bbding}

\newtheorem{intuition}[theorem]{Intuition}
\def\refRel#1#2#3{\stackrel{\mathclap{\text{#1\,\ref{#2}}}}{#3}}
\def\refRelObs#1#2{\stackrel{\mathclap{\text{(\ref{#1})}}}{#2}\quad}

\crefname{listing}{Algorithm}{Algorithms}
\newcommand{\mytitlerunning}{PHOBIC: Perfect Hashing with Optimized Bucket Sizes and Interleaved Coding}
\title{\texorpdfstring{PHOBIC: Perfect Hashing with\\Optimized Bucket Sizes and Interleaved Coding}{\mytitlerunning}}
\titlerunning{\mytitlerunning}

\author{Stefan Hermann}{Karlsruhe Institute of Technology, Germany}{hermann@kit.edu}{https://orcid.org/0000-0001-9183-2926}{}
\author{Hans-Peter Lehmann}{Karlsruhe Institute of Technology, Germany}{hans-peter.lehmann@kit.edu}{https://orcid.org/0000-0002-0474-1805}{}
\author{Giulio Ermanno Pibiri}{Ca' Foscari University of Venice, Italy\\ISTI-CNR, Italy}{giulioermanno.pibiri@unive.it}{https://orcid.org/0000-0003-0724-7092}{}
\author{Peter Sanders}{Karlsruhe Institute of Technology, Germany}{sanders@kit.edu}{https://orcid.org/0000-0003-3330-9349}{}
\author{Stefan Walzer}{Karlsruhe Institute of Technology, Germany}{stefan.walzer@kit.edu}{https://orcid.org/0000-0002-6477-0106}{}

\newcommand{\myauthorrunning}{S. Hermann, H.-P. Lehmann, G. E. Pibiri, P. Sanders, S. Walzer}
\authorrunning{\myauthorrunning}
\Copyright{Stefan Hermann, Hans-Peter Lehmann, Giulio Ermanno Pibiri, Peter Sanders, Stefan Walzer}
\hypersetup{
  colorlinks=true,
  pdftitle={\mytitlerunning},
  pdfauthor={\myauthorrunning},
  pdfsubject={}
}


\supplementdetails[subcategory={CPU implementation}]{Software}{https://github.com/jermp/pthash/tree/phobic}
\supplementdetails[subcategory={GPU implementation}]{Software}{https://github.com/stefanfred/PHOBIC-GPU}
\supplementdetails[subcategory={Internal comparison}]{Software}{https://github.com/stefanfred/PHOBIC-Scripts}
\supplementdetails[subcategory={Comparison with competitors}]{Software}{https://github.com/ByteHamster/MPHF-Experiments}

\acknowledgements{
This paper is based on and has text overlaps with the Master's thesis~\cite{hermann2023accelerating} of Stefan Hermann.
We refer readers to that thesis for a detailed evaluation and description of the GPU implementation.
}
\funding{
\flag[2cm]{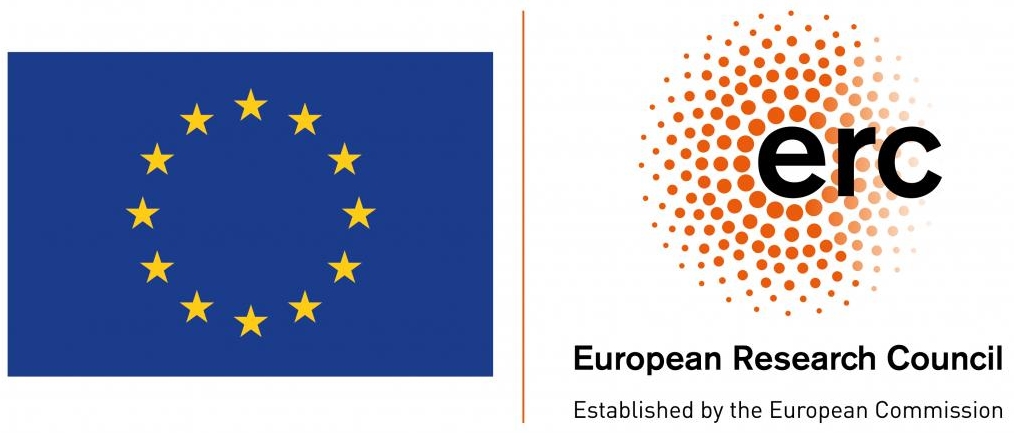}
This project has received funding from the European Research Council (ERC) under the European Union’s Horizon 2020 research and innovation programme (grant agreement No. 882500).
The project also received funding form the European Union’s Horizon Europe research and innovation programme (EFRA project, Grant Agreement Number 101093026).
This work was also partially supported by DAIS – Ca’ Foscari University of Venice within the IRIDE program.
}

\begin{document}
\maketitle

\begin{abstract}
	\input{sections/abstract.tex}

\end{abstract}

\section{Introduction}\label{s:intro}
\input{sections/introduction.tex}

\section{Related Work}\label{s:related}
\input{sections/related_work.tex}

\section{Optimizing Bucket Sizes}\label{s:bucketSizes}
\input{sections/bucketmapping-new.tex}

\section{Fine-Grained Partitioning}\label{s:partitioning}
\input{sections/achieving_locality.tex}

\section{Experiments}\label{s:experiments}
\input{sections/experiments.tex}

\section{Conclusion and Future Work}\label{s:conclusion}
\input{sections/conclusion.tex}

\FloatBarrier
\bibliography{paper}
\clearpage
\appendix

\section{Full proofs}
\input{sections/bucketmapping-proofs.tex}

\subsection{Why buckets should be processed in order of decreasing size}
\input{sections/primary_bucket_ordering}

\subsection{General bounds for Perfect Hashing with Bucket Placement}
\input{sections/bucket-placement-bounds}

\input{sections/non-minimal-bucket-assignment}

\section{Implementation Details}

\subsection{Details on the Bucket Assignment in practice}\label{s:bucketerimpl}
\input{sections/bucketerimpl.tex}
\section{Additional Experimental Data}
\input{sections/additional_experiments}

\end{document}

%% file: fig/macros.tex
\def\secsortdiff{0.044}

\def\abstractGpuBits{2.17}
\def\abstractGpuConstruction{28}
\def\abstractGpuQuery{37}

\def\totalBitImprovement{0.17}

\def\bucketerImprovement{0.14}

\def\interleavedImprovement{0.06}

\def\speedupPtHashVsPhobic{83}

\def\speedupPhobicVsGpu{62}

%% file: unicode.tex
\usepackage{newunicodechar}
\usepackage{silence}
\newunicodechar{♢}{\tikz \node[inner sep=1.5,draw,diamond] {};}
\newunicodechar{☆}{\tikz \node[inner sep=1,draw,star,star point ratio=2] {};}
\newunicodechar{△}{\triangle}
\newunicodechar{⬜}{\kern 0.5pt\tikz \node[inner sep=1.7,draw,regular polygon,regular polygon sides=4] {};\kern 0.5pt}
\newunicodechar{◯}{\tikz[baseline=-3pt] \node[inner sep=1.7,draw,cloud,cloud puffs=4,cloud puff arc=190] {};}

\newunicodechar{⊥}{\bot}
\newunicodechar{•}{\item}
\newunicodechar{✓}{\checkmark}
\newunicodechar{✗}{\xmark}
\newunicodechar{…}{\dots}
\newunicodechar{≔}{\coloneqq}
\newunicodechar{⁻}{^-}
\newunicodechar{⁺}{^+}
\newunicodechar{₋}{_-}
\newunicodechar{₊}{_+}
\newunicodechar{ℓ}{\ell}
\newunicodechar{•}{\item}
\newunicodechar{…}{\dots}
\newunicodechar{≔}{\coloneqq}
\newif\ifnormopen\normopenfalse
\newunicodechar{‖}{\ifnormopen\rVert\normopenfalse\else\rVert\normopentrue\fi}
\newunicodechar{≤}{\leq}
\newunicodechar{≥}{\geq}
\newunicodechar{≰}{\nleq}
\newunicodechar{≱}{\ngeq}
\newunicodechar{⊕}{\oplus}
\newunicodechar{⊗}{\otimes}
\newunicodechar{≠}{\neq}
\newunicodechar{¬}{\neg}
\newunicodechar{≡}{\equiv}
\newunicodechar{₀}{_0}
\newunicodechar{₁}{_1}
\newunicodechar{₂}{_2}
\newunicodechar{₃}{_3}
\newunicodechar{₄}{_4}
\newunicodechar{₅}{_5}
\newunicodechar{₆}{_6}
\newunicodechar{₇}{_7}
\newunicodechar{₈}{_8}
\newunicodechar{₉}{_9}
\newunicodechar{ₚ}{_p}
\newunicodechar{ₙ}{_n}
\newunicodechar{ₐ}{_a}
\newunicodechar{ₑ}{_e}
\newunicodechar{ₕ}{_h}
\newunicodechar{ₖ}{_k}
\newunicodechar{ₗ}{_l}
\newunicodechar{ₘ}{_m}
\newunicodechar{ₛ}{_s}
\newunicodechar{ₜ}{_t}
\newunicodechar{ₓ}{_x}
\newunicodechar{⁰}{^0}
\newunicodechar{¹}{^1}
\newunicodechar{²}{^2}
\newunicodechar{³}{^3}
\newunicodechar{⁴}{^4}
\newunicodechar{⁵}{^5}
\newunicodechar{⁶}{^6}
\newunicodechar{⁷}{^7}
\newunicodechar{⁸}{^8}
\newunicodechar{⁹}{^9}
\newunicodechar{ⁿ}{^n}
\newunicodechar{∈}{\in}
\newunicodechar{∉}{\notin}
\newunicodechar{⊂}{\subset}
\newunicodechar{⊃}{\supset}
\newunicodechar{⊆}{\subseteq}
\newunicodechar{⊇}{\supseteq}
\newunicodechar{⊄}{\nsubset}
\newunicodechar{⊅}{\nsupset}
\newunicodechar{⊈}{\nsubseteq}
\newunicodechar{⊉}{\nsupseteq}
\newunicodechar{∪}{\cup}
\newunicodechar{∩}{\cap}
\newunicodechar{∀}{\forall}
\newunicodechar{∃}{\exists}
\newunicodechar{∄}{\nexists}
\newunicodechar{∨}{\vee}
\newunicodechar{∧}{\wedge}
\newunicodechar{⊼}{\bar{\wedge}}
\newunicodechar{⊽}{\bar{\vee}}
\newunicodechar{ℝ}{\mathbb{R}}
\newunicodechar{ℙ}{\mathbb{P}}
\newunicodechar{ℕ}{\mathbb{N}}
\newunicodechar{𝔼}{\mathbb{E}}
\newunicodechar{𝔽}{\mathbb{F}}
\newunicodechar{ℤ}{\mathbb{Z}}
\newunicodechar{⌊}{\lfloor}
\newunicodechar{⌋}{\rfloor}
\newunicodechar{⌈}{\lceil}
\newunicodechar{⌉}{\rceil}
\newunicodechar{·}{\cdot}
\newunicodechar{∘}{\circ}
\newunicodechar{×}{\times}
\newunicodechar{↑}{\uparrow}
\newunicodechar{↓}{\downarrow}
\newunicodechar{→}{\rightarrow}
\newunicodechar{←}{\leftarrow}
\newunicodechar{⇒}{\Rightarrow}
\newunicodechar{⇐}{\Leftarrow}
\newunicodechar{↔}{\leftrightarrow}
\newunicodechar{⇔}{\Leftrightarrow}
\newunicodechar{↦}{\mapsto}
\newunicodechar{∅}{\varnothing}
\newunicodechar{∞}{\infty}
\newunicodechar{≅}{\cong}
\newunicodechar{≈}{\approx}
\newunicodechar{ℓ}{\ell}
\newunicodechar{𝟙}{\mathds{1}}
\newunicodechar{𝟘}{\mathds{0}}
\newunicodechar{↪}{\hookrightarrow}

\newunicodechar{α}{\alpha}
\newunicodechar{β}{\beta}
\newunicodechar{γ}{\gamma}
\newunicodechar{Γ}{\Gamma}
\newunicodechar{δ}{\delta}
\newunicodechar{Δ}{\Delta}
\newunicodechar{ε}{\varepsilon}
\newunicodechar{ζ}{\zeta}
\newunicodechar{η}{\eta}
\newunicodechar{θ}{\theta}
\newunicodechar{Θ}{\Theta}
\newunicodechar{ι}{\iota}
\newunicodechar{κ}{\kappa}
\newunicodechar{λ}{\lambda}
\newunicodechar{Λ}{\Lambda}
\newunicodechar{μ}{\mu}
\newunicodechar{ν}{\nu}
\newunicodechar{ξ}{\xi}
\newunicodechar{Ξ}{\Xi}
\newunicodechar{π}{\pi}
\newunicodechar{Π}{\Pi}
\newunicodechar{ρ}{\rho}
\newunicodechar{σ}{\sigma}
\newunicodechar{Σ}{\Sigma}
\newunicodechar{τ}{\tau}
\newunicodechar{υ}{\upsilon}
\newunicodechar{ϒ}{\Upsilon}
\newunicodechar{φ}{\varphi}
\newunicodechar{ϕ}{\phi}
\newunicodechar{Φ}{\Phi}
\newunicodechar{χ}{\chi}
\newunicodechar{ψ}{\psi}
\newunicodechar{Ψ}{\Psi}
\newunicodechar{ω}{\omega}
\newunicodechar{Ω}{\Omega}

\newunicodechar{𝒜}{\mathcal{A}}
\newunicodechar{ℬ}{\mathcal{B}}
\newunicodechar{𝒞}{\mathcal{C}}
\newunicodechar{𝒟}{\mathcal{D}}
\newunicodechar{ℰ}{\mathcal{E}}
\newunicodechar{ℱ}{\mathcal{F}}
\newunicodechar{𝒢}{\mathcal{G}}
\newunicodechar{ℋ}{\mathcal{H}}
\newunicodechar{ℐ}{\mathcal{I}}
\newunicodechar{𝒥}{\mathcal{J}}
\newunicodechar{𝒦}{\mathcal{K}}
\newunicodechar{ℒ}{\mathcal{L}}
\newunicodechar{ℳ}{\mathcal{M}}
\newunicodechar{𝒩}{\mathcal{N}}
\newunicodechar{𝒪}{\mathcal{O}}
\newunicodechar{𝒫}{\mathcal{P}}
\newunicodechar{𝒬}{\mathcal{Q}}
\newunicodechar{ℛ}{\mathcal{R}}
\newunicodechar{𝒮}{\mathcal{S}}
\newunicodechar{𝒯}{\mathcal{T}}
\newunicodechar{𝒰}{\mathcal{U}}
\newunicodechar{𝒱}{\mathcal{V}}
\newunicodechar{𝒲}{\mathcal{W}}
\newunicodechar{𝒳}{\mathcal{X}}
\newunicodechar{𝒴}{\mathcal{Y}}
\newunicodechar{𝒵}{\mathcal{Z}}
\newunicodechar{𝒶}{\mathcal{a}}
\newunicodechar{𝒷}{\mathcal{b}}
\newunicodechar{𝒸}{\mathcal{c}}
\newunicodechar{𝒹}{\mathcal{d}}
\newunicodechar{ℯ}{\mathcal{e}}
\newunicodechar{𝒻}{\mathcal{f}}
\newunicodechar{ℊ}{\mathcal{g}}
\newunicodechar{𝒽}{\mathcal{h}}
\newunicodechar{𝒾}{\mathcal{i}}
\newunicodechar{𝒿}{\mathcal{j}}
\newunicodechar{𝓀}{\mathcal{k}}
\newunicodechar{𝓁}{\mathcal{l}}
\newunicodechar{𝓂}{\mathcal{m}}
\newunicodechar{𝓃}{\mathcal{n}}
\newunicodechar{ℴ}{\mathcal{o}}
\newunicodechar{𝓅}{\mathcal{p}}
\newunicodechar{𝓆}{\mathcal{q}}
\newunicodechar{𝓇}{\mathcal{r}}
\newunicodechar{𝓈}{\mathcal{s}}
\newunicodechar{𝓉}{\mathcal{t}}
\newunicodechar{𝓊}{\mathcal{u}}
\newunicodechar{𝓋}{\mathcal{v}}
\newunicodechar{𝓌}{\mathcal{w}}
\newunicodechar{𝓍}{\mathcal{x}}
\newunicodechar{𝓎}{\mathcal{y}}
\newunicodechar{𝓏}{\mathcal{z}}

%% file: pgf_header.tex
\usepackage{pgfplots}
\pgfplotsset{compat=newest}

\usepgfplotslibrary{groupplots}
\pgfplotsset{every axis/.style={scale only axis}}

\definecolor{veryLightGrey}{HTML}{F2F2F2}
\definecolor{lightGrey}{HTML}{DDDDDD}
\definecolor{colorSimdRecSplit}{HTML}{444444}
\definecolor{colorChd}{HTML}{377EB8}
\definecolor{colorRustFmph}{HTML}{A65628}
\definecolor{colorRustFmphGo}{HTML}{A65628}
\definecolor{colorSicHash}{HTML}{4DAF4A}
\definecolor{colorPthash}{HTML}{984EA3}
\definecolor{colorDensePtHash}{HTML}{984EA3}
\definecolor{colorPthashHem}{HTML}{A65628}
\definecolor{colorRecSplit}{HTML}{FF7F00}
\definecolor{colorBbhash}{HTML}{F781BF}
\definecolor{colorShockHash}{HTML}{F8BA01}
\definecolor{colorBipartiteShockHash}{HTML}{F781BF}
\definecolor{colorBipartiteShockHashFlat}{HTML}{377EB8}

\colorlet{colorBruteForce}{colorSicHash}
\colorlet{colorRotationFitting}{colorChd}

\pgfdeclareimage[interpolate,height=1.85mm,width=1.85mm]{shockhashMark}{fig/shockhash}
\pgfdeclareplotmark{shockhash}{\pgftext[at=\pgfpointorigin]{\pgfuseimage{shockhashMark}}}

\pgfplotscreateplotcyclelist{myColorList}{%
  colorSicHash,mark=o\\%
  colorRecSplit,mark=triangle\\%
  colorSimdRecSplit,mark=pentagon\\%
  colorChd,mark=square\\%
  colorPthash,mark=x\\%
  colorBbhash,mark=diamond\\%
}
\pgfplotscreateplotcyclelist{transparentHeatmap}{%
  colorSicHash,mark=*,fill=colorSicHash\\
}
\pgfdeclareplotmark{flippedTriangle}{%
  \pgfpathmoveto{\pgfqpointpolar{-90}{1.2\pgfplotmarksize}}%
  \pgfpathlineto{\pgfqpointpolar{30}{1.2\pgfplotmarksize}}%
  \pgfpathlineto{\pgfqpointpolar{150}{1.2\pgfplotmarksize}}%
  \pgfpathclose%
  \pgfusepath{stroke}
}

\pgfplotsset{
  mark repeat*/.style={
    scatter,
    scatter src=x,
    scatter/@pre marker code/.code={
      \pgfmathtruncatemacro\usemark{
        or(mod(\coordindex,#1)==0, (\coordindex==(\numcoords-1))
      }
      \ifnum\usemark=0
        \pgfplotsset{mark=none}
      \fi
    },
    scatter/@post marker code/.code={}
  },
  major grid style={thin,dotted},
  minor grid style={thin,dotted},
  ymajorgrids,
  yminorgrids,
  every axis/.append style={
    line width=0.7pt,
    tick style={
      line cap=round,
      thin,
      major tick length=4pt,
      minor tick length=2pt,
    },
    mark options={solid},
  },
  legend cell align=left,
  legend style={
    line width=0.7pt,
    /tikz/every even column/.append style={column sep=3mm,black},
    /tikz/every odd column/.append style={black},
    mark options={solid},
  },
  legend style={font=\small},
  title style={yshift=-2pt},
  enlarge x limits=0.04,
  every tick label/.append style={font=\footnotesize},
  every axis label/.append style={font=\small},
  every axis y label/.append style={yshift=-1ex},
  /pgf/number format/1000 sep={},
  axis lines*=left,
  xlabel near ticks,
  ylabel near ticks,
  axis lines*=left,
  label style={font=\footnotesize},
  tick label style={font=\footnotesize},
  cycle list name=myColorList,
  plotPareto/.style={
    xlabel={Space (Bits/Key)},
    width=5.5cm,
    height=3.5cm,
    xmax=3.5,
  },
}

%% file: sections/abstract.tex
A minimal perfect hash function (or MPHF) maps a set of $n$ keys to $[n] := \{1, \ldots, n\}$ without collisions.
Such functions find widespread application e.g. in bioinformatics and databases.
In this paper we revisit PTHash -- a construction technique particularly designed for fast queries.
PTHash distributes the input keys into small buckets and, for each bucket,
it searches for a hash function seed that places its keys in the output domain without collisions. The collection of all seeds is then stored in a compressed manner.
Since the first buckets are easier to place, 
buckets are considered in non-increasing order of size.
Additionally, PTHash heuristically produces an imbalanced distribution of bucket sizes
by distributing 60\% of the keys into 30\% of the buckets.

Our main contribution is to characterize, up to lower order terms, an \emph{optimal} distribution of expected bucket sizes.
We arrive at a simple, closed form solution which improves construction throughput for space efficient configurations in practice.
Our second contribution is a novel encoding scheme for the seeds.
We split the keys into partitions.
Within each partition, we run the bucket distribution and search step.
We then store the seeds in an \emph{interleaved} manner by consecutively placing the seeds for the $i$-th buckets from all partitions.
The seeds for the $i$-th bucket of each partition follow the same statistical distribution.
This allows us to tune a compressor for each bucket.
Hence, we call our technique PHOBIC -- Perfect Hashing with Optimized Bucket sizes and Interleaved Coding.

Compared to PTHash, PHOBIC is \totalBitImprovement{} bits/key more space efficient for same query time and construction throughput.
We also contribute a GPU implementation to further accelerate MPHF construction.
For a configuration with fast queries, PHOBIC-GPU can construct a perfect hash function at \abstractGpuBits{} bits/key in \abstractGpuConstruction{} ns per key, which can be queried in \abstractGpuQuery{} ns per query on the CPU.

%% file: sections/introduction.tex
\label{sec:intro}

A \emph{hash function} maps a set $S$ of $n$ keys to a range of integers $[m] := \{1, \ldots, m\}$, regardless of whether multiple keys collide on the same output.
A \emph{perfect} hash function (PHF) on $S$ is a mapping \emph{without collisions}. This requires $m \geq n$.
The function does not necessarily have to store the keys explicitly.
It only has to store enough information to prevent collisions, which are more likely when $m$ is close to $n$. In the extreme case of $m=n$, the mapping is called a \emph{minimal} perfect hash function (MPHF). In this paper, we consider the minimal case only.
PHFs find widespread practical application e.g. in compressed full-text indexes \cite{belazzougui2014alphabet}, computer networks \cite{lu2006perfect}, databases \cite{chang2005perfect}, prefix-search data structures \cite{belazzougui2010fast}, language models \cite{pibiri2017efficient}, bioinformatics \cite{crawford2018practical,pibiri2022sparse}, and Bloom filters \cite{broder2004network}.
The three main performance attributes of an MPHF are low space consumption, fast construction, and fast queries. Concerning space, the lower bound is $\log_2(e) \approx 1.44$ bits/key \cite{mehlhorn1982program}. Practically viable approaches can get within a few percent of the lower bound, but do so with some sacrifices in running time \cite{lehmann2023shockhash,lehmann2023bipartite}. This paper is concerned with techniques that are focused on achieving fast query times.
For example, this is very important when using perfect hashing to implement a static hash table
that is both space-efficient and allows fast search.

\myparagraph{Perfect Hashing Through Bucket Placement}
Perfect hashing through bucket placement takes the $n$ keys and maps them to small buckets.
For each bucket, it uses brute-force search to find a seed of a hash function such that all keys of the bucket do not collide with each other or previously placed keys.
The first buckets are easier to place because the output domain is less full. Therefore, the methods insert the buckets in order of non-increasing size.
While CHD \cite{belazzougui2009hash} uses buckets of constant expected size,
FCH \cite{fox1992faster} and PTHash \cite{pibiri2021pthash,pibiri2024parallel} set aside 30\% of “heavy” buckets that receive 60\% of the keys in expectation, while 70\% of “light” buckets receive only 40\% of the keys in expectation.
This imbalance in expected bucket sizes improves construction speed by further decreasing the size of the last, hardest to place, buckets.
The resulting list of seed values are stored with various compression techniques, resulting in a variety of trade-offs between space consumption and query speed.

\myparagraph{Partitioning}
Any PHF construction algorithm can be trivially parallelized by splitting the input keys into disjoint subsets. We refer to those subsets as \emph{partitions}. The various PHFs are then logically ``concatenated'' into a single PHF taking the prefix-sum of the partition sizes. The respective offsets have to be looked up when querying a key, imposing some query time overhead. Each partition can be constructed independently in parallel. Partitioning is the usual approach for parallelization, and is applied to PTHash by PTHash-HEM \cite{pibiri2024parallel}.

\myparagraph{Contribution}
This paper aims at improving the space efficiency and construction speed of PTHash, while maintaining its fast query speed. There are three ingredients.
Our main contribution (in \cref{s:bucketSizes}) is to characterize, up to lower order terms, an optimal distribution of expected bucket sizes, effectively taking the imbalance-trick used in FCH and PTHash to its logical conclusion. The distribution is easy-to-implement and greatly improves construction time and space efficiency in practice.
Our second contribution (in \cref{ss:multi}) is to improve the compression of seed values when using partitioning.
Seeds are searched independently for each partition, but compressed together.
We exploit that the seeds of the $i$-th bucket of each partition follow the same statistical distribution.
This allows for tuning a compressor for each such index $i$.
We store the seeds in an \emph{interleaved} manner by consecutively placing the seeds for the $i$-th buckets from all partitions.
Finally, we contribute (in \cref{ss:gpuimpl}) an implementation for \emph{Graphics Processing Units} (GPUs) to speed up construction.

%% file: sections/related_work.tex
Perfect hashing is an active area of research. We provide an overview of state-of-the-art approaches.
For more details, refer to Section 2 of~\cite{pibiri2024parallel}.

\myparagraph{Fingerprinting}
Perfect hashing through fingerprinting \cite{CSLSR11,MSSZ14} is a technique aimed at fast construction and queries at the cost of reduced space efficiency.
The idea is to map the $n$ keys to $\gamma n$ positions using a hash function, where $\gamma$ is a tuning parameter.
A bit vector of length $\gamma n$ indicates positions that received exactly one key.
Keys that take place in collisions are handled recursively on another layer of the same data structure.
A query operation descends through the recursive layers until it finds a 1-bit, meaning that it was the only key mapping to that position.
A rank operation on the bit vector for that position then gives the final MPHF value.
FMPH \cite{beling2023fingerprinting} and BBHash \cite{limasset2017fast} are publicly available implementations of the approach.
FMPHGO \cite{beling2023fingerprinting} extends on this idea using a small number of brute-force re-tries to reduce the number of colliding keys. Fingerprinting based approaches are fast to construct but are outperformed in terms of space consumption and query time by PTHash.%

\myparagraph{Brute Force}
RecSplit \cite{esposito2020recsplit} first partitions the input into sets of equal expected size.
It then recursively splits the key set of each partition until sets of small constant size (usually $\leq 16$) are left.
Within these sets, it finds a perfect hash function by brute force.
RecSplit achieves space usage of about $1.56$ bits/key.
The resulting splitting tree has to be traversed during querying which implies considerably higher query costs compared to PTHash.
The brute force search was later improved in SIMDRecSplit\cite{bez2022high}, which also parallelizes the construction on the GPU.
To the best of our knowledge, RecSplit is the only other PHF construction technique that has a GPU implementation.

\myparagraph{Perfect Hashing Through Retrieval}
In perfect hashing through retrieval, every key has a number of candidate positions, determined by different hash functions.
A retrieval data structure then stores which of the choices should be used for each key.
Note that this implies some query overhead compared to PTHash.
Early implementations include BPZ \cite{botelho2007simple} and GOV \cite{genuzio2016fast}.
SicHash \cite{lehmann2023sichash} reduces space consumption using a mix of different retrieval data structures and some retries.
ShockHash-RS \cite{lehmann2023shockhash,lehmann2023bipartite} combines 1-bit retrieval with
the brute-force approach of RecSplit and currently is the most space-efficient approach to MPHFs with
as little as 1.49 bits/key \cite{lehmann2023bipartite}.

%% file: sections/bucketmapping-new.tex
\def\Po{\mathrm{Po}}
\def\Bin{\mathrm{Bin}}
\def\hx{\hat{x}}
\def\ha{\hat{α}}
\def\med{\mathrm{med}}

Consider perfect hashing through bucket placement with $n$ keys, for $m = n$ and $B$ buckets, i.e.\ an average bucket size of $λ = n/B$.
Previous literature overlooked the simple insight that large $λ$ already guarantees a space consumption close to the lower bound of $\log₂ e$ bits per key, \emph{without any assumptions} on the bucket sizes or their distribution.
\begin{proposition}
	\label{prop:phbp-space}
	Any specialization of perfect hashing through bucket placement requires between $\log₂ e$ bits per key and $\log₂ e + 𝒪(\frac{\log λ}{λ})$ bits per key in expectation.
\end{proposition}
Our goal in this section need therefore only be to minimize construction time. Here we are faced with a lower bound for our family of approaches.
\begin{proposition}
	\label{prop:phbp-time-lowerbound}
	Any specialization of perfect hashing through bucket placement has an expected construction time of $Ω(e^λ/λ)$ per bucket.
\end{proposition}
\Cref{prop:phbp-space,prop:phbp-time-lowerbound} are restated more formally as \cref{prop:bucket-placement-general-bounds} and proved in \cref{sec:bucket-placement-bounds}.
It is intuitively clear (and proved in \cref{prop:largest-to-smallest} in \cref{app:primary-bucket-ordering}) that buckets should be processed in order from largest to smallest. The only remaining degree of freedom is to choose the expected sizes of the buckets.
We characterize asymptotically optimal ways of doing so, formalized by what we call bucket assignment functions and achieving a construction time of $e^{λ(1+ε)}$ per bucket. Proofs are found in the appendix.

\subsection{Bucket Assignment Functions}\label{ss:bas}
Let $w₁,…,w_B$ be the probability that a key hashes to bucket $i$ for $i ∈ [B]$. We may assume without loss of generality that these probabilities are given in decreasing order.
An equivalent view considers the prefix sums $σ_i := w₁+…+w_i$. A key with (normalized) hash value $x ∈ (0,1]$ is then assigned to bucket $i$ if $x ∈ (σ_{i-1},σ_i]$.

We can conveniently represent this information using a \emph{bucket assignment function} $γ : [0,1] → [0,1]$ that: interpolates the points $\{(σ_i,i/B) \mid 0 ≤ i ≤ B\}$, is increasing and smooth on $(0,1)$, and has non-decreasing derivative. The bucket assigned to hash value $x ∈ (0,1]$ is then $⌈γ(x)·B⌉$. It is a non-trivial insight of this section that a single bucket assignment function (not depending on $B$ and $n$) can result in good construction times for many values of $B$ and $n$ simultaneously.

From now on, let $\lambda := n/B$. We summarize some useful intuitions about bucket assignment functions. These intuitions are valid for large $n$ and $B$ (when $γ$, $γ^{-1}$, and $γ'$ are approximately constant on intervals of length $\frac{1}{n}$ and $\frac{1}{B}$). For now, we neglect edge cases related to $γ$ or $γ^{-1}$ not being smooth at $0$ or not being smooth at $1$ (but just smooth on $(0,1)$).

\begin{intuition}\upshape
	\label{int:derivative-and-bucket-sizes}
	Let $x ∈ (0,1]$ be a normalized hash and $b = γ(x)$ the normalized bucket index of the bucket assigned to $x$. Then
	\begin{enumerate}[(i)]
		• The expected size of the bucket assigned to $x$ is $λ/γ'(x)$.\\\emph{(Reason: In the vicinity of $x$ and for infinitesimal $δ$, a $δ$-fraction of the hash range (used by $δn$ keys in expectation) is shared by a $(γ'(x)·δ)$-fraction of the $B$ buckets. The quotient is $δn/(γ'(x)δB) = λ/γ'(x)$.)}
		• The expected size of the bucket with normalized index $b$ is $λ/γ'(γ^{-1}(b)) = λ(γ^{-1})'(b)$.\\\emph{(Follows from \lipicsLabel{(i)} and the inverse function rule.)}
		• The expected size of a bucket is decreasing in its normalized index.\\\emph{(Follows from \lipicsLabel{(ii)} and monotonicity of $γ'$ and $γ^{-1}$.)}
		• The expected fraction of keys with normalized hash in $(0,x]$ is $x$.
		• If $μ > 0$ and $x_μ ∈ (0,1)$ is such that $λ/γ'(x_μ) = μ$ then the expected fraction of keys in buckets of size at least $μ$ is $x_μ$.
            \emph{(Follows from \lipicsLabel{(i)}, \lipicsLabel{(iii)} and \lipicsLabel{(iv)}.)}
	\end{enumerate}
\end{intuition}

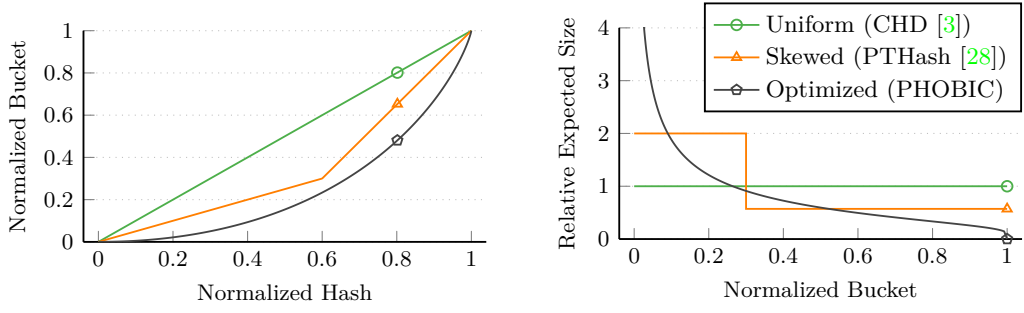
\begin{figure}[t]
	\begin{subfigure}[b]{0.49\textwidth}
		\centering
		\input{fig/bucketasignment.tex}
		\caption{The bucket assignment functions map a normalized hash value $x$ to a normalized bucket index $\gamma(x)$.}
		\label{fig:bucketassign}
	\end{subfigure}
	\hfill
	\begin{subfigure}[b]{0.49\textwidth}
		\centering
		\input{fig/bucketfunction.tex}
		\caption{The expected bucket sizes relative to the average size $\lambda$ are $(\gamma^{-1})'(b)$ for normalized bucket index $b$.}
		\label{fig:bucketfunction}
	\end{subfigure}
	\caption{Comparison of bucket assignment functions $\gamma(x)$ of related work and PHOBIC ($\gamma=\beta_*$). }
	\label{fig:bucket-assignment-and-sizes}
\end{figure}

\subsection{An Optimal Bucket Assignment Function}

Intuitively, we identify the following bucket assignment function to be optimal, although our precise result stated below is more subtle.
\[ β_*(x) = x + (1-x)·\ln(1-x) \text{ with derivative } β_*'(x) = -\ln(1-x).\]
For comparison, \Cref{fig:bucketassign} shows $β_*$ as well as the bucket assignment functions used by CHD and PTHash. In \cref{fig:bucketfunction} we see the distribution of expected bucket sizes, which is uniform for CHD, imbalanced for PTHash, and even more imbalanced for $β_*$.

Recall that $\lambda = n/B$ is the average bucket size.
By \cref{prop:phbp-time-lowerbound} a lower bound for the expected work is $Ω(n·e^\lambda/λ)$. We prove, firstly, that any bucket assignment function $γ$ that differs from $β_*$ leads to expected work exceeding $n·e^{(1+ε)\lambda}$ for some $ε = ε(γ)$, provided that $\lambda$ is large enough. Conversely, we show that a slight perturbation $β_ε$ of $β_*$ leads to an expected work of essentially at most $n·e^{(1+ε)\lambda}$ for any $ε > 0$, provided that $\lambda ≥ \lambda₀(ε)$ is large enough. Essentially, we can get arbitrarily close to a cost of $e^\lambda$ per key and only functions close to $β_*$ can achieve this.

\def\wcoupon{w_{\mathrm{coupon}}}
Our results bound the work $w_{n,\lambda}(γ)$ associated with $γ$ and involve a “coupon collector term” $\wcoupon$, which is the work required to place buckets of size $1$. We will define these more precisely below. Proofs are found in \cref{sec:not-beta-is-bad,sec:beta-is-good}. We have reason to believe that our results generalize for the non-minimal case of $m>n$, as explained in \cref{s:nonminimal}.

\begin{theorem}
	\label{thm:not-beta-is-bad}
	Let $γ:[0,1]→[0,1]$ be a continuous bucket assignment function that is smooth on $(0,1)$ with non-decreasing derivative. If $β_* ≠ γ$ then
	\[ ∃ ε > 0: ∀\lambda ≥ \lambda_0(ε): ∀n≥n₀(\lambda,ε): w_{n,\lambda}(γ) ≥ n·e^{\lambda(1+ε)}+\wcoupon\text{ whp}.\]
\end{theorem}
While this leaves the relationship between $γ$ and $ε(γ)$ open, our proof suggests that any $ε < \sup_{x ∈ (0,1)} \frac{β'_*(x)}{γ'(x)}-1$ is a possible choice.
\begin{theorem}
	\label{thm:beta-is-good}
	Let $β_ε(x) := εx + (1-ε)β_*(x)$ for some $ε > 0$. Then
	\[ ∀ε > 0: ∀\lambda ≥ \lambda_0(ε): ∀n ≥ n₀(\lambda,ε): w_{n,\lambda}(β_ε) ≤ n·e^{\lambda(1+𝒪(ε))} + \wcoupon\text{ whp}. \]
\end{theorem}
By \emph{with high probability} (whp) we mean probability $1-𝒪(n^{-c})$ for some $c > 0$. Note that both theorems are phrased such that we may assume that $n$ is much larger than $\lambda$ and $\lambda$ is much larger than $1/ε$. We give implementation details concerning the use of $β_ε$ in \cref{s:bucketerimpl}.

\myparagraph{What we do not prove} Note that our analysis leaves undecided whether $β_*$ is itself a “good” bucket assignment function, i.e.\ whether $w_{n,\lambda}(β_*)$ approaches $e^\lambda$ in any meaningful sense. We suspect that it does. However, the perturbation simplifies the analysis and improves running times in practice. Our analysis also does not imply that the particular perturbation we choose is the best choice: there may be an alternative to $β_ε$ such that the overall work approaches $e^\lambda$ more quickly.

\myparagraph{The work associated with a bucket assignment function}
To place a bucket of size $s ∈ ℕ$ into a hash table of size $n$ that already has load factor $α ∈ [0,1-\frac{s}{n}]$ we repeatedly try seeds for a hash function mapping keys to $[n]$, until all keys hash to free positions. The expected cost $cₙ(s,α)$ associated with this task under the simple uniform hashing assumption is described precisely in \cref{app:primary-bucket-ordering}. We have to take into account self-collisions, i.e.\ while checking the keys one after the other, the load factor gradually increases and is $α' = α+\frac{s-1}{n}$ for the last key. For our purposes, the following bounds on $cₙ(s,α)$ suffice
\begin{equation}
	(1-α)^{-s} ≤ cₙ(s,α) ≤ s·(1-α')^{-s}. \label{eq:bucket-cost-bound}
\end{equation}
This uses that $(1-α)^{-s}$ and $(1-α')^{-s}$ are lower and upper bounds on the number of seeds that have to be tried and that, to test a seed, at least $1$ and at most $s$ keys have to be considered.
Now assume we are given a bucket assignment function $γ$ as well as $n ∈ ℕ$, $\lambda ∈ ℝ₊$ and $B = n/\lambda$. By assigning keys to buckets according to $γ$ and hash values in $(0,1]$ we obtain buckets. Let $s₁ ≥ … ≥ s_B$ be their sizes in decreasing order. By \cref{prop:largest-to-smallest} in \cref{app:primary-bucket-ordering} it is advantageous to process the buckets in this order. Defining $α_i := \frac{1}{n}\sum_{j = 1}^{i-1} s_i$, the total cost is then
$w_{n,\lambda}(γ) = \sum_{i = 1}^{B} cₙ(s_i,α_i)$.

Note that while this describes the \emph{expected} cost when \emph{given} $(s_i)_{i ∈ [B]}$, overall $w_{n,\lambda}(γ)$ is still a random variable because the numbers $(s_i)_{i ∈ [B]}$ are random.
Assume now that there are exactly $k$ buckets of size $1$ that are placed last (we may ignore buckets of size $0$). For these buckets, the upper and lower bounds in \cref{eq:bucket-cost-bound} coincide so they incur a cost of
\[ \wcoupon := \sum_{i = 1}^{k} c(1,\tfrac{n-i}{n}) = \sum_{i = 1}^{k} \frac{n}{i} = n·H_k. \]
Here $H_k$ is the $k$th harmonic number, which satisfies $H_k = Θ(\log k)$. If $n$ is sufficiently large compared to $\lambda$ then we have $k ≥ n^d$ whp for some constant $d > 0$, giving a cost of $Θ(n·H_{n^d}) = Θ(n\log n)$. This dominates overall construction time if $n$ is sufficiently large compared to $\lambda$. Our theorems list this work for buckets of size $1$ as a separate term because there are techniques to mitigate the problem: The hash function may permit to directly compute for a given key $x$ and table position $i$ a seed for which $x$ is mapped to $i$. This is the case if the seed includes an additive displacement term, as is the case in our implementation and in FCH~\cite{fox1992faster}.

\myparagraph{Intuition: What makes $β_*$ uniquely promising}
For $\mu > 0$ let $x_\mu ∈ (0,1)$ be such that $\lambda/β'_*(x_\mu) = \mu$. By \cref{int:derivative-and-bucket-sizes} (v), roughly an expected $x_\mu$-fraction of the keys (those with hashes in $[0,x_\mu]$) land in buckets of expected size at least $\mu$. Assume for now that a bucket of expected size $\mu$ has actual size $\mu$ (ignoring the issue that $\mu$ may not be integer). Then, since we process buckets in order of increasing size, we would process a bucket of size $\mu$ when the load factor is $x_\mu$. The expected cost for this is, by \cref{eq:bucket-cost-bound}, around $(1-x_\mu)^{-\mu}$. Using that $β'_*(x) = -\log(1-x)$ gives $\mu = -\lambda/\log(1-x)$ and hence $(1-x_\mu)^{-\mu} = (1-x_\mu)^{\lambda/\log(1-x_\mu)} = e^\lambda$, i.e.\ a cost independent of $\mu$. The idea behind \cref{thm:not-beta-is-bad} is that any bucket assignment function $γ ≠ β_*$ fails to balance bucket sizes in this way, leading to significantly higher costs.

The simplification we made seems innocent for large $\mu$ since a bucket of large expected size typically has actual size close to its expectation. But consider $\mu = 1.5$ and cease to ignore rounding issues. If at load factor $x_\mu$ we would process buckets of size $1$ half the time and buckets of size $2$ half the time, the resulting costs are around $e^{\frac 23A}$ and $e^{\frac 43A}$, respectively, which does not average out to $e^\lambda$. Luckily, things are more complicated. It turns out that for small $s ∈ ℕ$ and assuming large $\lambda$, the \emph{expected number} of buckets of size $s$ is meaningfully greater than the number of buckets of \emph{expected size} in the range $[s-1,s+1]$. This means that we get more small buckets than we seem to have called for, decreasing costs at high load factors. It seems clear that the flipside of this beneficial effect must be a detrimental effect for larger bucket sizes that a proof of \cref{thm:beta-is-good} must quantify. When it comes to \emph{very} large bucket sizes, we bail ourselves out by using $β_ε$ instead of $β_*$: since $β_ε'$ is lower bounded by $ε$, the expected bucket sizes are capped at $\lambda/ε$. It is buckets of intermediate sizes that have to pay the price.

%% file: fig/bucketasignment.tex
\centering
\begin{tikzpicture}

\begin{axis}[
	ylabel={Normalized Bucket},
	xlabel={Normalized Hash},
	width=5.3cm,	
	height=2.8cm,	
	ymax=1,
	ymin=0,
	legend style={at={(1.0,0.6)},anchor=south east},
	mark options={mark indices=100},
	]
	
  \addplot coordinates { (1.0,1.0) (0.997996,0.997996) (0.995992,0.995992) (0.993988,0.993988) (0.991984,0.991984) (0.98998,0.98998) (0.987976,0.987976) (0.985972,0.985972) (0.983968,0.983968) (0.981964,0.981964) (0.97996,0.97996) (0.977956,0.977956) (0.975952,0.975952) (0.973948,0.973948) (0.971944,0.971944) (0.96994,0.96994) (0.967936,0.967936) (0.965932,0.965932) (0.963928,0.963928) (0.961924,0.961924) (0.95992,0.95992) (0.957916,0.957916) (0.955912,0.955912) (0.953908,0.953908) (0.951904,0.951904) (0.9499,0.9499) (0.947896,0.947896) (0.945892,0.945892) (0.943888,0.943888) (0.941884,0.941884) (0.93988,0.93988) (0.937876,0.937876) (0.935872,0.935872) (0.933868,0.933868) (0.931864,0.931864) (0.92986,0.92986) (0.927856,0.927856) (0.925852,0.925852) (0.923848,0.923848) (0.921844,0.921844) (0.91984,0.91984) (0.917836,0.917836) (0.915832,0.915832) (0.913828,0.913828) (0.911824,0.911824) (0.90982,0.90982) (0.907816,0.907816) (0.905812,0.905812) (0.903808,0.903808) (0.901804,0.901804) (0.8998,0.8998) (0.897796,0.897796) (0.895792,0.895792) (0.893788,0.893788) (0.891784,0.891784) (0.88978,0.88978) (0.887776,0.887776) (0.885772,0.885772) (0.883768,0.883768) (0.881764,0.881764) (0.87976,0.87976) (0.877756,0.877756) (0.875752,0.875752) (0.873747,0.873747) (0.871743,0.871743) (0.869739,0.869739) (0.867735,0.867735) (0.865731,0.865731) (0.863727,0.863727) (0.861723,0.861723) (0.859719,0.859719) (0.857715,0.857715) (0.855711,0.855711) (0.853707,0.853707) (0.851703,0.851703) (0.849699,0.849699) (0.847695,0.847695) (0.845691,0.845691) (0.843687,0.843687) (0.841683,0.841683) (0.839679,0.839679) (0.837675,0.837675) (0.835671,0.835671) (0.833667,0.833667) (0.831663,0.831663) (0.829659,0.829659) (0.827655,0.827655) (0.825651,0.825651) (0.823647,0.823647) (0.821643,0.821643) (0.819639,0.819639) (0.817635,0.817635) (0.815631,0.815631) (0.813627,0.813627) (0.811623,0.811623) (0.809619,0.809619) (0.807615,0.807615) (0.805611,0.805611) (0.803607,0.803607) (0.801603,0.801603) (0.799599,0.799599) (0.797595,0.797595) (0.795591,0.795591) (0.793587,0.793587) (0.791583,0.791583) (0.789579,0.789579) (0.787575,0.787575) (0.785571,0.785571) (0.783567,0.783567) (0.781563,0.781563) (0.779559,0.779559) (0.777555,0.777555) (0.775551,0.775551) (0.773547,0.773547) (0.771543,0.771543) (0.769539,0.769539) (0.767535,0.767535) (0.765531,0.765531) (0.763527,0.763527) (0.761523,0.761523) (0.759519,0.759519) (0.757515,0.757515) (0.755511,0.755511) (0.753507,0.753507) (0.751503,0.751503) (0.749499,0.749499) (0.747495,0.747495) (0.745491,0.745491) (0.743487,0.743487) (0.741483,0.741483) (0.739479,0.739479) (0.737475,0.737475) (0.735471,0.735471) (0.733467,0.733467) (0.731463,0.731463) (0.729459,0.729459) (0.727455,0.727455) (0.725451,0.725451) (0.723447,0.723447) (0.721443,0.721443) (0.719439,0.719439) (0.717435,0.717435) (0.715431,0.715431) (0.713427,0.713427) (0.711423,0.711423) (0.709419,0.709419) (0.707415,0.707415) (0.705411,0.705411) (0.703407,0.703407) (0.701403,0.701403) (0.699399,0.699399) (0.697395,0.697395) (0.695391,0.695391) (0.693387,0.693387) (0.691383,0.691383) (0.689379,0.689379) (0.687375,0.687375) (0.685371,0.685371) (0.683367,0.683367) (0.681363,0.681363) (0.679359,0.679359) (0.677355,0.677355) (0.675351,0.675351) (0.673347,0.673347) (0.671343,0.671343) (0.669339,0.669339) (0.667335,0.667335) (0.665331,0.665331) (0.663327,0.663327) (0.661323,0.661323) (0.659319,0.659319) (0.657315,0.657315) (0.655311,0.655311) (0.653307,0.653307) (0.651303,0.651303) (0.649299,0.649299) (0.647295,0.647295) (0.645291,0.645291) (0.643287,0.643287) (0.641283,0.641283) (0.639279,0.639279) (0.637275,0.637275) (0.635271,0.635271) (0.633267,0.633267) (0.631263,0.631263) (0.629259,0.629259) (0.627255,0.627255) (0.625251,0.625251) (0.623246,0.623246) (0.621242,0.621242) (0.619238,0.619238) (0.617234,0.617234) (0.61523,0.61523) (0.613226,0.613226) (0.611222,0.611222) (0.609218,0.609218) (0.607214,0.607214) (0.60521,0.60521) (0.603206,0.603206) (0.601202,0.601202) (0.599198,0.599198) (0.597194,0.597194) (0.59519,0.59519) (0.593186,0.593186) (0.591182,0.591182) (0.589178,0.589178) (0.587174,0.587174) (0.58517,0.58517) (0.583166,0.583166) (0.581162,0.581162) (0.579158,0.579158) (0.577154,0.577154) (0.57515,0.57515) (0.573146,0.573146) (0.571142,0.571142) (0.569138,0.569138) (0.567134,0.567134) (0.56513,0.56513) (0.563126,0.563126) (0.561122,0.561122) (0.559118,0.559118) (0.557114,0.557114) (0.55511,0.55511) (0.553106,0.553106) (0.551102,0.551102) (0.549098,0.549098) (0.547094,0.547094) (0.54509,0.54509) (0.543086,0.543086) (0.541082,0.541082) (0.539078,0.539078) (0.537074,0.537074) (0.53507,0.53507) (0.533066,0.533066) (0.531062,0.531062) (0.529058,0.529058) (0.527054,0.527054) (0.52505,0.52505) (0.523046,0.523046) (0.521042,0.521042) (0.519038,0.519038) (0.517034,0.517034) (0.51503,0.51503) (0.513026,0.513026) (0.511022,0.511022) (0.509018,0.509018) (0.507014,0.507014) (0.50501,0.50501) (0.503006,0.503006) (0.501002,0.501002) (0.498998,0.498998) (0.496994,0.496994) (0.49499,0.49499) (0.492986,0.492986) (0.490982,0.490982) (0.488978,0.488978) (0.486974,0.486974) (0.48497,0.48497) (0.482966,0.482966) (0.480962,0.480962) (0.478958,0.478958) (0.476954,0.476954) (0.47495,0.47495) (0.472946,0.472946) (0.470942,0.470942) (0.468938,0.468938) (0.466934,0.466934) (0.46493,0.46493) (0.462926,0.462926) (0.460922,0.460922) (0.458918,0.458918) (0.456914,0.456914) (0.45491,0.45491) (0.452906,0.452906) (0.450902,0.450902) (0.448898,0.448898) (0.446894,0.446894) (0.44489,0.44489) (0.442886,0.442886) (0.440882,0.440882) (0.438878,0.438878) (0.436874,0.436874) (0.43487,0.43487) (0.432866,0.432866) (0.430862,0.430862) (0.428858,0.428858) (0.426854,0.426854) (0.42485,0.42485) (0.422846,0.422846) (0.420842,0.420842) (0.418838,0.418838) (0.416834,0.416834) (0.41483,0.41483) (0.412826,0.412826) (0.410822,0.410822) (0.408818,0.408818) (0.406814,0.406814) (0.40481,0.40481) (0.402806,0.402806) (0.400802,0.400802) (0.398798,0.398798) (0.396794,0.396794) (0.39479,0.39479) (0.392786,0.392786) (0.390782,0.390782) (0.388778,0.388778) (0.386774,0.386774) (0.38477,0.38477) (0.382766,0.382766) (0.380762,0.380762) (0.378758,0.378758) (0.376754,0.376754) (0.374749,0.374749) (0.372745,0.372745) (0.370741,0.370741) (0.368737,0.368737) (0.366733,0.366733) (0.364729,0.364729) (0.362725,0.362725) (0.360721,0.360721) (0.358717,0.358717) (0.356713,0.356713) (0.354709,0.354709) (0.352705,0.352705) (0.350701,0.350701) (0.348697,0.348697) (0.346693,0.346693) (0.344689,0.344689) (0.342685,0.342685) (0.340681,0.340681) (0.338677,0.338677) (0.336673,0.336673) (0.334669,0.334669) (0.332665,0.332665) (0.330661,0.330661) (0.328657,0.328657) (0.326653,0.326653) (0.324649,0.324649) (0.322645,0.322645) (0.320641,0.320641) (0.318637,0.318637) (0.316633,0.316633) (0.314629,0.314629) (0.312625,0.312625) (0.310621,0.310621) (0.308617,0.308617) (0.306613,0.306613) (0.304609,0.304609) (0.302605,0.302605) (0.300601,0.300601) (0.298597,0.298597) (0.296593,0.296593) (0.294589,0.294589) (0.292585,0.292585) (0.290581,0.290581) (0.288577,0.288577) (0.286573,0.286573) (0.284569,0.284569) (0.282565,0.282565) (0.280561,0.280561) (0.278557,0.278557) (0.276553,0.276553) (0.274549,0.274549) (0.272545,0.272545) (0.270541,0.270541) (0.268537,0.268537) (0.266533,0.266533) (0.264529,0.264529) (0.262525,0.262525) (0.260521,0.260521) (0.258517,0.258517) (0.256513,0.256513) (0.254509,0.254509) (0.252505,0.252505) (0.250501,0.250501) (0.248497,0.248497) (0.246493,0.246493) (0.244489,0.244489) (0.242485,0.242485) (0.240481,0.240481) (0.238477,0.238477) (0.236473,0.236473) (0.234469,0.234469) (0.232465,0.232465) (0.230461,0.230461) (0.228457,0.228457) (0.226453,0.226453) (0.224449,0.224449) (0.222445,0.222445) (0.220441,0.220441) (0.218437,0.218437) (0.216433,0.216433) (0.214429,0.214429) (0.212425,0.212425) (0.210421,0.210421) (0.208417,0.208417) (0.206413,0.206413) (0.204409,0.204409) (0.202405,0.202405) (0.200401,0.200401) (0.198397,0.198397) (0.196393,0.196393) (0.194389,0.194389) (0.192385,0.192385) (0.190381,0.190381) (0.188377,0.188377) (0.186373,0.186373) (0.184369,0.184369) (0.182365,0.182365) (0.180361,0.180361) (0.178357,0.178357) (0.176353,0.176353) (0.174349,0.174349) (0.172345,0.172345) (0.170341,0.170341) (0.168337,0.168337) (0.166333,0.166333) (0.164329,0.164329) (0.162325,0.162325) (0.160321,0.160321) (0.158317,0.158317) (0.156313,0.156313) (0.154309,0.154309) (0.152305,0.152305) (0.150301,0.150301) (0.148297,0.148297) (0.146293,0.146293) (0.144289,0.144289) (0.142285,0.142285) (0.140281,0.140281) (0.138277,0.138277) (0.136273,0.136273) (0.134269,0.134269) (0.132265,0.132265) (0.130261,0.130261) (0.128257,0.128257) (0.126253,0.126253) (0.124248,0.124248) (0.122244,0.122244) (0.12024,0.12024) (0.118236,0.118236) (0.116232,0.116232) (0.114228,0.114228) (0.112224,0.112224) (0.11022,0.11022) (0.108216,0.108216) (0.106212,0.106212) (0.104208,0.104208) (0.102204,0.102204) (0.1002,0.1002) (0.0981964,0.0981964) (0.0961924,0.0961924) (0.0941884,0.0941884) (0.0921844,0.0921844) (0.0901804,0.0901804) (0.0881764,0.0881764) (0.0861723,0.0861723) (0.0841683,0.0841683) (0.0821643,0.0821643) (0.0801603,0.0801603) (0.0781563,0.0781563) (0.0761523,0.0761523) (0.0741483,0.0741483) (0.0721443,0.0721443) (0.0701403,0.0701403) (0.0681363,0.0681363) (0.0661323,0.0661323) (0.0641283,0.0641283) (0.0621242,0.0621242) (0.0601202,0.0601202) (0.0581162,0.0581162) (0.0561122,0.0561122) (0.0541082,0.0541082) (0.0521042,0.0521042) (0.0501002,0.0501002) (0.0480962,0.0480962) (0.0460922,0.0460922) (0.0440882,0.0440882) (0.0420842,0.0420842) (0.0400802,0.0400802) (0.0380762,0.0380762) (0.0360721,0.0360721) (0.0340681,0.0340681) (0.0320641,0.0320641) (0.0300601,0.0300601) (0.0280561,0.0280561) (0.0260521,0.0260521) (0.0240481,0.0240481) (0.0220441,0.0220441) (0.0200401,0.0200401) (0.0180361,0.0180361) (0.0160321,0.0160321) (0.0140281,0.0140281) (0.012024,0.012024) (0.01002,0.01002) (0.00801603,0.00801603) (0.00601202,0.00601202) (0.00400802,0.00400802) (0.00200401,0.00200401) (0.0,0.0) };
  
  \addlegendentry{Uniform (CHD \cite{belazzougui2009hash})};

  \addplot coordinates { (1.0,1.0) (0.997996,0.996493) (0.995992,0.992986) (0.993988,0.989479) (0.991984,0.985972) (0.98998,0.982465) (0.987976,0.978958) (0.985972,0.975451) (0.983968,0.971944) (0.981964,0.968437) (0.97996,0.96493) (0.977956,0.961423) (0.975952,0.957916) (0.973948,0.954409) (0.971944,0.950902) (0.96994,0.947395) (0.967936,0.943888) (0.965932,0.940381) (0.963928,0.936874) (0.961924,0.933367) (0.95992,0.92986) (0.957916,0.926353) (0.955912,0.922846) (0.953908,0.919339) (0.951904,0.915832) (0.9499,0.912325) (0.947896,0.908818) (0.945892,0.905311) (0.943888,0.901804) (0.941884,0.898297) (0.93988,0.89479) (0.937876,0.891283) (0.935872,0.887776) (0.933868,0.884269) (0.931864,0.880762) (0.92986,0.877255) (0.927856,0.873747) (0.925852,0.87024) (0.923848,0.866733) (0.921844,0.863226) (0.91984,0.859719) (0.917836,0.856212) (0.915832,0.852705) (0.913828,0.849198) (0.911824,0.845691) (0.90982,0.842184) (0.907816,0.838677) (0.905812,0.83517) (0.903808,0.831663) (0.901804,0.828156) (0.8998,0.824649) (0.897796,0.821142) (0.895792,0.817635) (0.893788,0.814128) (0.891784,0.810621) (0.88978,0.807114) (0.887776,0.803607) (0.885772,0.8001) (0.883768,0.796593) (0.881764,0.793086) (0.87976,0.789579) (0.877756,0.786072) (0.875752,0.782565) (0.873747,0.779058) (0.871743,0.775551) (0.869739,0.772044) (0.867735,0.768537) (0.865731,0.76503) (0.863727,0.761523) (0.861723,0.758016) (0.859719,0.754509) (0.857715,0.751002) (0.855711,0.747495) (0.853707,0.743988) (0.851703,0.740481) (0.849699,0.736974) (0.847695,0.733467) (0.845691,0.72996) (0.843687,0.726453) (0.841683,0.722946) (0.839679,0.719439) (0.837675,0.715932) (0.835671,0.712425) (0.833667,0.708918) (0.831663,0.705411) (0.829659,0.701904) (0.827655,0.698397) (0.825651,0.69489) (0.823647,0.691383) (0.821643,0.687876) (0.819639,0.684369) (0.817635,0.680862) (0.815631,0.677355) (0.813627,0.673848) (0.811623,0.670341) (0.809619,0.666834) (0.807615,0.663327) (0.805611,0.65982) (0.803607,0.656313) (0.801603,0.652806) (0.799599,0.649299) (0.797595,0.645792) (0.795591,0.642285) (0.793587,0.638778) (0.791583,0.635271) (0.789579,0.631764) (0.787575,0.628257) (0.785571,0.624749) (0.783567,0.621242) (0.781563,0.617735) (0.779559,0.614228) (0.777555,0.610721) (0.775551,0.607214) (0.773547,0.603707) (0.771543,0.6002) (0.769539,0.596693) (0.767535,0.593186) (0.765531,0.589679) (0.763527,0.586172) (0.761523,0.582665) (0.759519,0.579158) (0.757515,0.575651) (0.755511,0.572144) (0.753507,0.568637) (0.751503,0.56513) (0.749499,0.561623) (0.747495,0.558116) (0.745491,0.554609) (0.743487,0.551102) (0.741483,0.547595) (0.739479,0.544088) (0.737475,0.540581) (0.735471,0.537074) (0.733467,0.533567) (0.731463,0.53006) (0.729459,0.526553) (0.727455,0.523046) (0.725451,0.519539) (0.723447,0.516032) (0.721443,0.512525) (0.719439,0.509018) (0.717435,0.505511) (0.715431,0.502004) (0.713427,0.498497) (0.711423,0.49499) (0.709419,0.491483) (0.707415,0.487976) (0.705411,0.484469) (0.703407,0.480962) (0.701403,0.477455) (0.699399,0.473948) (0.697395,0.470441) (0.695391,0.466934) (0.693387,0.463427) (0.691383,0.45992) (0.689379,0.456413) (0.687375,0.452906) (0.685371,0.449399) (0.683367,0.445892) (0.681363,0.442385) (0.679359,0.438878) (0.677355,0.435371) (0.675351,0.431864) (0.673347,0.428357) (0.671343,0.42485) (0.669339,0.421343) (0.667335,0.417836) (0.665331,0.414329) (0.663327,0.410822) (0.661323,0.407315) (0.659319,0.403808) (0.657315,0.400301) (0.655311,0.396794) (0.653307,0.393287) (0.651303,0.38978) (0.649299,0.386273) (0.647295,0.382766) (0.645291,0.379259) (0.643287,0.375752) (0.641283,0.372244) (0.639279,0.368737) (0.637275,0.36523) (0.635271,0.361723) (0.633267,0.358216) (0.631263,0.354709) (0.629259,0.351202) (0.627255,0.347695) (0.625251,0.344188) (0.623246,0.340681) (0.621242,0.337174) (0.619238,0.333667) (0.617234,0.33016) (0.61523,0.326653) (0.613226,0.323146) (0.611222,0.319639) (0.609218,0.316132) (0.607214,0.312625) (0.60521,0.309118) (0.603206,0.305611) (0.601202,0.302104) (0.599198,0.299599) (0.597194,0.298597) (0.59519,0.297595) (0.593186,0.296593) (0.591182,0.295591) (0.589178,0.294589) (0.587174,0.293587) (0.58517,0.292585) (0.583166,0.291583) (0.581162,0.290581) (0.579158,0.289579) (0.577154,0.288577) (0.57515,0.287575) (0.573146,0.286573) (0.571142,0.285571) (0.569138,0.284569) (0.567134,0.283567) (0.56513,0.282565) (0.563126,0.281563) (0.561122,0.280561) (0.559118,0.279559) (0.557114,0.278557) (0.55511,0.277555) (0.553106,0.276553) (0.551102,0.275551) (0.549098,0.274549) (0.547094,0.273547) (0.54509,0.272545) (0.543086,0.271543) (0.541082,0.270541) (0.539078,0.269539) (0.537074,0.268537) (0.53507,0.267535) (0.533066,0.266533) (0.531062,0.265531) (0.529058,0.264529) (0.527054,0.263527) (0.52505,0.262525) (0.523046,0.261523) (0.521042,0.260521) (0.519038,0.259519) (0.517034,0.258517) (0.51503,0.257515) (0.513026,0.256513) (0.511022,0.255511) (0.509018,0.254509) (0.507014,0.253507) (0.50501,0.252505) (0.503006,0.251503) (0.501002,0.250501) (0.498998,0.249499) (0.496994,0.248497) (0.49499,0.247495) (0.492986,0.246493) (0.490982,0.245491) (0.488978,0.244489) (0.486974,0.243487) (0.48497,0.242485) (0.482966,0.241483) (0.480962,0.240481) (0.478958,0.239479) (0.476954,0.238477) (0.47495,0.237475) (0.472946,0.236473) (0.470942,0.235471) (0.468938,0.234469) (0.466934,0.233467) (0.46493,0.232465) (0.462926,0.231463) (0.460922,0.230461) (0.458918,0.229459) (0.456914,0.228457) (0.45491,0.227455) (0.452906,0.226453) (0.450902,0.225451) (0.448898,0.224449) (0.446894,0.223447) (0.44489,0.222445) (0.442886,0.221443) (0.440882,0.220441) (0.438878,0.219439) (0.436874,0.218437) (0.43487,0.217435) (0.432866,0.216433) (0.430862,0.215431) (0.428858,0.214429) (0.426854,0.213427) (0.42485,0.212425) (0.422846,0.211423) (0.420842,0.210421) (0.418838,0.209419) (0.416834,0.208417) (0.41483,0.207415) (0.412826,0.206413) (0.410822,0.205411) (0.408818,0.204409) (0.406814,0.203407) (0.40481,0.202405) (0.402806,0.201403) (0.400802,0.200401) (0.398798,0.199399) (0.396794,0.198397) (0.39479,0.197395) (0.392786,0.196393) (0.390782,0.195391) (0.388778,0.194389) (0.386774,0.193387) (0.38477,0.192385) (0.382766,0.191383) (0.380762,0.190381) (0.378758,0.189379) (0.376754,0.188377) (0.374749,0.187375) (0.372745,0.186373) (0.370741,0.185371) (0.368737,0.184369) (0.366733,0.183367) (0.364729,0.182365) (0.362725,0.181363) (0.360721,0.180361) (0.358717,0.179359) (0.356713,0.178357) (0.354709,0.177355) (0.352705,0.176353) (0.350701,0.175351) (0.348697,0.174349) (0.346693,0.173347) (0.344689,0.172345) (0.342685,0.171343) (0.340681,0.170341) (0.338677,0.169339) (0.336673,0.168337) (0.334669,0.167335) (0.332665,0.166333) (0.330661,0.165331) (0.328657,0.164329) (0.326653,0.163327) (0.324649,0.162325) (0.322645,0.161323) (0.320641,0.160321) (0.318637,0.159319) (0.316633,0.158317) (0.314629,0.157315) (0.312625,0.156313) (0.310621,0.155311) (0.308617,0.154309) (0.306613,0.153307) (0.304609,0.152305) (0.302605,0.151303) (0.300601,0.150301) (0.298597,0.149299) (0.296593,0.148297) (0.294589,0.147295) (0.292585,0.146293) (0.290581,0.145291) (0.288577,0.144289) (0.286573,0.143287) (0.284569,0.142285) (0.282565,0.141283) (0.280561,0.140281) (0.278557,0.139279) (0.276553,0.138277) (0.274549,0.137275) (0.272545,0.136273) (0.270541,0.135271) (0.268537,0.134269) (0.266533,0.133267) (0.264529,0.132265) (0.262525,0.131263) (0.260521,0.130261) (0.258517,0.129259) (0.256513,0.128257) (0.254509,0.127255) (0.252505,0.126253) (0.250501,0.125251) (0.248497,0.124248) (0.246493,0.123246) (0.244489,0.122244) (0.242485,0.121242) (0.240481,0.12024) (0.238477,0.119238) (0.236473,0.118236) (0.234469,0.117234) (0.232465,0.116232) (0.230461,0.11523) (0.228457,0.114228) (0.226453,0.113226) (0.224449,0.112224) (0.222445,0.111222) (0.220441,0.11022) (0.218437,0.109218) (0.216433,0.108216) (0.214429,0.107214) (0.212425,0.106212) (0.210421,0.10521) (0.208417,0.104208) (0.206413,0.103206) (0.204409,0.102204) (0.202405,0.101202) (0.200401,0.1002) (0.198397,0.0991984) (0.196393,0.0981964) (0.194389,0.0971944) (0.192385,0.0961924) (0.190381,0.0951904) (0.188377,0.0941884) (0.186373,0.0931864) (0.184369,0.0921844) (0.182365,0.0911824) (0.180361,0.0901804) (0.178357,0.0891784) (0.176353,0.0881764) (0.174349,0.0871743) (0.172345,0.0861723) (0.170341,0.0851703) (0.168337,0.0841683) (0.166333,0.0831663) (0.164329,0.0821643) (0.162325,0.0811623) (0.160321,0.0801603) (0.158317,0.0791583) (0.156313,0.0781563) (0.154309,0.0771543) (0.152305,0.0761523) (0.150301,0.0751503) (0.148297,0.0741483) (0.146293,0.0731463) (0.144289,0.0721443) (0.142285,0.0711423) (0.140281,0.0701403) (0.138277,0.0691383) (0.136273,0.0681363) (0.134269,0.0671343) (0.132265,0.0661323) (0.130261,0.0651303) (0.128257,0.0641283) (0.126253,0.0631263) (0.124248,0.0621242) (0.122244,0.0611222) (0.12024,0.0601202) (0.118236,0.0591182) (0.116232,0.0581162) (0.114228,0.0571142) (0.112224,0.0561122) (0.11022,0.0551102) (0.108216,0.0541082) (0.106212,0.0531062) (0.104208,0.0521042) (0.102204,0.0511022) (0.1002,0.0501002) (0.0981964,0.0490982) (0.0961924,0.0480962) (0.0941884,0.0470942) (0.0921844,0.0460922) (0.0901804,0.0450902) (0.0881764,0.0440882) (0.0861723,0.0430862) (0.0841683,0.0420842) (0.0821643,0.0410822) (0.0801603,0.0400802) (0.0781563,0.0390782) (0.0761523,0.0380762) (0.0741483,0.0370741) (0.0721443,0.0360721) (0.0701403,0.0350701) (0.0681363,0.0340681) (0.0661323,0.0330661) (0.0641283,0.0320641) (0.0621242,0.0310621) (0.0601202,0.0300601) (0.0581162,0.0290581) (0.0561122,0.0280561) (0.0541082,0.0270541) (0.0521042,0.0260521) (0.0501002,0.0250501) (0.0480962,0.0240481) (0.0460922,0.0230461) (0.0440882,0.0220441) (0.0420842,0.0210421) (0.0400802,0.0200401) (0.0380762,0.0190381) (0.0360721,0.0180361) (0.0340681,0.0170341) (0.0320641,0.0160321) (0.0300601,0.0150301) (0.0280561,0.0140281) (0.0260521,0.0130261) (0.0240481,0.012024) (0.0220441,0.011022) (0.0200401,0.01002) (0.0180361,0.00901804) (0.0160321,0.00801603) (0.0140281,0.00701403) (0.012024,0.00601202) (0.01002,0.00501002) (0.00801603,0.00400802) (0.00601202,0.00300601) (0.00400802,0.00200401) (0.00200401,0.001002) (0.0,0.0) };
  \addlegendentry{Skewed (PTHash \cite{pibiri2021pthash})};

 \addplot coordinates { (1.0,1.0) (0.997996,0.985546) (0.995992,0.97387) (0.993988,0.963243) (0.991984,0.953296) (0.98998,0.943856) (0.987976,0.934819) (0.985972,0.926118) (0.983968,0.917705) (0.981964,0.909542) (0.97996,0.901603) (0.977956,0.893864) (0.975952,0.886308) (0.973948,0.878919) (0.971944,0.871684) (0.96994,0.864593) (0.967936,0.857635) (0.965932,0.850802) (0.963928,0.844088) (0.961924,0.837485) (0.95992,0.830987) (0.957916,0.82459) (0.955912,0.818288) (0.953908,0.812077) (0.951904,0.805953) (0.9499,0.799913) (0.947896,0.793953) (0.945892,0.788071) (0.943888,0.782262) (0.941884,0.776525) (0.93988,0.770857) (0.937876,0.765256) (0.935872,0.75972) (0.933868,0.754246) (0.931864,0.748833) (0.92986,0.743479) (0.927856,0.738182) (0.925852,0.732941) (0.923848,0.727754) (0.921844,0.72262) (0.91984,0.717537) (0.917836,0.712504) (0.915832,0.70752) (0.913828,0.702584) (0.911824,0.697695) (0.90982,0.692851) (0.907816,0.688051) (0.905812,0.683295) (0.903808,0.678582) (0.901804,0.673911) (0.8998,0.66928) (0.897796,0.66469) (0.895792,0.660139) (0.893788,0.655626) (0.891784,0.651151) (0.88978,0.646713) (0.887776,0.642312) (0.885772,0.637947) (0.883768,0.633616) (0.881764,0.629321) (0.87976,0.625059) (0.877756,0.62083) (0.875752,0.616635) (0.873747,0.612472) (0.871743,0.60834) (0.869739,0.60424) (0.867735,0.600171) (0.865731,0.596132) (0.863727,0.592123) (0.861723,0.588143) (0.859719,0.584193) (0.857715,0.580271) (0.855711,0.576377) (0.853707,0.572512) (0.851703,0.568673) (0.849699,0.564862) (0.847695,0.561077) (0.845691,0.557319) (0.843687,0.553587) (0.841683,0.549881) (0.839679,0.5462) (0.837675,0.542544) (0.835671,0.538912) (0.833667,0.535306) (0.831663,0.531723) (0.829659,0.528164) (0.827655,0.524629) (0.825651,0.521117) (0.823647,0.517628) (0.821643,0.514162) (0.819639,0.510718) (0.817635,0.507297) (0.815631,0.503897) (0.813627,0.50052) (0.811623,0.497164) (0.809619,0.493829) (0.807615,0.490516) (0.805611,0.487223) (0.803607,0.483951) (0.801603,0.480699) (0.799599,0.477468) (0.797595,0.474256) (0.795591,0.471065) (0.793587,0.467893) (0.791583,0.464741) (0.789579,0.461608) (0.787575,0.458494) (0.785571,0.455398) (0.783567,0.452322) (0.781563,0.449264) (0.779559,0.446225) (0.777555,0.443204) (0.775551,0.4402) (0.773547,0.437215) (0.771543,0.434248) (0.769539,0.431298) (0.767535,0.428365) (0.765531,0.42545) (0.763527,0.422552) (0.761523,0.41967) (0.759519,0.416806) (0.757515,0.413959) (0.755511,0.411128) (0.753507,0.408313) (0.751503,0.405515) (0.749499,0.402732) (0.747495,0.399966) (0.745491,0.397216) (0.743487,0.394482) (0.741483,0.391763) (0.739479,0.389059) (0.737475,0.386372) (0.735471,0.383699) (0.733467,0.381042) (0.731463,0.378399) (0.729459,0.375772) (0.727455,0.37316) (0.725451,0.370562) (0.723447,0.367979) (0.721443,0.36541) (0.719439,0.362856) (0.717435,0.360316) (0.715431,0.35779) (0.713427,0.355279) (0.711423,0.352781) (0.709419,0.350298) (0.707415,0.347828) (0.705411,0.345372) (0.703407,0.342929) (0.701403,0.3405) (0.699399,0.338085) (0.697395,0.335683) (0.695391,0.333294) (0.693387,0.330918) (0.691383,0.328556) (0.689379,0.326206) (0.687375,0.32387) (0.685371,0.321546) (0.683367,0.319235) (0.681363,0.316937) (0.679359,0.314651) (0.677355,0.312378) (0.675351,0.310117) (0.673347,0.307869) (0.671343,0.305633) (0.669339,0.303409) (0.667335,0.301197) (0.665331,0.298997) (0.663327,0.29681) (0.661323,0.294634) (0.659319,0.29247) (0.657315,0.290318) (0.655311,0.288178) (0.653307,0.286049) (0.651303,0.283932) (0.649299,0.281827) (0.647295,0.279732) (0.645291,0.27765) (0.643287,0.275578) (0.641283,0.273518) (0.639279,0.271469) (0.637275,0.269431) (0.635271,0.267405) (0.633267,0.265389) (0.631263,0.263384) (0.629259,0.26139) (0.627255,0.259407) (0.625251,0.257435) (0.623246,0.255473) (0.621242,0.253522) (0.619238,0.251582) (0.617234,0.249652) (0.61523,0.247733) (0.613226,0.245824) (0.611222,0.243926) (0.609218,0.242038) (0.607214,0.24016) (0.60521,0.238292) (0.603206,0.236435) (0.601202,0.234587) (0.599198,0.23275) (0.597194,0.230923) (0.59519,0.229105) (0.593186,0.227298) (0.591182,0.225501) (0.589178,0.223713) (0.587174,0.221935) (0.58517,0.220167) (0.583166,0.218409) (0.581162,0.21666) (0.579158,0.21492) (0.577154,0.213191) (0.57515,0.211471) (0.573146,0.20976) (0.571142,0.208058) (0.569138,0.206366) (0.567134,0.204684) (0.56513,0.20301) (0.563126,0.201346) (0.561122,0.199691) (0.559118,0.198046) (0.557114,0.196409) (0.55511,0.194781) (0.553106,0.193163) (0.551102,0.191553) (0.549098,0.189952) (0.547094,0.188361) (0.54509,0.186778) (0.543086,0.185204) (0.541082,0.183638) (0.539078,0.182082) (0.537074,0.180534) (0.53507,0.178995) (0.533066,0.177464) (0.531062,0.175943) (0.529058,0.174429) (0.527054,0.172924) (0.52505,0.171428) (0.523046,0.16994) (0.521042,0.168461) (0.519038,0.16699) (0.517034,0.165527) (0.51503,0.164073) (0.513026,0.162627) (0.511022,0.161189) (0.509018,0.159759) (0.507014,0.158338) (0.50501,0.156924) (0.503006,0.155519) (0.501002,0.154122) (0.498998,0.152733) (0.496994,0.151352) (0.49499,0.149979) (0.492986,0.148614) (0.490982,0.147256) (0.488978,0.145907) (0.486974,0.144566) (0.48497,0.143232) (0.482966,0.141906) (0.480962,0.140588) (0.478958,0.139278) (0.476954,0.137975) (0.47495,0.13668) (0.472946,0.135393) (0.470942,0.134113) (0.468938,0.132841) (0.466934,0.131577) (0.46493,0.13032) (0.462926,0.12907) (0.460922,0.127828) (0.458918,0.126594) (0.456914,0.125367) (0.45491,0.124147) (0.452906,0.122935) (0.450902,0.12173) (0.448898,0.120532) (0.446894,0.119341) (0.44489,0.118158) (0.442886,0.116982) (0.440882,0.115814) (0.438878,0.114652) (0.436874,0.113498) (0.43487,0.11235) (0.432866,0.11121) (0.430862,0.110077) (0.428858,0.108951) (0.426854,0.107832) (0.42485,0.10672) (0.422846,0.105615) (0.420842,0.104517) (0.418838,0.103426) (0.416834,0.102342) (0.41483,0.101265) (0.412826,0.100194) (0.410822,0.0991308) (0.408818,0.098074) (0.406814,0.097024) (0.40481,0.0959808) (0.402806,0.0949444) (0.400802,0.0939146) (0.398798,0.0928916) (0.396794,0.0918753) (0.39479,0.0908656) (0.392786,0.0898625) (0.390782,0.0888661) (0.388778,0.0878762) (0.386774,0.0868929) (0.38477,0.0859162) (0.382766,0.084946) (0.380762,0.0839823) (0.378758,0.0830251) (0.376754,0.0820743) (0.374749,0.08113) (0.372745,0.0801922) (0.370741,0.0792607) (0.368737,0.0783356) (0.366733,0.0774169) (0.364729,0.0765045) (0.362725,0.0755984) (0.360721,0.0746986) (0.358717,0.0738051) (0.356713,0.0729179) (0.354709,0.0720369) (0.352705,0.0711622) (0.350701,0.0702936) (0.348697,0.0694313) (0.346693,0.0685751) (0.344689,0.067725) (0.342685,0.0668811) (0.340681,0.0660433) (0.338677,0.0652116) (0.336673,0.0643859) (0.334669,0.0635663) (0.332665,0.0627527) (0.330661,0.0619452) (0.328657,0.0611437) (0.326653,0.0603481) (0.324649,0.0595585) (0.322645,0.0587749) (0.320641,0.0579971) (0.318637,0.0572253) (0.316633,0.0564594) (0.314629,0.0556994) (0.312625,0.0549452) (0.310621,0.0541969) (0.308617,0.0534544) (0.306613,0.0527176) (0.304609,0.0519867) (0.302605,0.0512616) (0.300601,0.0505422) (0.298597,0.0498286) (0.296593,0.0491207) (0.294589,0.0484185) (0.292585,0.047722) (0.290581,0.0470312) (0.288577,0.046346) (0.286573,0.0456665) (0.284569,0.0449926) (0.282565,0.0443243) (0.280561,0.0436616) (0.278557,0.0430045) (0.276553,0.042353) (0.274549,0.041707) (0.272545,0.0410666) (0.270541,0.0404316) (0.268537,0.0398022) (0.266533,0.0391783) (0.264529,0.0385598) (0.262525,0.0379468) (0.260521,0.0373393) (0.258517,0.0367372) (0.256513,0.0361405) (0.254509,0.0355492) (0.252505,0.0349633) (0.250501,0.0343827) (0.248497,0.0338076) (0.246493,0.0332377) (0.244489,0.0326732) (0.242485,0.032114) (0.240481,0.0315601) (0.238477,0.0310115) (0.236473,0.0304682) (0.234469,0.0299301) (0.232465,0.0293973) (0.230461,0.0288697) (0.228457,0.0283474) (0.226453,0.0278302) (0.224449,0.0273182) (0.222445,0.0268114) (0.220441,0.0263098) (0.218437,0.0258133) (0.216433,0.025322) (0.214429,0.0248358) (0.212425,0.0243547) (0.210421,0.0238787) (0.208417,0.0234078) (0.206413,0.0229419) (0.204409,0.0224811) (0.202405,0.0220254) (0.200401,0.0215747) (0.198397,0.021129) (0.196393,0.0206884) (0.194389,0.0202527) (0.192385,0.019822) (0.190381,0.0193963) (0.188377,0.0189755) (0.186373,0.0185597) (0.184369,0.0181489) (0.182365,0.0177429) (0.180361,0.0173419) (0.178357,0.0169458) (0.176353,0.0165545) (0.174349,0.0161682) (0.172345,0.0157866) (0.170341,0.01541) (0.168337,0.0150382) (0.166333,0.0146712) (0.164329,0.014309) (0.162325,0.0139517) (0.160321,0.0135991) (0.158317,0.0132513) (0.156313,0.0129083) (0.154309,0.0125701) (0.152305,0.0122366) (0.150301,0.0119078) (0.148297,0.0115838) (0.146293,0.0112645) (0.144289,0.0109498) (0.142285,0.0106399) (0.140281,0.0103347) (0.138277,0.0100341) (0.136273,0.0097382) (0.134269,0.00944694) (0.132265,0.00916031) (0.130261,0.00887832) (0.128257,0.00860095) (0.126253,0.00832818) (0.124248,0.00806) (0.122244,0.00779642) (0.12024,0.00753741) (0.118236,0.00728296) (0.116232,0.00703307) (0.114228,0.00678772) (0.112224,0.00654691) (0.11022,0.00631062) (0.108216,0.00607884) (0.106212,0.00585157) (0.104208,0.00562879) (0.102204,0.0054105) (0.1002,0.00519667) (0.0981964,0.00498731) (0.0961924,0.00478241) (0.0941884,0.00458194) (0.0921844,0.00438591) (0.0901804,0.00419431) (0.0881764,0.00400712) (0.0861723,0.00382433) (0.0841683,0.00364594) (0.0821643,0.00347193) (0.0801603,0.0033023) (0.0781563,0.00313704) (0.0761523,0.00297613) (0.0741483,0.00281957) (0.0721443,0.00266734) (0.0701403,0.00251945) (0.0681363,0.00237587) (0.0661323,0.0022366) (0.0641283,0.00210164) (0.0621242,0.00197096) (0.0601202,0.00184457) (0.0581162,0.00172245) (0.0561122,0.00160459) (0.0541082,0.00149099) (0.0521042,0.00138163) (0.0501002,0.00127652) (0.0480962,0.00117562) (0.0460922,0.00107895) (0.0440882,0.00098649) (0.0420842,0.000898229) (0.0400802,0.000814161) (0.0380762,0.000734276) (0.0360721,0.000658567) (0.0340681,0.000587024) (0.0320641,0.000519638) (0.0300601,0.000456402) (0.0280561,0.000397306) (0.0260521,0.000342342) (0.0240481,0.000291502) (0.0220441,0.000244776) (0.0200401,0.000202157) (0.0180361,0.000163637) (0.0160321,0.000129206) (0.0140281,9.88565e-05) (0.012024,7.25804e-05) (0.01002,5.03691e-05) (0.00801603,3.22146e-05) (0.00601202,1.81085e-05) (0.00400802,8.04285e-06) (0.00200401,2.00937e-06) (0.0,0.0) };
    \addlegendentry{Optimized (PHOBIC)};
 
 \legend{};
\end{axis}

\end{tikzpicture}

%% file: fig/bucketfunction.tex
\centering
\begin{tikzpicture}

\begin{axis}[
	xlabel={Normalized Bucket},
	ylabel={Relative Expected Size},
	width=5.3cm,	
	height=2.8cm,	
	ymax=4,
	ymin=0,
	legend style={at={(1.0,0.6)},anchor=south east},
	mark options={mark indices=1},
	]
	
  \addplot coordinates { (1.0,1.0) (0.0,1.0) };
  \addlegendentry{Uniform (CHD \cite{belazzougui2009hash})};

  \addplot coordinates { (1.0,0.571429) (0.3,0.571429) (0.3,2.0) (0.0,2.0) };
  \addlegendentry{Skewed (PTHash \cite{pibiri2021pthash})};

 \addplot coordinates { (1.0,0.0) (0.992082,0.144262) (0.988729,0.153589) (0.985546,0.160821) (0.982488,0.166873) (0.979529,0.17215) (0.976653,0.176871) (0.973848,0.181171) (0.971104,0.185139) (0.968415,0.188838) (0.965776,0.192314) (0.963181,0.195601) (0.960626,0.198725) (0.95811,0.201708) (0.955628,0.204567) (0.953179,0.207317) (0.95076,0.209968) (0.94837,0.212532) (0.946007,0.215015) (0.943669,0.217426) (0.941356,0.21977) (0.939066,0.222053) (0.936797,0.22428) (0.93455,0.226455) (0.932323,0.228582) (0.930115,0.230664) (0.927925,0.232704) (0.925754,0.234704) (0.923599,0.236668) (0.921461,0.238597) (0.919339,0.240494) (0.917233,0.242359) (0.915141,0.244196) (0.913064,0.246005) (0.911001,0.247788) (0.908952,0.249546) (0.906915,0.25128) (0.904892,0.252992) (0.902881,0.254682) (0.900882,0.256352) (0.898896,0.258003) (0.89692,0.259634) (0.894956,0.261248) (0.893003,0.262844) (0.891061,0.264424) (0.88913,0.265988) (0.887208,0.267537) (0.885297,0.269072) (0.883395,0.270592) (0.881503,0.272099) (0.879621,0.273592) (0.877747,0.275073) (0.875883,0.276542) (0.874027,0.277999) (0.872181,0.279445) (0.870342,0.28088) (0.868513,0.282304) (0.866691,0.283718) (0.864877,0.285122) (0.863072,0.286517) (0.861274,0.287902) (0.859483,0.289278) (0.857701,0.290645) (0.855925,0.292004) (0.854157,0.293355) (0.852396,0.294697) (0.850642,0.296032) (0.848895,0.29736) (0.847155,0.29868) (0.845422,0.299993) (0.843695,0.301299) (0.841975,0.302598) (0.840261,0.303891) (0.838553,0.305177) (0.836852,0.306457) (0.835157,0.307732) (0.833468,0.309) (0.831785,0.310262) (0.830108,0.311519) (0.828437,0.312771) (0.826771,0.314017) (0.825111,0.315258) (0.823457,0.316494) (0.821808,0.317726) (0.820165,0.318952) (0.818527,0.320174) (0.816895,0.321391) (0.815267,0.322604) (0.813645,0.323812) (0.812028,0.325016) (0.810416,0.326216) (0.80881,0.327412) (0.807208,0.328604) (0.805611,0.329793) (0.804019,0.330977) (0.802432,0.332158) (0.800849,0.333335) (0.799271,0.334509) (0.797698,0.335679) (0.796129,0.336847) (0.794565,0.33801) (0.793006,0.339171) (0.791451,0.340328) (0.7899,0.341483) (0.788354,0.342634) (0.786812,0.343783) (0.785274,0.344929) (0.783741,0.346071) (0.782211,0.347212) (0.780686,0.348349) (0.779165,0.349484) (0.777648,0.350617) (0.776135,0.351747) (0.774626,0.352874) (0.773121,0.354) (0.77162,0.355123) (0.770123,0.356243) (0.76863,0.357362) (0.76714,0.358478) (0.765655,0.359592) (0.764173,0.360704) (0.762694,0.361814) (0.76122,0.362923) (0.759749,0.364029) (0.758282,0.365133) (0.756818,0.366236) (0.755358,0.367336) (0.753901,0.368435) (0.752448,0.369533) (0.750999,0.370628) (0.749552,0.371722) (0.74811,0.372815) (0.74667,0.373906) (0.745234,0.374995) (0.743802,0.376083) (0.742372,0.377169) (0.740946,0.378254) (0.739523,0.379338) (0.738104,0.38042) (0.736687,0.381501) (0.735274,0.382581) (0.733864,0.383659) (0.732457,0.384737) (0.731053,0.385813) (0.729653,0.386888) (0.728255,0.387962) (0.72686,0.389034) (0.725469,0.390106) (0.72408,0.391177) (0.722695,0.392246) (0.721312,0.393315) (0.719932,0.394383) (0.718555,0.39545) (0.717181,0.396516) (0.71581,0.397581) (0.714442,0.398645) (0.713077,0.399708) (0.711714,0.400771) (0.710355,0.401833) (0.708998,0.402894) (0.707644,0.403954) (0.706292,0.405014) (0.704943,0.406073) (0.703597,0.407132) (0.702254,0.408189) (0.700914,0.409247) (0.699576,0.410303) (0.69824,0.411359) (0.696907,0.412415) (0.695577,0.41347) (0.69425,0.414525) (0.692925,0.415579) (0.691602,0.416632) (0.690283,0.417685) (0.688965,0.418738) (0.68765,0.419791) (0.686338,0.420843) (0.685028,0.421894) (0.683721,0.422946) (0.682416,0.423997) (0.681113,0.425047) (0.679813,0.426098) (0.678516,0.427148) (0.67722,0.428198) (0.675928,0.429248) (0.674637,0.430297) (0.673349,0.431347) (0.672063,0.432396) (0.670779,0.433445) (0.669498,0.434494) (0.668219,0.435543) (0.666943,0.436592) (0.665668,0.43764) (0.664396,0.438689) (0.663126,0.439737) (0.661859,0.440786) (0.660593,0.441834) (0.65933,0.442883) (0.658069,0.443931) (0.656811,0.44498) (0.655554,0.446028) (0.654299,0.447077) (0.653047,0.448126) (0.651797,0.449174) (0.650549,0.450223) (0.649303,0.451272) (0.648059,0.452321) (0.646818,0.453371) (0.645578,0.45442) (0.644341,0.45547) (0.643105,0.45652) (0.641872,0.45757) (0.64064,0.45862) (0.639411,0.45967) (0.638184,0.460721) (0.636958,0.461772) (0.635735,0.462823) (0.634514,0.463875) (0.633295,0.464927) (0.632077,0.465979) (0.630862,0.467032) (0.629649,0.468084) (0.628437,0.469138) (0.627228,0.470191) (0.62602,0.471245) (0.624814,0.472299) (0.623611,0.473354) (0.622409,0.474409) (0.621209,0.475465) (0.620011,0.476521) (0.618815,0.477578) (0.61762,0.478635) (0.616428,0.479692) (0.615237,0.48075) (0.614048,0.481808) (0.612861,0.482867) (0.611676,0.483927) (0.610493,0.484987) (0.609311,0.486047) (0.608132,0.487109) (0.606954,0.48817) (0.605778,0.489233) (0.604603,0.490296) (0.603431,0.491359) (0.60226,0.492423) (0.601091,0.493488) (0.599924,0.494553) (0.598758,0.49562) (0.597594,0.496686) (0.596432,0.497754) (0.595272,0.498822) (0.594113,0.499891) (0.592956,0.50096) (0.591801,0.502031) (0.590647,0.503102) (0.589495,0.504173) (0.588345,0.505246) (0.587196,0.506319) (0.586049,0.507393) (0.584904,0.508468) (0.58376,0.509544) (0.582618,0.51062) (0.581478,0.511698) (0.580339,0.512776) (0.579202,0.513855) (0.578067,0.514935) (0.576933,0.516015) (0.5758,0.517097) (0.57467,0.518179) (0.573541,0.519263) (0.572413,0.520347) (0.571287,0.521432) (0.570163,0.522518) (0.56904,0.523605) (0.567918,0.524694) (0.566799,0.525783) (0.565681,0.526873) (0.564564,0.527964) (0.563449,0.529055) (0.562335,0.530148) (0.561223,0.531242) (0.560113,0.532337) (0.559004,0.533434) (0.557896,0.534531) (0.55679,0.535629) (0.555686,0.536728) (0.554582,0.537828) (0.553481,0.53893) (0.552381,0.540032) (0.551282,0.541136) (0.550185,0.542241) (0.549089,0.543347) (0.547995,0.544454) (0.546902,0.545562) (0.545811,0.546671) (0.544721,0.547782) (0.543633,0.548893) (0.542546,0.550006) (0.54146,0.55112) (0.540376,0.552235) (0.539293,0.553352) (0.538212,0.55447) (0.537132,0.555589) (0.536053,0.556709) (0.534976,0.557831) (0.5339,0.558953) (0.532826,0.560077) (0.531753,0.561203) (0.530681,0.562329) (0.529611,0.563457) (0.528542,0.564587) (0.527474,0.565717) (0.526408,0.56685) (0.525343,0.567983) (0.52428,0.569118) (0.523218,0.570254) (0.522157,0.571391) (0.521097,0.57253) (0.520039,0.573671) (0.518982,0.574812) (0.517927,0.575956) (0.516873,0.5771) (0.51582,0.578246) (0.514768,0.579394) (0.513718,0.580543) (0.512669,0.581693) (0.511621,0.582845) (0.510575,0.583999) (0.50953,0.585154) (0.508486,0.58631) (0.507444,0.587469) (0.506402,0.588628) (0.505363,0.589789) (0.504324,0.590952) (0.503286,0.592116) (0.50225,0.593282) (0.501215,0.59445) (0.500182,0.595619) (0.499149,0.59679) (0.498118,0.597962) (0.497088,0.599136) (0.496059,0.600312) (0.495032,0.601489) (0.494006,0.602669) (0.492981,0.603849) (0.491957,0.605032) (0.490934,0.606216) (0.489913,0.607402) (0.488893,0.608589) (0.487874,0.609779) (0.486856,0.61097) (0.48584,0.612163) (0.484824,0.613357) (0.48381,0.614554) (0.482797,0.615752) (0.481785,0.616952) (0.480775,0.618154) (0.479765,0.619357) (0.478757,0.620563) (0.47775,0.62177) (0.476744,0.62298) (0.475739,0.624191) (0.474736,0.625404) (0.473733,0.626619) (0.472732,0.627835) (0.471732,0.629054) (0.470733,0.630275) (0.469735,0.631497) (0.468739,0.632722) (0.467743,0.633949) (0.466749,0.635177) (0.465755,0.636408) (0.464763,0.63764) (0.463772,0.638875) (0.462782,0.640111) (0.461794,0.64135) (0.460806,0.642591) (0.459819,0.643833) (0.458834,0.645078) (0.45785,0.646325) (0.456866,0.647574) (0.455884,0.648825) (0.454903,0.650078) (0.453923,0.651334) (0.452945,0.652591) (0.451967,0.653851) (0.45099,0.655113) (0.450015,0.656377) (0.44904,0.657643) (0.448067,0.658911) (0.447094,0.660182) (0.446123,0.661455) (0.445153,0.66273) (0.444184,0.664007) (0.443216,0.665287) (0.442249,0.666569) (0.441283,0.667853) (0.440318,0.66914) (0.439354,0.670428) (0.438391,0.67172) (0.43743,0.673013) (0.436469,0.674309) (0.435509,0.675607) (0.434551,0.676908) (0.433593,0.678211) (0.432637,0.679517) (0.431681,0.680825) (0.430727,0.682135) (0.429773,0.683448) (0.428821,0.684763) (0.42787,0.686081) (0.426919,0.687401) (0.42597,0.688724) (0.425021,0.690049) (0.424074,0.691377) (0.423128,0.692708) (0.422183,0.694041) (0.421238,0.695376) (0.420295,0.696714) (0.419353,0.698055) (0.418411,0.699399) (0.417471,0.700745) (0.416532,0.702093) (0.415593,0.703445) (0.414656,0.704799) (0.41372,0.706155) (0.412784,0.707515) (0.41185,0.708877) (0.410917,0.710242) (0.409984,0.711609) (0.409053,0.71298) (0.408122,0.714353) (0.407193,0.715729) (0.406264,0.717108) (0.405337,0.71849) (0.40441,0.719874) (0.403484,0.721261) (0.40256,0.722652) (0.401636,0.724045) (0.400713,0.725441) (0.399791,0.72684) (0.39887,0.728242) (0.39795,0.729647) (0.397031,0.731054) (0.396113,0.732465) (0.395196,0.733879) (0.39428,0.735296) (0.393365,0.736716) (0.39245,0.738139) (0.391537,0.739565) (0.390625,0.740994) (0.389713,0.742426) (0.388803,0.743861) (0.387893,0.7453) (0.386984,0.746741) (0.386076,0.748186) (0.385169,0.749634) (0.384263,0.751085) (0.383358,0.752539) (0.382454,0.753997) (0.381551,0.755458) (0.380649,0.756922) (0.379747,0.758389) (0.378847,0.75986) (0.377947,0.761334) (0.377048,0.762811) (0.376151,0.764292) (0.375254,0.765776) (0.374358,0.767264) (0.373462,0.768755) (0.372568,0.770249) (0.371675,0.771747) (0.370782,0.773249) (0.369891,0.774753) (0.369,0.776262) (0.36811,0.777774) (0.367221,0.779289) (0.366333,0.780808) (0.365446,0.782331) (0.36456,0.783857) (0.363674,0.785387) (0.36279,0.78692) (0.361906,0.788457) (0.361023,0.789998) (0.360141,0.791543) (0.35926,0.793091) (0.35838,0.794643) (0.357501,0.796199) (0.356622,0.797758) (0.355745,0.799322) (0.354868,0.800889) (0.353992,0.80246) (0.353117,0.804035) (0.352243,0.805614) (0.351369,0.807197) (0.350497,0.808783) (0.349625,0.810374) (0.348754,0.811969) (0.347884,0.813567) (0.347015,0.81517) (0.346147,0.816777) (0.345279,0.818388) (0.344412,0.820003) (0.343547,0.821622) (0.342682,0.823245) (0.341818,0.824872) (0.340954,0.826504) (0.340092,0.828139) (0.33923,0.829779) (0.338369,0.831424) (0.337509,0.833072) (0.33665,0.834725) (0.335792,0.836382) (0.334934,0.838044) (0.334077,0.83971) (0.333221,0.84138) (0.332366,0.843055) (0.331512,0.844734) (0.330659,0.846417) (0.329806,0.848106) (0.328954,0.849798) (0.328103,0.851495) (0.327253,0.853197) (0.326403,0.854904) (0.325555,0.856615) (0.324707,0.858331) (0.32386,0.860051) (0.323014,0.861776) (0.322168,0.863506) (0.321324,0.865241) (0.32048,0.86698) (0.319637,0.868724) (0.318794,0.870474) (0.317953,0.872228) (0.317112,0.873986) (0.316272,0.87575) (0.315433,0.877519) (0.314595,0.879293) (0.313758,0.881072) (0.312921,0.882856) (0.312085,0.884645) (0.31125,0.886439) (0.310415,0.888238) (0.309582,0.890042) (0.308749,0.891852) (0.307917,0.893667) (0.307085,0.895487) (0.306255,0.897312) (0.305425,0.899143) (0.304596,0.900979) (0.303768,0.90282) (0.30294,0.904667) (0.302114,0.906519) (0.301288,0.908377) (0.300463,0.91024) (0.299638,0.912109) (0.298815,0.913984) (0.297992,0.915864) (0.29717,0.91775) (0.296348,0.919641) (0.295528,0.921538) (0.294708,0.923441) (0.293889,0.92535) (0.29307,0.927264) (0.292253,0.929184) (0.291436,0.931111) (0.29062,0.933043) (0.289804,0.934981) (0.28899,0.936925) (0.288176,0.938875) (0.287363,0.940831) (0.28655,0.942794) (0.285739,0.944762) (0.284928,0.946737) (0.284118,0.948718) (0.283308,0.950705) (0.2825,0.952698) (0.281692,0.954698) (0.280884,0.956704) (0.280078,0.958717) (0.279272,0.960736) (0.278467,0.962762) (0.277663,0.964794) (0.276859,0.966833) (0.276057,0.968878) (0.275255,0.97093) (0.274453,0.972989) (0.273653,0.975054) (0.272853,0.977126) (0.272053,0.979205) (0.271255,0.981291) (0.270457,0.983384) (0.26966,0.985484) (0.268864,0.987591) (0.268068,0.989705) (0.267273,0.991826) (0.266479,0.993954) (0.265686,0.99609) (0.264893,0.998232) (0.264101,1.00038) (0.26331,1.00254) (0.262519,1.0047) (0.261729,1.00688) (0.26094,1.00906) (0.260152,1.01124) (0.259364,1.01344) (0.258577,1.01564) (0.257791,1.01785) (0.257005,1.02007) (0.25622,1.02229) (0.255436,1.02453) (0.254652,1.02677) (0.253869,1.02902) (0.253087,1.03127) (0.252306,1.03354) (0.251525,1.03581) (0.250745,1.03809) (0.249966,1.04038) (0.249187,1.04268) (0.248409,1.04499) (0.247632,1.0473) (0.246855,1.04963) (0.24608,1.05196) (0.245304,1.0543) (0.24453,1.05665) (0.243756,1.059) (0.242983,1.06137) (0.24221,1.06374) (0.241439,1.06613) (0.240668,1.06852) (0.239897,1.07092) (0.239128,1.07333) (0.238359,1.07575) (0.23759,1.07818) (0.236823,1.08062) (0.236056,1.08307) (0.235289,1.08552) (0.234524,1.08799) (0.233759,1.09047) (0.232995,1.09295) (0.232231,1.09545) (0.231468,1.09795) (0.230706,1.10046) (0.229944,1.10299) (0.229184,1.10552) (0.228423,1.10807) (0.227664,1.11062) (0.226905,1.11318) (0.226147,1.11576) (0.225389,1.11834) (0.224632,1.12094) (0.223876,1.12354) (0.223121,1.12616) (0.222366,1.12879) (0.221612,1.13142) (0.220858,1.13407) (0.220105,1.13673) (0.219353,1.1394) (0.218602,1.14208) (0.217851,1.14477) (0.217101,1.14747) (0.216351,1.15019) (0.215602,1.15291) (0.214854,1.15565) (0.214106,1.1584) (0.213359,1.16116) (0.212613,1.16393) (0.211868,1.16671) (0.211123,1.16951) (0.210378,1.17231) (0.209635,1.17513) (0.208892,1.17796) (0.208149,1.18081) (0.207408,1.18366) (0.206667,1.18653) (0.205926,1.18941) (0.205187,1.1923) (0.204448,1.19521) (0.203709,1.19813) (0.202971,1.20106) (0.202234,1.204) (0.201498,1.20696) (0.200762,1.20993) (0.200027,1.21291) (0.199292,1.21591) (0.198558,1.21892) (0.197825,1.22195) (0.197093,1.22498) (0.196361,1.22804) (0.195629,1.2311) (0.194899,1.23418) (0.194169,1.23728) (0.193439,1.24039) (0.19271,1.24351) (0.191982,1.24665) (0.191255,1.2498) (0.190528,1.25297) (0.189802,1.25615) (0.189076,1.25934) (0.188351,1.26256) (0.187627,1.26578) (0.186903,1.26903) (0.18618,1.27229) (0.185458,1.27556) (0.184736,1.27885) (0.184015,1.28216) (0.183295,1.28548) (0.182575,1.28882) (0.181856,1.29217) (0.181137,1.29554) (0.180419,1.29893) (0.179702,1.30234) (0.178985,1.30576) (0.178269,1.3092) (0.177554,1.31265) (0.176839,1.31613) (0.176125,1.31962) (0.175411,1.32313) (0.174698,1.32665) (0.173986,1.3302) (0.173274,1.33376) (0.172563,1.33734) (0.171853,1.34094) (0.171143,1.34456) (0.170434,1.3482) (0.169725,1.35186) (0.169017,1.35553) (0.16831,1.35923) (0.167603,1.36294) (0.166897,1.36668) (0.166192,1.37043) (0.165487,1.37421) (0.164783,1.378) (0.164079,1.38182) (0.163376,1.38565) (0.162674,1.38951) (0.161972,1.39339) (0.161271,1.39729) (0.160571,1.40121) (0.159871,1.40515) (0.159172,1.40912) (0.158473,1.41311) (0.157775,1.41712) (0.157078,1.42115) (0.156381,1.4252) (0.155685,1.42928) (0.154989,1.43338) (0.154294,1.43751) (0.1536,1.44166) (0.152906,1.44583) (0.152213,1.45003) (0.151521,1.45425) (0.150829,1.4585) (0.150138,1.46277) (0.149447,1.46707) (0.148757,1.47139) (0.148068,1.47574) (0.147379,1.48011) (0.146691,1.48451) (0.146003,1.48894) (0.145316,1.49339) (0.14463,1.49787) (0.143944,1.50238) (0.143259,1.50692) (0.142574,1.51148) (0.141891,1.51607) (0.141207,1.5207) (0.140525,1.52535) (0.139843,1.53003) (0.139161,1.53473) (0.13848,1.53947) (0.1378,1.54424) (0.13712,1.54904) (0.136441,1.55387) (0.135763,1.55874) (0.135085,1.56363) (0.134408,1.56856) (0.133732,1.57351) (0.133056,1.5785) (0.13238,1.58353) (0.131705,1.58859) (0.131031,1.59368) (0.130358,1.5988) (0.129685,1.60396) (0.129012,1.60916) (0.128341,1.61439) (0.12767,1.61965) (0.126999,1.62496) (0.126329,1.6303) (0.12566,1.63567) (0.124991,1.64109) (0.124323,1.64654) (0.123656,1.65203) (0.122989,1.65756) (0.122323,1.66313) (0.121657,1.66874) (0.120992,1.67439) (0.120327,1.68008) (0.119664,1.68581) (0.119,1.69159) (0.118338,1.6974) (0.117676,1.70326) (0.117014,1.70917) (0.116353,1.71512) (0.115693,1.72111) (0.115034,1.72715) (0.114375,1.73323) (0.113716,1.73936) (0.113058,1.74554) (0.112401,1.75176) (0.111745,1.75804) (0.111089,1.76436) (0.110433,1.77073) (0.109779,1.77715) (0.109124,1.78363) (0.108471,1.79015) (0.107818,1.79673) (0.107166,1.80336) (0.106514,1.81005) (0.105863,1.81679) (0.105212,1.82358) (0.104562,1.83043) (0.103913,1.83734) (0.103264,1.84431) (0.102616,1.85133) (0.101969,1.85842) (0.101322,1.86556) (0.100676,1.87277) (0.10003,1.88004) (0.099385,1.88737) (0.0987405,1.89476) (0.0980967,1.90222) (0.0974535,1.90975) (0.0968109,1.91734) (0.0961689,1.925) (0.0955275,1.93274) (0.0948867,1.94054) (0.0942465,1.94841) (0.0936069,1.95635) (0.0929679,1.96437) (0.0923296,1.97246) (0.0916918,1.98063) (0.0910547,1.98887) (0.0904182,1.99719) (0.0897823,2.0056) (0.089147,2.01408) (0.0885123,2.02264) (0.0878782,2.03129) (0.0872448,2.04003) (0.0866119,2.04885) (0.0859797,2.05775) (0.0853481,2.06675) (0.0847171,2.07584) (0.0840867,2.08502) (0.0834569,2.09429) (0.0828277,2.10366) (0.0821992,2.11313) (0.0815713,2.12269) (0.080944,2.13236) (0.0803173,2.14212) (0.0796912,2.152) (0.0790657,2.16198) (0.0784409,2.17206) (0.0778167,2.18226) (0.0771931,2.19257) (0.0765701,2.20299) (0.0759478,2.21353) (0.075326,2.22419) (0.0747049,2.23497) (0.0740845,2.24588) (0.0734646,2.25691) (0.0728454,2.26806) (0.0722268,2.27935) (0.0716088,2.29077) (0.0709914,2.30233) (0.0703747,2.31403) (0.0697586,2.32587) (0.0691431,2.33785) (0.0685283,2.34998) (0.0679141,2.36226) (0.0673005,2.3747) (0.0666876,2.38729) (0.0660753,2.40004) (0.0654636,2.41296) (0.0648526,2.42604) (0.0642422,2.4393) (0.0636324,2.45273) (0.0630233,2.46634) (0.0624148,2.48013) (0.0618069,2.49411) (0.0611997,2.50828) (0.0605932,2.52265) (0.0599872,2.53722) (0.059382,2.55199) (0.0587773,2.56697) (0.0581734,2.58217) (0.05757,2.59759) (0.0569673,2.61324) (0.0563653,2.62912) (0.0557639,2.64523) (0.0551632,2.66159) (0.0545631,2.6782) (0.0539637,2.69507) (0.053365,2.7122) (0.0527669,2.7296) (0.0521694,2.74727) (0.0515727,2.76523) (0.0509765,2.78349) (0.0503811,2.80204) (0.0497863,2.82091) (0.0491922,2.84009) (0.0485988,2.8596) (0.048006,2.87944) (0.0474139,2.89964) (0.0468225,2.92019) (0.0462317,2.9411) (0.0456417,2.9624) (0.0450523,2.98408) (0.0444636,3.00617) (0.0438756,3.02867) (0.0432883,3.0516) (0.0427016,3.07497) (0.0421157,3.0988) (0.0415304,3.12311) (0.0409459,3.14789) (0.040362,3.17319) (0.0397789,3.199) (0.0391964,3.22536) (0.0386147,3.25227) (0.0380337,3.27977) (0.0374533,3.30786) (0.0368737,3.33658) (0.0362949,3.36594) (0.0357167,3.39598) (0.0351392,3.42671) (0.0345625,3.45817) (0.0339865,3.49038) (0.0334113,3.52338) (0.0328368,3.55719) (0.032263,3.59186) (0.03169,3.62741) (0.0311177,3.6639) (0.0305462,3.70135) (0.0299754,3.73981) (0.0294054,3.77933) (0.0288361,3.81996) (0.0282676,3.86175) (0.0276999,3.90476) (0.027133,3.94905) (0.0265668,3.99469) (0.0260014,4.04173) (0.0254369,4.09027) (0.0248731,4.14037) (0.0243101,4.19213) (0.023748,4.24564) (0.0231866,4.301) (0.0226261,4.35832) (0.0220664,4.41771) (0.0215075,4.47932) (0.0209495,4.54327) (0.0203923,4.60972) (0.019836,4.67885) (0.0192805,4.75083) (0.0187259,4.82586) (0.0181722,4.90417) (0.0176194,4.98602) (0.0170674,5.07166) (0.0165164,5.1614) (0.0159663,5.2556) (0.0154172,5.35461) (0.0148689,5.45888) (0.0143217,5.56888) (0.0137754,5.68517) (0.01323,5.80835) (0.0126857,5.93916) (0.0121424,6.0784) (0.0116001,6.22704) (0.0110588,6.38617) (0.0105187,6.55709) (0.00997955,6.74135) (0.00944155,6.94078) (0.00890469,7.15757) (0.00836897,7.39443) (0.00783445,7.65464) (0.00730113,7.94234) (0.00676907,8.26274) (0.0062383,8.62258) (0.00570886,9.03069) (0.00518081,9.499) (0.0046542,10.044) (0.00412909,10.6893) (0.00360557,11.4702) (0.00308373,12.4428) (0.0025637,13.7018) (0.00204562,15.4264) (0.00152971,18.0127) (0.00101627,22.6397) (0.000505864,44.1922) (0.0,62.5437) };
    \addlegendentry{Optimized (PHOBIC)};

\end{axis}

\end{tikzpicture}

%% file: sections/achieving_locality.tex
Any PHF construction can trivially be parallelized by hashing the keys into subsets of expected equal size and building a PHF for each subset in parallel. We refer to those subsets as \emph{partitions}.
The various PHFs are then logically ``concatenated'' into a single PHF taking the prefix sum of the partition sizes. The respective offsets have to be looked up when querying a key, imposing some query time overhead. Partitioning is widely used for a variety of construction techniques. It was also used by PTHash in the PTHash-HEM variant \cite{pibiri2024parallel}. In this paper, we use partitions that are several magnitudes smaller than the ones used in PTHash-HEM. In itself, reducing the partition size results in only marginal construction time improvements.
However, small partitions enable a new, more efficient encoding scheme which we introduce in \cref{ss:multi}.
Additionally, they enable a fast GPU parallelization which we describe in \cref{ss:gpuimpl}.

\myparagraph{Encoding the offsets of the partitions}
The offsets and sizes of the many partitions require a non-negligible amount of space.
We mitigate this without incurring too much query overhead by storing sizes only implicitly as the difference of two subsequent offsets. Offsets are stored as the difference to their expectation.%

\myparagraph{Hash Function}
The original PTHash hash function works by XOR-ing the hash value of the key with the hash value of the seed.
This reduces hash function evaluation to a simple XOR operation, which improves performance in practice.
Although the technique works well on large-enough partition sizes, it might fail for small sizes because of correlations in the hash values.
Given the small partition sizes we use here, we have to rely on a different technique.
We use the seed value $p$ to store two numbers, namely $p = s \cdot m + d$, where $m$ is the actual partition size and $d \in [m]$ is an additive displacement.
The position of a key $x$ is $(h(x,s) + d)\bmod m$. While searching for a seed, we calculate $h(x,s)$ for all keys of the bucket and then only have to increment the values to obtain the positions for the next seeds. If no position is found within $d \in [m]$ we continue by incrementing $s$, re-calculating $h(x,s)$ and setting $d=0$.
At query time, we have to calculate the position of a key $x$ using seed $p$.
Note that we have $d = p \bmod m$,
so we can compute the position of the key $x$ as $(h(x,\lfloor p/m \rfloor) + p)\bmod m$.

\subsection{Interleaved Coding of Seeds}
\label{ss:multi}

Once the search has finished, the seeds found for each bucket have to be stored in some
compressed manner. Ideally, the seeds should be encoded such that they require little space and are quickly accessible during querying.
We mainly use \emph{Compact} and \emph{Golomb-Rice} encoding as building blocks for our new technique. Compact encoding is also used in the original PTHash implementation.
In Compact encoding, all values are stored consecutively by concatenating their binary representation. All values use the same bit length, allowing for quick access. The bit length is chosen such that the highest seed can be accommodated. 
Golomb-Rice \cite{golomb1966run,rice1979some} encoding stores the $b$ least significant bits of each seed using compact encoding. The most significant bits are stored in unary representation. A selection structure enables access to the unary part of the seeds in constant time. We apply the formula by Kiely \cite{kiely2004selecting} to select $b$.
A straightforward approach would be to encode all seeds using a single encoder.
However, the seeds do not follow the same statistical distributions across
different buckets, hence using the same encoder for all buckets is suboptimal.
It is instead beneficial to group seeds which follow the same distribution and encode them using the same encoder.
PTHash does this only partially by using two encoders -- one for each expected bucket size (the so-called ``front-back'' compression~\cite{pibiri2021pthash}).

\begin{figure}[t]
	\centering
	\includegraphics[scale=0.9]{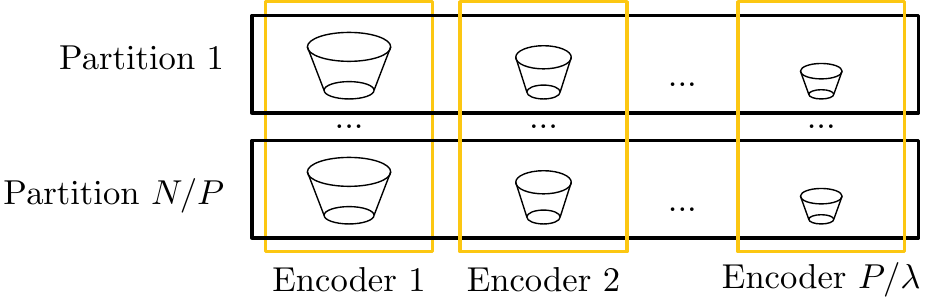}
	\caption{Interleaved coding. Encoder $i$ stores the seed of bucket $i$ from all partitions.}
	\label{fig:orthoillustration}
\end{figure}

We now introduce our new technique.
For each partition we hash to the same number of buckets $B = P / \lambda$, based on the average partition size $P$ and average bucket size $\lambda$.
The $i$-th bucket of a partition has the same expected size and the corresponding seed follows the same statistical distribution as the $i$-th bucket of any other partition. Although the idea of our optimized bucket assignment function is to give all buckets the same seed distribution, this is not achievable in practice. At least one reason for this are the discrete bucket sizes. This results in discrete jumps in the probability that a seed is found when processing one bucket after another.
In \emph{interleaved coding} we therefore employ $B$ encoders and the $i$-th encoder stores the seeds of the $i$-th buckets of all partitions.
Each encoder can thus use tuning parameters for its specific distribution (e.g., different Golomb-Rice parameters).
\cref{fig:orthoillustration} gives an illustration of interleaved coding.

It is also possible to mix different encoding techniques, similar to what PTHash does.
Larger buckets are accessed more often than smaller buckets because they contain more keys.
Hence, it is beneficial to use an encoding technique which is optimized for fast lookup time (e.g., Compact) for the larger buckets.
Conversely, the encoding for the seeds belonging to smaller buckets should be tuned for space efficiency (e.g., Golomb-Rice).
To conclude this section, we point out that each of the $B$ encoders introduces some metadata overhead (e.g., for storing its parameters).
Using rather small partition sizes $P$ decreases the number $B$ of encoders and therefore the constant overhead.

\subsection{GPU Parallelization}
\label{ss:gpuimpl}
We provide a GPU implementation for even faster construction.
On a GPU, each \emph{workgroup} executes independently, typically with its own subset of data. Within each workgroup, individual \emph{threads} execute concurrently. Threads within the same workgroup can share data and synchronize with each other through mechanisms like barriers and shared memory. Thread level parallelism allows for fine-grained parallel execution of instructions within a workgroup. However, only threads which follow the same control path and thus execute the same instruction at the same time can be executed in parallel. It is therefore crucial to avoid divergent control paths.
As a first step, our parallel implementation transfers the keys to the GPU, before partitioning them. Afterwards, we sort the buckets and start the search.

\myparagraph{Search}
Our fine-grained partitioning naturally maps to the architecture of a GPU.
Each partition is processed by one workgroup.
The small partition sizes enable performing the search entirely using fast but small shared memory.
Bucket by bucket, all threads of the workgroup cooperate to quickly find the smallest working seed.
A CPU implementation would usually do so using a nested loop.
The outer loop would iterate over seed values and the inner loop over keys.
If a collision occurs, it would immediately leave the inner loop and continue with the first key and the next seed.
However, on a GPU, leaving the inner loop would result in divergence because threads might encounter a collision after a different number of keys.
This is illustrated in the left of \cref{fig:brute_force}.

Instead, we use the technique described in \cite{sanders1994emulating} to emulate the behavior of the nested loop using a single loop to reduce divergence.
Hence, our GPU implementation parallelizes over partitions, seeds and keys.
The inner loop is emulated by incrementing the key index in each iteration.
If a collision occurs, we reset the key index and emulate the behavior of the outer loop by atomically incrementing a seed counter which is shared among all threads.
If the last key did not collide, we found a working seed.
Multiple threads can find a seed in the same iteration. 
To reduce entropy, we use an atomic minimum to identify the smallest of those seeds.
Note that this finds the smallest working seed overall because all threads finding a working seed must have processed each key.
Therefore, if there was a smaller working seed, it would have been found in an earlier loop iteration.
We give pseudocode in \cref{alg:simt} and illustrate the behavior on the right in \cref{fig:brute_force}.
During search, we only access shared memory and perform fast arithmetic operations.
We remark that specific optimizations for our additive displacement hash function are not shown here, i.e. after calculating the initial positions we can apply new displacements using only additions.

\begin{figure}
	\centering
	\begin{minipage}{.48\textwidth}
		\centering
		\begin{algorithm}[H]
			\caption{Seed search for one bucket.}
			\label{alg:simt}
			\begin{algorithmic}
				\State \textbf{shared} sFound $\gets \infty$
				\State \textbf{shared} sNext $\gets$ threadCount
				\State seed $\gets$ threadId, keyIndex $\gets$ 0
				\While{sFound $= \infty$}
				\State isCollision $\gets$ \Call{coll}{seed, keyIndex}
				\State keyIndex $\gets$ keyIndex + 1
				\If{isCollision}
				\State keyIndex $\gets$ 0
				\State seed $\gets$ \Call{atomAdd}{sNext, 1}
				\ElsIf{keyIndex = bucketSize}
				\State sFound $\gets$ \Call{atomMin}{sFound, seed}
				\EndIf
				\EndWhile
			\end{algorithmic}
		\end{algorithm}
	\end{minipage}
	\hfill
	\begin{minipage}{.50\textwidth}
		\begin{figure}[H]
			\centering
			\includegraphics[scale=0.9]{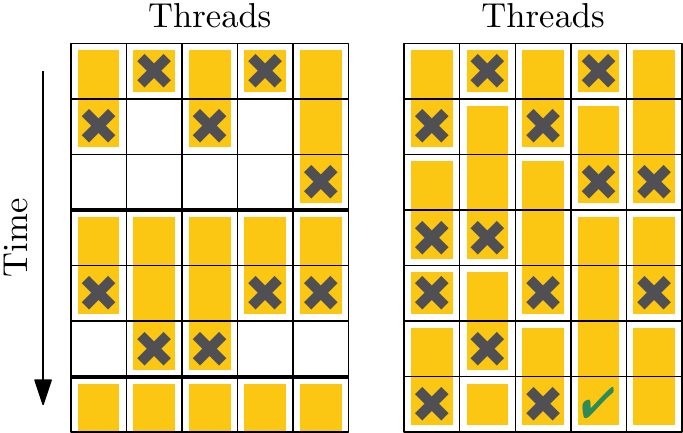}
			\caption{Each box represents one seed tested by one thread.
				Left: Synchronized nested loop.
				Right: \cref{alg:simt} where we continue with the first key and the next seed after a collision.
			}
			\label{fig:brute_force}
		\end{figure}
	\end{minipage}
\end{figure}

%% file: sections/experiments.tex
In \cref{s:internalComparison}, we gradually integrate our improvements to show the individual effects.
Then we compare our GPU and CPU implementation with the state of the art in \cref{s:competitorComparison}.
We use a machine with an Intel Core i7-11700 CPU with 64~GiB of DDR4 RAM running Ubuntu 22.04.1.
Each core has 48~KiB L1 and 512~KiB L2 data cache.
As a GPU, we use an Nvidia RTX 3090 and use Vulkan 1.3.236 to interface with it.
We compile using GCC 11.4.0 and compiler options \texttt{-march=native} and \texttt{-O3}.
All benchmarks use random strings of random length between 10 and 50 characters as input which is adopted from previous work \cite{lehmann2023sichash, bez2022high, lehmann2023shockhash, lehmann2023bipartite}.
Note that almost all competitors first generate master hash codes of the input.
This makes the construction largely independent of the input distribution.
We measure the query time by querying each key once in random order.
All experiments use 100~million keys, $\lambda=8$ and an average partition size of 2\,500 if not stated otherwise.
Our source code is public under the General Public License.
You can find it through the links on the title page of this paper.

\begin{figure}[t]
  \begin{minipage}[b]{0.49\textwidth}
    \input{fig/bucketSizePlot.tex}
    \caption{Comparing original PTHash (EF, $\alpha=0.99$) to fine-grained partitioning with additive displacement hash (mono-R). We then add interleaved coding (IC-R) and optimized bucketing (OB) individually. Putting it all together we arrive at PHOBIC. All single-threaded. There are small differences in query time.
      }
    \label{fig:bucketSizePlot}
  \end{minipage}
  \hfill
  \begin{minipage}[b]{0.49\textwidth}
    \input{fig/encodersPlot.tex}
    \caption{
      Query time and space consumption of \textbf{E}lias-\textbf{F}ano, \textbf{D}ictionary, \textbf{C}ompact, and \textbf{R}ice encoders.
      Points prefixed ``mono'' place all seeds into a single encoder and those prefixed ``IC'' use interleaved coding.
      The curve shows different mixtures (see \cref{s:partitioning}) of Compact and Rice encoders in the interleaved coding.
    }
    \label{fig:encodersPlot}
  \end{minipage}
\end{figure}

\subsection{From PTHash to PHOBIC}\label{s:internalComparison}
We now gradually introduce our improvements to PTHash.
As basic improvements, we replace the initial hash function with xxHash \cite{xxHash} and implement faster parallel partitioning.
In all experiments, PTHash contains these changes as well to focus on our algorithmic improvements.
\Cref{fig:bucketSizePlot} gives measurements for the different improvement steps and shows them in different combinations.

\myparagraph{Interleaved Coding}
Partitioning of PTHash is already used in PTHash-HEM for parallelization \cite{pibiri2024parallel}.
PTHash-HEM uses partitions of size $\approx$$10^6$.
Smaller partitions only lead to minor improvements, as we show in \cref{app:additionalExperiments}.
However, smaller partition sizes shine when used with our newly introduced interleaved coding (\cref{ss:multi}).
Interleaved coding uses $P/\lambda$ encoders, where $P$ is the expected partition size.
Reducing the partition size can significantly reduce constant space overheads, as we also show in \cref{app:additionalExperiments}. %
\Cref{fig:bucketSizePlot} compares the technique of placing all seeds into a single Rice encoder (orange curve) to placing the keys using interleaved Rice coding (black). Interleaved coding consistently improves space efficiency by \interleavedImprovement{} bits/key.
\cref{fig:encodersPlot} compares different combinations of encoders, which were partially used in the original implementation.
Interleaved coding allows for mixing of different encoding techniques.
If we use this to encode the keys using different numbers of Compact and Rice encoders, we can cover the entire query time to space trade-off in our configuration.
Note that the construction performance is similar for all approaches because the encoding is fast compared to the remaining construction.

\myparagraph{Optimized Bucket Assignment}
In \cref{fig:bucketSizePlot} we also show how using the optimized bucket function affects construction speed and space.
Our optimization of the bucket assignment function is particularly helpful to construct very space efficient configurations.
When compared for same construction time, the optimized function is up to \bucketerImprovement{} bits/key more space efficient relative to the original PTHash bucket assignment.

\myparagraph{Further Remarks}
With interleaved coding, another improvement originates from the secondary bucket ordering.
Primarily the buckets are sorted in non-increasing size.
Secondarily sorting in increasing expected size reduces the space consumption by \secsortdiff{} bits/key compared to decreasing expected size.
The reason for this behavior remains an open problem.

Original PTHash observed significant performance improvements by first calculating a non-minimal PHF and repairing the ``gaps'' afterwards.
Refer to  \cite{pibiri2021pthash} for details.
This trick does not result in an improvement when using PHOBIC.

\begin{figure}[t]
  \begin{subfigure}[b]{0.49\textwidth}
    \centering
    \input{fig/cpuVsGpu}
    \caption{Construction performance on CPU and GPU for $\lambda=8.0$ and 8 threads, comparing different $n$. The GPU requires large $n$ to fully utilize its computing resources.}
    \label{fig:cpuVsGpu}
  \end{subfigure}
  \hfill
  \begin{subfigure}[b]{0.49\textwidth}
    \centering
      \input{fig/stepsByAPlot.tex}
  \caption{Different construction steps by values of average bucket size $\lambda$.  }
  \label{fig:stepsByAPlot}
  \end{subfigure}
  \caption{GPU performance for different input sizes (left) and $\lambda$ (right).}
\end{figure}
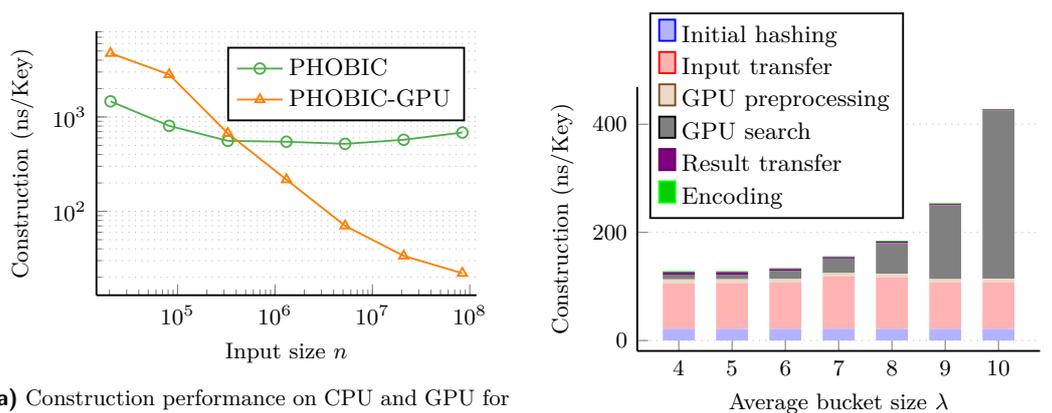

\myparagraph{GPU Parallelization}
Our final contribution is a GPU implementation to speed up construction.
Our implementation parallelizes over partitions, seeds and keys.
The GPU implementation is mainly useful for large average bucket sizes $\lambda$.
This is well illustrated in \cref{fig:stepsByAPlot}: For smaller values of $\lambda$, the construction time is dominated by the time to transfer the input data to the GPU.
We also compare CPU and GPU construction speed for different input sizes in \cref{fig:cpuVsGpu}.
The GPU requires a large number of input keys and thus a large number of partitions before its computing resources are fully utilized.
Overall, the GPU outperforms the CPU for a sufficiently large $\lambda$ and $n$. 
We use the GPU only to accelerate construction, while measuring all queries on the CPU.%

\subsection{Comparison to Other Methods}\label{s:competitorComparison}
We compare our new approach to several other methods from the literature.
First and foremost, we compare against the original PTHash \cite{pibiri2021pthash,pibiri2024parallel} implementation.
The comparison also includes the fingerprinting approaches BBHash \cite{limasset2017fast}, FMPH \cite{beling2023fingerprinting}, and FMPHGO \cite{beling2023fingerprinting}.
We also compare against RecSplit \cite{esposito2020recsplit} and approaches based on it, such as SIMDRecSplit \cite{bez2022high}, ShockHash-RS \cite{lehmann2023shockhash}, and bipartite ShockHash-RS \cite{lehmann2023bipartite}.
Finally, we also compare against CHD \cite{belazzougui2009hash} and SicHash \cite{lehmann2023sichash}.

\begin{figure}[t]
  \centering
  \begin{tikzpicture}
    \ref*{legendEvalParetoQuery}
  \end{tikzpicture}
  \vspace{3mm}
    \input{fig/pareto_construction}
    \input{fig/pareto_queries}
    \caption{Construction throughput (left) and query throughput (right) for various methods on 100 million keys and using a single processing thread.
        }
  \label{fig:pareto}
\end{figure}
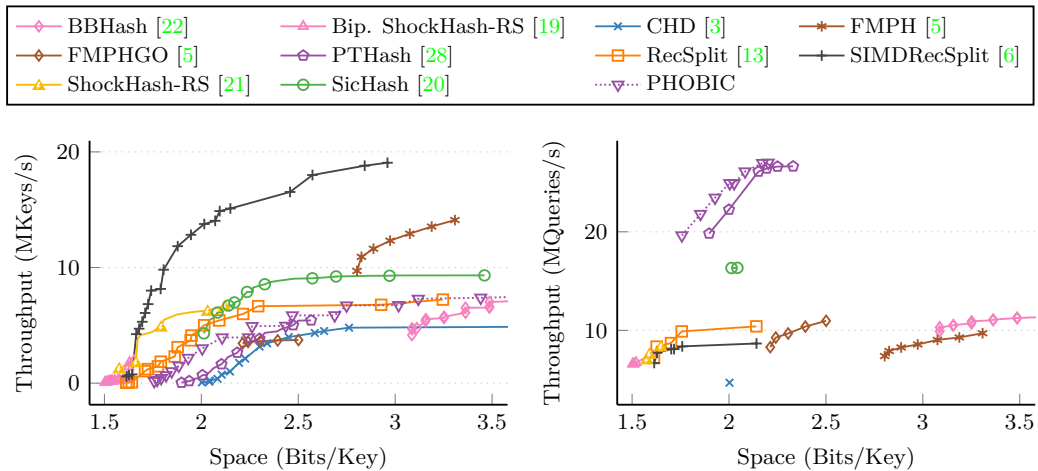

\begin{table}[t]
    \caption{Performance of various methods on $100$ million keys.}
	\label{tab:overviewTable}
	\input{fig/overviewTable}

\end{table}

Each method has a wide range of configurations that provide a trade-off between space, construction time, and query time.
To give an initial overview, we show a Pareto front for each method in \cref{fig:pareto}.
A configuration is on the Pareto front if no other configuration of the same method is simultaneously faster and more space efficient.
For this plot we use a single thread (a multithreaded measurement would mainly show what method implemented the partitioning step most efficiently instead of focusing on the algorithmic aspects).
The figure shows that PTHash and PHOBIC are clear winners in terms of query performance.
Even though BBHash \cite{limasset2017fast} and FMPH \cite{beling2023fingerprinting} are also focused on fast queries, they are significantly slower than PTHash and PHOBIC.
\cref{fig:bucketSizePlot} shows that PHOBIC consistently saves about \totalBitImprovement{} bits/key for a large range of different construction times while maintaining the good query speed.
We remark that this is a significant reduction in space considering the proximity to the space lower bound.
The competitors achieving even lower space consumption (i.e. RecSplit \cite{esposito2020recsplit}, SIMDRecSplit \cite{bez2022high}, ShockHash-RS \cite{lehmann2023shockhash}, and bipartite ShockHash-RS \cite{lehmann2023bipartite}) all have a rather slow query performance.
However, somewhat surprisingly, SIMDRecSplit has the fastest construction even for less space efficient configurations.
SicHash \cite{lehmann2023sichash} takes a middle ground with faster construction than PHOBIC and query performance between PHOBIC and the RecSplit variants.

\Cref{tab:overviewTable} gives a selection of configuration parameters for direct comparison, mostly taken from the corresponding papers.
\Cref{s:competitorComparisonAppendix} gives the same table measured on a large machine with 64 threads.
Comparing configurations with the same space consumption,
PHOBIC is significantly faster to construct than the original PTHash implementation.
Comparing configurations that both need 1.86 bits/key and have a similar query time, PHOBIC can be constructed \speedupPtHashVsPhobic{} times faster than PTHash.

On the GPU, we compare against the only available GPU construction, RecSplit-GPU \cite{bez2022high}. \cref{fig:paretoGpu} in \cref{s:competitorComparisonAppendix} illustrates the comparison.
Basically, we achieve the same peak construction throughput as RecSplit-GPU for the less space efficient configurations.
The queries of both approaches are done on the CPU, so the fact that PHOBIC offers much faster queries applies here as well (see \cref{fig:pareto}).
Comparing the multithreaded CPU implementation and the GPU implementation of PHOBIC, we get a construction speedup of \speedupPhobicVsGpu{} for $\lambda=9$ with interleaved Rice coding.
Note that with $\lambda=9$, the GPU still spends a lot of its construction time on transferring the input data (see \cref{fig:stepsByAPlot}), but much larger values of $\lambda$ are not feasible on the CPU.

Directly comparing the performance of CPU and GPU is always difficult because of the different hardware architectures.
Given that the power consumption is a major cost factor in production environments, we measure it using a Voltcraft Energy Check 3000 wattmeter.
For CPU-only measurements, we dismount the GPU.
The machine requires about $405~W$ during the search step of the GPU version and $195~W$ for the multithreaded CPU implementation.
Thus, the above speedup of \speedupPhobicVsGpu{} translates to roughly 30 times lower energy consumption for constructing an MPHF on the GPU.
Single-threaded CPU construction requires $74~W$ which is less energy efficient compared to multithreading.

%% file: fig/bucketSizePlot.tex
\centering
\begin{tikzpicture}

	\begin{axis}[
		xlabel={Space (Bits/Key)},
		ylabel={Construction (ns/Key)},
		xtick distance={0.1},
		width=5.4cm,	
		height=3.5cm,
		xmax=2.2,
		ymode=log,
		legend style={at={(1.01,0.8)},anchor=east},
		legend columns=1,
		]

  \addplot coordinates { (2.484872,211.884) (2.334752,257.13) (2.236447,307.87) (2.170549,405.465) (2.108244,544.89) (2.05025,797.29) (2.011478,1180.33) (1.98556,1659.56) (1.953413,2741.92) (1.925782,4318.94) (1.907809,6602.46) (1.896884,11147.6) (1.885368,16811.5) (1.87035,29453) (1.861008,49125.3) (1.856227,80052) };
		
		\addlegendentry{PTHash};

  \addplot coordinates { (2.869866,136.701) (2.56302,153.1444) (2.318139,177.3099) (2.165057,215.987) (2.061756,276.96) (2.011145,374.744) (1.961846,533.025) (1.925813,788.057) (1.906485,1189.39) (1.894352,1865.65) (1.886261,3034.89) (1.876763,4882.17) (1.871863,8327.76) (1.871922,13862.9) (1.880054,24911.3) (1.87833,44067.9) (1.892414,80243.9) };
  
  \addlegendentry{Partitioned};

  \addplot coordinates { (2.752205,136.394) (2.459467,152.498) (2.25753,177.3744) (2.115628,215.53) (2.02267,276.916) (1.960056,375.012) (1.916887,532.719) (1.884018,787.895) (1.862107,1189.45) (1.844773,1865.8) (1.831431,3035.09) (1.82363,4881.99) (1.819362,8327.59) (1.819077,13862.8) (1.820701,24911.2) (1.824565,44067.8) (1.832079,80243.7) };
		
		\addlegendentry{Part. + IC};

 \addplot coordinates { (3.057469,156.895) (2.709622,177.628) (2.478659,209.696) (2.290874,257.807) (2.165754,335.29) (2.078433,450.879) (2.024198,633.555) (1.960632,912.809) (1.923194,1339.75) (1.899759,2017.26) (1.861879,3080.25) (1.838928,4708.19) (1.797238,7425.84) (1.78236,11415.1) (1.763211,18424.6) (1.752998,28236.7) (1.746439,45015.6) (1.735935,72224.1) };
	
	\addlegendentry{Part. + OB};
		
  \addplot coordinates { (2.751731,156.078) (2.471857,177.659) (2.263498,209.3401) (2.109494,257.994) (2.00719,334.991) (1.937021,451.041) (1.886925,633.124) (1.851239,912.35) (1.821931,1339.82) (1.795384,2017.45) (1.775178,3079.91) (1.758723,4708.09) (1.742651,7425.9) (1.728431,11414.8) (1.716192,18424.3) (1.709285,28236.6) (1.700233,45015.5) (1.690934,72223.8) };
		
		\addlegendentry{PHOBIC};

        \addplot[mark=|,color=gray] coordinates { (1.856227,76137.9) (1.690934,76137.9) };

        \node[color=gray,scale=0.8,anchor=west] at (axis cs: 1.68,112000) {\totalBitImprovement{} bits/key};

	\end{axis}
  
\end{tikzpicture}

%% file: fig/encodersPlot.tex
\centering
\begin{tikzpicture}
	
\begin{axis}[
	xlabel={Space (Bits/Key)},
	ylabel={Query (ns/Query)},
	width=8cm,
	height=4cm,
	nodes near coords,
	every node near coord/.append style={font=\footnotesize},
	nodes near coords align={right},
	mark=*,
	xmin=1.7,
	xtick distance={0.2},
	width=5cm,	
	height=3.5cm,	
	]
    \addplot[point meta={symbolic={mono-R}},color=colorSicHash,] coordinates { (1.862005,48.7568) };
    \addplot[point meta={symbolic={IC-R}},color=colorSicHash,] coordinates { (1.775256,52.1112) };
    \addplot[point meta={symbolic={mono-C}},nodes near coords align={left},color=colorRecSplit,] coordinates { (3.012125,38.729) };
    \addplot[point meta={symbolic={IC-C}},nodes near coords align={left},color=colorRecSplit,] coordinates { (2.207524,37.5902) };
    \addplot[point meta={symbolic={mono-D}},color=colorSimdRecSplit,] coordinates { (2.162942,38.6138) };
    \addplot[point meta={symbolic={IC-D}},color=colorSimdRecSplit,] coordinates { (2.598593,55.9759) };
    \addplot[point meta={symbolic={mono-EF}},color=colorChd,] coordinates { (1.908968,50.3403) };
    \addplot[point meta={symbolic={IC-EF}},color=colorChd,] coordinates { (1.82427,53.7209) };

    \addplot[nodes near coords={},no marks] coordinates { (1.775257,52.8667) (1.774664,52.1663) (1.808672,48.4661) (1.838603,47.6356) (1.869595,46.9851) (1.899602,45.6212) (1.928496,44.225) (1.951897,43.163) (1.985273,42.259) (2.008074,41.6127) (2.034071,40.6119) (2.056511,40.2319) (2.084452,39.2592) (2.118569,38.91) (2.138764,38.0823) (2.208724,37.1224) };
 
\end{axis}

\end{tikzpicture}

%% file: fig/cpuVsGpu.tex
\centering
    \begin{tikzpicture}
        \begin{axis}[
            xlabel={Input size $n$},
            ylabel={Construction (ns/Key)},
            legend style={at={(0.97,0.8)},anchor=east},
            width=5cm,
            height=3.5cm,
            xmode=log,
            ymode=log,
          ]
          \addplot coordinates { (20480,1464.84) (81920,805.664) (327680,558.472) (1310720,546.265) (5242880,518.99) (20971520,573.158) (83886080,684.237) };
          \addlegendentry{PHOBIC};
          \addplot coordinates { (20480,4736.33) (81920,2819.82) (327680,671.387) (1310720,216.675) (5242880,70.1904) (20971520,33.6647) (83886080,22.0656) };
          \addlegendentry{PHOBIC-GPU};
        \end{axis}
    \end{tikzpicture}

%% file: fig/stepsByAPlot.tex
\centering
\begin{tikzpicture}
  \begin{axis}[
      yticklabel style={
        /pgf/number format/fixed,
        /pgf/number format/precision=5
      },
      ybar stacked,
      bar width=0.6,
      scaled y ticks=false,
      ylabel={Construction (ns/Key)},
      xlabel={Average bucket size $\lambda$},
      xtick distance={1},
      width=5.3cm,
      height=3.5cm,
      xmin=3.5,
      xmax=10.5,
      legend style={at={(0.03,0.5)},anchor=south west},
      every axis plot post/.append style={draw=none},
    ]
    \addplot coordinates { (4,21.58321) (5,21.40904) (6,21.4251) (7,21.88522) (8,21.40786) (9,21.39463) (10,21.4107) };
    \addlegendentry{Initial hashing};

    \addplot coordinates { (4,84.4124) (5,85.08671) (6,85.63466) (7,96.41894) (8,95.68914) (9,85.8862) (10,86.52367) };
    \addlegendentry{Input transfer};

    \addplot coordinates { (4,7.13292) (5,6.84422) (6,6.66811) (7,6.44224) (8,6.34409) (9,6.30569) (10,6.33174) };
    \addlegendentry{GPU preprocessing};

    \addplot coordinates { (4,7.55374) (5,8.55419) (6,14.57101) (7,26.93745) (8,57.03956) (9,137.239) (10,311.364) };
    \addlegendentry{GPU search};

    \addplot coordinates { (4,5.00403) (5,4.42276) (6,4.01965) (7,2.54838) (8,2.74455) (9,1.77296) (10,1.71137) };
    \addlegendentry{Result transfer};

    \addplot coordinates { (4,2.34075) (5,1.8824) (6,1.53264) (7,1.41731) (8,1.24917) (9,1.18146) (10,0.96184) };
    \addlegendentry{Encoding};
  \end{axis}
\end{tikzpicture}

%% file: fig/pareto_construction.tex
    \begin{tikzpicture}
        \begin{axis}[
            plotPareto,
            ylabel={Throughput (MKeys/s)},
            legend to name=legendEvalParetoQuery,
            legend style={nodes={scale=0.9, transform shape}},
            legend columns=4,
          ]
          \addplot[mark=diamond,color=colorBbhash,solid] coordinates { (3.0864,4.1587) (3.08677,4.67181) (3.08767,4.74451) (3.11123,4.7533) (3.15659,5.37143) (3.15797,5.52944) (3.25169,5.63634) (3.25238,5.74581) (3.36289,6.09273) (3.36289,6.51848) (3.4863,6.5304) (3.4863,6.62252) (3.48663,7.01607) (3.61817,7.10227) (3.61817,7.2469) (3.61817,7.57978) };
          \addlegendentry{BBHash \cite{limasset2017fast}};
          \addplot[mark=triangle,color=colorBipartiteShockHash,solid] coordinates { (1.50308,0.0766555) (1.50968,0.09322) (1.51513,0.115142) (1.52467,0.173996) (1.53164,0.177995) (1.53603,0.179202) (1.54023,0.204943) (1.55267,0.211247) (1.55295,0.213541) (1.56942,0.256082) (1.57851,0.263762) (1.59513,0.278768) (1.61785,0.410152) (1.62489,1.41323) (1.62941,1.78164) (1.65197,1.96059) };
          \addlegendentry{Bip. ShockHash-RS \cite{lehmann2023bipartite}};
          \addplot[mark=x,color=colorChd,solid] coordinates { (2.00266,0.0619156) (2.04135,0.117405) (2.04312,0.233917) (2.08329,0.375311) (2.10518,0.734792) (2.1477,1.00521) (2.19464,1.76032) (2.22621,2.12377) (2.30042,3.15836) (2.33922,3.44364) (2.45054,3.9847) (2.58529,4.34216) (2.63336,4.50938) (2.762,4.79754) (3.67454,4.87543) (3.82265,5.08699) };
          \addlegendentry{CHD \cite{belazzougui2009hash}};
          \addplot[mark=asterisk,color=colorRustFmph,solid] coordinates { (2.80301,9.70403) (2.82608,10.9218) (2.88811,11.6266) (2.9738,12.3229) (3.0755,12.9266) (3.18827,13.5428) (3.30824,14.1044) };
          \addlegendentry{FMPH \cite{beling2023fingerprinting}};
          \addplot[mark=diamond,color=colorRustFmphGo,solid] coordinates { (2.21264,3.43914) (2.24107,3.56455) (2.30511,3.63544) (2.39382,3.69891) (2.50102,3.72953) };
          \addlegendentry{FMPHGO \cite{beling2023fingerprinting}};
          \addplot[mark=pentagon,color=colorPthash,solid,mark repeat*=4] coordinates { (1.89746,0.0397137) (1.9011,0.0586782) (1.9192,0.0915153) (1.92152,0.120967) (1.93744,0.180195) (1.95091,0.239669) (1.96597,0.365619) (1.99081,0.492397) (2.00574,0.692329) (2.02118,0.694444) (2.04394,0.96652) (2.06024,1.27259) (2.10402,1.63929) (2.12874,2.07641) (2.13494,2.2307) (2.18041,2.58759) (2.18517,2.63852) (2.21169,3.15348) (2.23859,3.29012) (2.27566,3.73874) (2.30018,3.8373) (2.32786,4.35123) (2.39055,4.47908) (2.41032,4.65354) (2.47379,5.01605) (2.49976,5.40862) (2.56621,5.4277) };
          \addlegendentry{PTHash \cite{pibiri2021pthash}};
          \addplot[mark=square,color=colorRecSplit,solid] coordinates { (1.61269,0.0389603) (1.62521,0.0704746) (1.64126,0.0737047) (1.7096,1.04029) (1.72459,1.21714) (1.77618,1.45575) (1.79101,1.85694) (1.8625,2.3092) (1.87797,3.11265) (1.94477,3.70178) (1.9495,4.11421) (2.0147,5.00952) (2.09208,5.42446) (2.21561,5.98444) (2.29325,6.65557) (2.92784,6.77736) (3.24602,7.227) };
          \addlegendentry{RecSplit \cite{esposito2020recsplit}};
          \addplot[mark=+,color=colorSimdRecSplit,solid] coordinates { (1.61401,0.589803) (1.62674,0.694811) (1.64301,0.754461) (1.66403,4.23801) (1.67891,4.71143) (1.69505,5.3087) (1.70996,6.0698) (1.72535,6.83527) (1.7412,8.00512) (1.78983,8.1493) (1.80523,9.81258) (1.87807,11.8455) (1.94642,12.8304) (2.01525,13.7627) (2.07131,14.041) (2.09478,14.8832) (2.14974,15.1035) (2.45761,16.5344) (2.57332,17.9986) (2.8426,18.8041) (2.96025,19.0621) };
          \addlegendentry{SIMDRecSplit \cite{bez2022high}};
          \addplot[mark=shockhash,color=colorShockHash,solid,mark repeat*=4] coordinates { (1.57574,1.29808) (1.58941,1.38263) (1.60608,1.50705) (1.64858,1.75052) (1.66222,1.83284) (1.67118,3.81098) (1.68473,4.14869) (1.74476,4.41073) (1.79193,4.90918) (1.80633,5.47106) (1.87429,5.93754) (1.94997,6.14553) (2.03217,6.28141) (2.12034,6.68494) };
          \addlegendentry{ShockHash-RS \cite{lehmann2023shockhash}};
          \addplot[mark=o,color=colorSicHash,solid,mark repeat*=4] coordinates { (2.0129,4.30775) (2.04322,5.48637) (2.05228,5.5457) (2.0735,5.79039) (2.08252,6.10501) (2.10937,6.13911) (2.11293,6.28575) (2.113,6.38284) (2.14007,6.71096) (2.14786,6.71366) (2.1705,6.7627) (2.17066,6.80504) (2.17075,6.97593) (2.20077,7.17309) (2.20898,7.30034) (2.20907,7.42115) (2.23557,7.88768) (2.26578,7.90701) (2.26621,8.26173) (2.29696,8.42105) (2.32715,8.55066) (2.35758,8.74202) (2.38843,8.83548) (2.48036,9.03832) (2.5724,9.06701) (2.60313,9.10664) (2.63366,9.13409) (2.69495,9.21319) (2.69506,9.22679) (2.75631,9.243) (2.78678,9.25326) (2.87879,9.28678) (2.97114,9.28936) (3.00103,9.31272) (3.09348,9.31532) (3.46074,9.3214) };
          \addlegendentry{SicHash \cite{lehmann2023sichash}};
          \addplot[mark=flippedTriangle,color=colorDensePtHash,densely dotted] coordinates { (1.75671,0.175464) (1.77226,0.273979) (1.79268,0.435411) (1.81781,0.679103) (1.84684,1.02221) (1.88378,1.50632) (1.93259,2.18379) (2.00462,3.00499) (2.10746,3.94976) (2.2432,3.95413) (2.26187,4.92781) (2.43504,4.95589) (2.46952,5.85789) (2.68795,5.86476) (2.74919,6.70961) (3.01775,6.74491) (3.11932,7.29076) (3.44142,7.38716) (3.61344,7.45323) (3.95896,7.47943) (4.30327,7.76458) (4.75314,7.83699) (5.3373,7.88706) (6.01293,7.96432) };
          \addlegendentry{PHOBIC};
        \end{axis}
    \end{tikzpicture}

%% file: fig/pareto_queries.tex
    \begin{tikzpicture}
        \begin{axis}[
            plotPareto,
            ylabel={Throughput (MQueries/s)},
          ]
          \addplot[mark=diamond,color=colorBbhash,solid] coordinates { (3.0864,9.89609) (3.08677,10.3082) (3.15659,10.4943) (3.15797,10.5296) (3.25118,10.7135) (3.25169,10.9146) (3.36289,11.04) (3.36289,11.0939) (3.4863,11.2511) (3.4863,11.2587) (3.48663,11.2625) (3.61817,11.3417) (3.61817,11.4025) (3.61817,11.4051) };
          \addlegendentry{BBHash \cite{limasset2017fast}};
          \addplot[mark=triangle,color=colorBipartiteShockHash,solid] coordinates { (1.50308,6.59935) (1.50968,6.59979) (1.52045,6.73627) (1.59888,7.53864) };
          \addlegendentry{Bip. ShockHash-RS \cite{lehmann2023bipartite}};
          \addplot[mark=x,color=colorChd,solid] coordinates { (2.00266,4.69241) };
          \addlegendentry{CHD \cite{belazzougui2009hash}};
          \addplot[mark=asterisk,color=colorRustFmph,solid] coordinates { (2.80301,7.39426) (2.82608,7.90326) (2.88811,8.26583) (2.9738,8.57265) (3.0755,9.03342) (3.18827,9.28505) (3.30824,9.71251) };
          \addlegendentry{FMPH \cite{beling2023fingerprinting}};
          \addplot[mark=diamond,color=colorRustFmphGo,solid] coordinates { (2.21264,8.29807) (2.24107,9.27128) (2.30511,9.68429) (2.39382,10.3853) (2.50102,10.9697) };
          \addlegendentry{FMPHGO \cite{beling2023fingerprinting}};
          \addplot[mark=flippedTriangle,color=colorDensePtHash,densely dotted] coordinates { (1.75671,19.6309) (1.85289,21.8055) (1.92802,23.4522) (2.00375,24.9252) (2.02659,24.9252) (2.08274,26.1165) (2.17307,26.9687) (2.20272,27.0051) };
          \addlegendentry{PHOBIC};
          \addplot[mark=pentagon,color=colorPthash,solid] coordinates { (1.89746,19.8255) (1.99992,22.2519) (2.15187,26.1097) (2.19463,26.4271) (2.25015,26.617) (2.33082,26.6383) };
          \addlegendentry{PTHash \cite{pibiri2021pthash}};
          \addplot[mark=square,color=colorRecSplit,solid] coordinates { (1.61269,7.19269) (1.62691,8.35492) (1.7009,8.71384) (1.75626,9.89609) (2.14104,10.408) };
          \addlegendentry{RecSplit \cite{esposito2020recsplit}};
          \addplot[mark=+,color=colorSimdRecSplit,solid] coordinates { (1.61401,6.69478) (1.62881,7.72499) (1.70183,8.09651) (1.71662,8.10241) (1.75873,8.37872) (2.1418,8.66927) };
          \addlegendentry{SIMDRecSplit \cite{bez2022high}};
          \addplot[mark=shockhash,color=colorShockHash,solid] coordinates { (1.57574,7.0932) (1.58941,7.65697) (1.64858,8.40336) };
          \addlegendentry{ShockHash-RS \cite{lehmann2023shockhash}};
          \addplot[mark=o,color=colorSicHash,solid] coordinates { (2.0129,16.3239) (2.04324,16.3425) };
          \addlegendentry{SicHash \cite{lehmann2023sichash}};

          \legend{};
        \end{axis}
    \end{tikzpicture}

%% file: fig/overviewTable.tex
\addtolength\tabcolsep{-1.9pt}
\small
\centering
	\begin{tabular}[t]{l rrrrr}
		\toprule
        Method & Space & Query & \multicolumn{3}{c}{Construction (ns/key)} \\
               \cmidrule(lr){4-6}
    & (bits/key) & (ns/query) & 1 Thread & 8 Threads & Speedup \\ \midrule

                   Bip. SH-RS, $n$=$64$, $b$=$2000$ & 1.52 & 160 &  5\,756 &  1\,218 & 4.7 \\ \midrule
                                 CHD, $\lambda$=$3$ & 2.27 & 222 &     352 &       - &   - \\
                                 CHD, $\lambda$=$5$ & 2.07 & 207 &  2\,206 &       - &   - \\ \midrule
                               FMPH, $\gamma$=$2.0$ & 3.40 & 100 &      69 &      17 & 4.0 \\
                               FMPH, $\gamma$=$1.0$ & 2.80 & 134 &      99 &      24 & 4.0 \\ \midrule
                   SIMDRecSplit, $n$=$8$, $b$=$100$ & 1.81 & 124 &     109 &      20 & 5.2 \\
                 SIMDRecSplit, $n$=$14$, $b$=$2000$ & 1.58 & 143 & 11\,062 &  2\,360 & 4.7 \\ \midrule
    SicHash, $\alpha$=$0.9$, $p_1$=$21$, $p_2$=$78$ & 2.41 &  72 &     129 &      25 & 5.0 \\
   SicHash, $\alpha$=$0.97$, $p_1$=$45$, $p_2$=$31$ & 2.08 &  64 &     179 &      32 & 5.6 \\ \midrule
      PTHash, $\lambda$=$4.0$, $\alpha$=$0.99$, C-C & 3.19 &  35 &     314 &     143 & 2.2 \\
       PTHash, $\lambda$=$5.0$, $\alpha$=$0.99$, EF & 2.11 &  54 &     525 &     252 & 2.1 \\
      PTHash, $\lambda$=$10.5$, $\alpha$=$0.99$, EF & 1.86 &  49 & 82\,721 & 35\,048 & 2.4 \\ \midrule
  PTHash-HEM, $\lambda$=$4.0$, $\alpha$=$0.99$, C-C & 3.19 &  39 &     299 &      45 & 6.6 \\
   PTHash-HEM, $\lambda$=$5.0$, $\alpha$=$0.99$, EF & 2.11 &  58 &     582 &      86 & 6.7 \\ \midrule
      PHOBIC, $\lambda$=$3.9$, $\alpha$=$1.0$, IC-C & 3.18 &  40 &     197 &      32 & 6.2 \\
      PHOBIC, $\lambda$=$4.5$, $\alpha$=$1.0$, IC-R & 2.11 &  57 &     254 &      40 & 6.2 \\
      PHOBIC, $\lambda$=$6.5$, $\alpha$=$1.0$, IC-R & 1.85 &  52 &     992 &     176 & 5.6 \\
      PHOBIC, $\lambda$=$9.0$, $\alpha$=$1.0$, IC-R & 1.74 &  50 &  9\,171 &  1\,781 & 5.1 \\
        \midrule
    &            &            &  \multicolumn{3}{r}{GPU + 8 CPU Threads} \\ \midrule

      PHOBIC-GPU, $\lambda$=$9.0$, IC-C & 2.17 &  37 &   &     28 &   \\
      PHOBIC-GPU, $\lambda$=$9.0$, IC-R & 1.76 &  52 &   &     27 &   \\
     PHOBIC-GPU, $\lambda$=$13.0$, IC-R & 1.68 &  50 &   &    560 &   \\
     PHOBIC-GPU, $\lambda$=$14.0$, IC-R & 1.67 &  49 &   & 1\,470 &   \\ \midrule
    RecSplit-GPU, $\ell$=$8$, $b$=$100$ & 1.81 & 126 &   &     24 &   \\
  RecSplit-GPU, $\ell$=$14$, $b$=$2000$ & 1.58 & 147 &   &     80 &   \\
  RecSplit-GPU, $\ell$=$18$, $b$=$2000$ & 1.55 & 135 &   & 1\,732 &   \\
		\bottomrule

	\end{tabular}

%% file: sections/conclusion.tex
PHOBIC introduces optimized bucket sizes and interleaved encoding to PTHash.
Our improvements result in \totalBitImprovement{} bits/key better space efficiency when compared to PTHash for similar construction and query speed.
When compared for the same space consumption, PHOBIC can be constructed up to \speedupPtHashVsPhobic{} times faster than PTHash, while still having the same query time.
Finally, our GPU implementation improves the construction by a factor of up to \speedupPhobicVsGpu{} compared to the multithreaded CPU implementation.

Future work may explore combinations of the most time efficient approaches to perfect hashing and the most space efficient approaches. Concretely, we are hopeful that a hybrid between PHOBIC and ShockHash \cite{lehmann2023shockhash} puts further trade-offs between space and time into reach.

%% file: sections/bucketmapping-proofs.tex
We begin with a more formal statement of what we need from \cref{int:derivative-and-bucket-sizes}.
Let $\mu_i := n·w_i$ be the expected size of bucket $i$ (i.e.\ the expected number of keys assigned to bucket $i$) and again $\lambda := n/B$.

\begin{observation}
	\label{obs:derivative-and-bucket-sizes}
	Let $γ:[0,1]→[0,1]$ be a continuous bucket assignment function that is smooth on $(0,1)$ with non-decreasing derivative, $x₀ ∈ (0,1)$ a hash and $i = ⌈B·γ(x₀)⌉$ the bucket assigned to $x$. Let $\mu = \lambda/γ'(x₀)$. Then the expected bucket sizes satisfy
	\begin{itemize}
		• $\mu_j ≥ \mu$ for all $j < i$,
		• $\mu_j ≤ \mu$ for all $j > i$,
	\end{itemize}
\end{observation}
\begin{proof}
	Consider $j < i$ and the range $(σ_{j-1},σ_j]$ assigned to bucket $j$. Since $x₀$ is assigned to bucket $i$ we have $σ_j < x₀$. By monotonicity $γ'(x) < γ'(x₀)$ for all $x ∈ [σ_{j-1},σ_j]$. Since $γ(σ_j)-γ(σ_{j-1}) = 1/B$ by construction we have
	\[ \frac 1B
	= \int_{σ_{j-1}}^{σ_j} γ'(x)dx
	≤ \int_{σ_{j-1}}^{σ_j} γ'(x₀)dx
	= (σ_j-σ_{j-1})·γ'(x₀)
	= w_j · \frac{\lambda}{\mu}
	= \frac{\mu_j}{n} \frac{\lambda}{\mu}
	= \frac{\mu_j}{\mu} \frac{1}{B}.
	\]
	Rearranging gives $\mu_j ≥ \mu$. The second claim is obtained analogously.
\end{proof}

\subsection{Proof of Theorem \ref{thm:not-beta-is-bad}}
\label{sec:not-beta-is-bad}

We assume the context of \cref{thm:not-beta-is-bad}. 
In particular, $γ$ is a monotonic function, smooth on $(0,1)$ and satisfies w.l.o.g.\ $γ(0) = β_*(0) = 0$ and $γ(1) = β_*(1) = 1$. We have also assumed that $γ(x) > 0$ for $x > 0$, but never defended this assumption. Let us deal with this rather silly case now, i.e.\ assume $γ(x) = 0$ for some $x > 0$. Then by monotonicity $γ(x') = 0$ for all $x' ∈ [0,x]$. The first bucket then receives at least $nx$ keys in expectation and $Ω(nx)$ keys with high probability. A simple argument shows that if $n ≥ 1/x³$ then the expected number of seeds that need to be tried for the first bucket is $\exp(Ω(n^{1/3}))$. This vastly exceeds the lower bound we wish to prove even for $ε = 1$.

It was therefore w.l.o.g.\ that we assumed $γ(x) > 0$ for all $x ∈ (0,1]$. Since $γ'$ is non-decreasing this implies $γ'(x) > 0$ for all $x ∈ (0,1)$.

We now use the assumption that $γ ≠ β_*$. Let $x ∈ (0,1)$ with $γ(x) ≠ β_*(x)$. If $γ(x) > β_*(x)$ then there is some $y ∈ [x,1)$ with $γ'(y) < β_*'(y)$ and if $γ(x) < β_*(x)$ then there is some $y ∈(0,x]$ with $γ'(y) < β_*'(y)$. In both cases we have $1/γ'(y) > 1/β_*'(y)$. Because $γ'$ and $β_*'$ are continuous on $(0,1)$ we can find $ε > 0$ and $z ∈ (0,1-3ε)$ such that
\[
1/γ'(z+3ε) > (1+3ε)/β_*'(z).
\]

Let $i$ be the bucket that $γ$ assigns to hash $z+3ε$. Then all hashes in $[0,z+3ε]$ are assigned to buckets with index at most $i$. Using a concentration bound and assuming large $n$, the number of keys assigned to buckets with index less than $i$ is at least $n·(z+2ε)$. By \cref{obs:derivative-and-bucket-sizes} all these buckets have expected size at least $\lambda/γ'(z+3ε)$, hence expected size at least $\lambda(1+3ε)/β_*'(z)$. For large $\lambda$, most will have actual size at least $s := \lambda(1+2ε)/β_*'(z)$. Again by a concentration bound, the number of keys in buckets of actual size at least $s$ is at least $n·(z+ε)$. 
In particular, at least $εn$ keys are in buckets of size at least $s$ that are processed when the load factor is already at least $z$. To get a lower bound on the work this causes, we may assume that these $εn$ keys are in buckets of size exactly $s$ and processed at load factor exactly $z$.
By \cref{eq:bucket-cost-bound} the expected cost of each such bucket is then lower bounded by
\[ (1-z)^{-s} = (1-z)^{-\lambda(1+2ε)/β_*'(z)} = (1-z)^{\lambda(1+2ε)/\ln(1-z)} = e^{\lambda(1+2ε)}.\]
The last step uses that $c^{1/\ln(c)} = e$ for all $c > 0$. Multiplying this by the number $εn / s = 𝒪(εn/\lambda)$ of such buckets yields a contribution of $n·𝒪(ε/\lambda)·e^{\lambda(1+2ε)}$. This is at least $n·e^{\lambda(1+ε)}$ if $\lambda ≥ \lambda_0(ε)$ is large enough. Since the buckets of size $1$ contribute a cost of $\wcoupon$ and are distinct from the buckets that contributed to $n·e^{\lambda(1+ε)}$, we obtain the lower bound of $n·e^{\lambda(1+ε)} + \wcoupon$ as claimed.

\subsection{Proof of Theorem \ref{thm:beta-is-good}}
\label{sec:beta-is-good}

In this section we prove \cref{thm:beta-is-good} and assume the corresponding context. In particular, we are given $ε > 0$, which defines a bucket assignment function $β_ε$, we assume $\lambda ≥ \lambda_0(ε)$ is large enough and $n ≥ n₀(\lambda,ε)$ is large enough, which defines $B = n/\lambda$. We begin with a few definitions and give corresponding intuition.
\begin{definition}
	Let $i ∈ [B]$ and $s ∈ ℕ$. We define the following.
	\begin{itemize}
		• $λ_i := n·(β^{-1}_ε(i/B) - β^{-1}_ε((i-1)/B))$ is the expected size of the $i$th bucket. The stated formula involves the length of the range of hashes that $β_ε$ maps to bucket $i$.
		• $s_i$ is the number of keys assigned to bucket $i$. This random variable has distribution $s_i \sim \Bin(n,λ_i)$.
		• $α_s := \tfrac{1}{n}\sum_{i = 1}^B s_i · 𝟙_{s_i ≥ s}$ is the random fraction of keys within buckets of size at least $s$ and therefore the load factor after all buckets of size at most $s$ have been processed.
		• $d_s := 1-e^{-\lambda/s}$ is the \emph{deadline} for bucket size $s$. It is the load factor up to which \cref{eq:bucket-cost-bound} guarantees that processing a bucket of size $s$ incurs an expected cost of at most $s·e^\lambda$.
	\end{itemize}
\end{definition}
The key lemma of this section proves a suitably weakened variant of the claim “$∀s: α_s ≤ d_s$”, i.e.\ that buckets of size $s$ or larger are handled before their deadline $d_s$.
\begin{lemma}[Deadline Lemma]
	\label{lem:deadline}\ 
	\begin{enumerate}[(i)]
		• For $2 ≤ s ≤ \sqrt{\lambda}$: $(\frac{1-α_s}{1-d_s})^s = e^{-𝒪(εA)}$ whp.
		• For $\sqrt{\lambda} < s ≤ 2e²A/ε$: $(\frac{1-α_s}{1-d_s})^s = e^{-𝒪(εA)}$ whp.
		• For $2e²A/ε < s < 2\frac{\log n}{\log \log n}$: $α_s ≤ d_s$ whp.
		• For $s = 2\frac{\log n}{\log \log n}$: $α_s = 0$ whp.
	\end{enumerate}
\end{lemma}
Each case is handled in a separate paragraph in \cref{sec:deadline-lemma}. We now check that the Deadline Lemma implies \cref{thm:beta-is-good}.
\begin{proof}[Proof of \cref{thm:beta-is-good}.]
	We apply \cref{lem:deadline} for each $2 ≤ s ≤ 2\frac{\log n}{\log \log n}$. Since each corresponding event holds whp (i.e.\ with probability $1-𝒪(n^{-c})$ for some $c > 0$) and since there are $𝒪(\frac{\log n}{\log \log n})$ events, they \emph{jointly} hold whp. We may assume that this is the case.

	We have to bound the expected cost for handling the buckets by $n·e^{\lambda(1+𝒪(ε)} + \wcoupon$. Consider therefore any bucket $b$ with actual size $s$. By \cref{lem:deadline}, we know $s ≤ 2\frac{\log n}{\log \log n}$. If $s = 1$ then the work for $b$ is accounted for by $\wcoupon$. So assume $s ∈ \{2,…,2\frac{\log n}{\log \log n}\}$. After all buckets of size at least $s$ are handled, the load factor is $α_s$. In particular, after $b$ is handled, the load factor is at most $α_s$. From \cref{eq:bucket-cost-bound} we get that the expected cost for handling $b$ is at most $s·(1-α_s)^{-s}$. If $s$ falls into case (iii) of \cref{lem:deadline} then we can bound this as follows:
	\[s(1-α_s)^{-s} ≤ s(1-d_s)^{-s} = s(e^{-\lambda/s})^{-s} = se^\lambda. \]
	If $s$ falls into case \lipicsLabel{(i)} or \lipicsLabel{(ii)} then we get
	\begin{align*}
		s(1-α_s)^{-s} &= s(1-d_s)^{-s} \Big(\frac{1-α_s}{1-d_s}\Big)^{-s} = se^\lambda·e^{𝒪(εA)} = se^{(1+𝒪(ε))\lambda}.
	\end{align*}
	Summing this over all buckets of size at least $2$, the sizes of which sum to at most $n$, we get an overall expected cost of at most $n·e^{(1+𝒪(ε))\lambda}$, which is the main term of our bound.
\end{proof}

\subsubsection{Proof of the Deadline Lemma}
\label{sec:deadline-lemma}

We begin with observations that relate to several of the cases \lipicsLabel{(i)},\lipicsLabel{(ii)}, \lipicsLabel{(iii)} and \lipicsLabel{(iv)}. To sharpen up presentation we will occasionally use $≈$, $\gtrsim$ and $\lesssim$ when a relation only holds whp and/or we suppress an error term that is clearly negligible in a given context when $n$ is chosen large enough. This notation is not meant to indicate a gap in the proof.

We have previously encountered the random variable $α_s$, the fraction of keys in buckets of size at least $s$. This is not to be confused with the following non-random quantities.
\begin{definition} Let $\mu > 0$ and $s ∈ ℕ$. We define
	\begin{itemize}
		• $x_λ := \tfrac{1}{n}\sum_{i = 1}^B λ_i · 𝟙_{λ_i ≥ \mu} ∈ [0,1]$, the range of hashes assigned to buckets of \emph{expected} size at least $\mu$. Equivalently, the probability that a key is placed in a bucket of \emph{expected} size least $\mu$.
		• $\hx_λ = 1-\exp(-\frac{\lambda/\mu-ε}{1-ε})$, the unique number satisfying $\lambda/β'_ε(\hx_λ) = \mu$ if $\mu ≤ \lambda/ε$ and $\hx_λ = 0$ otherwise. Equivalently this is the length of the range $[0,\hx_λ]$ where $\lambda/β'_ε(x)$ attains values of at least $\mu$.
		• $\ha_s := 𝔼[α_s]$, the \emph{expected} fraction of keys in buckets of size at least $s$.
	\end{itemize}
\end{definition}
\begin{observation}
	\label{obs:size-bound}
	$\max_{i ∈ [B]} λ_i ≤ \lambda/ε$.
\end{observation}
\begin{proof}
	This follows because $β'_ε(0) ≥ ε$. See also \cref{obs:derivative-and-bucket-sizes}.
\end{proof}

\begin{observation}
	For any $\mu > 0$ we have $|\hx_λ - x_λ| ≤ \frac{\lambda}{εn}$, hence in most contexts $\hx_λ ≈ x_λ$.
\end{observation}
\begin{proof}
	Let $i_λ$ be the bucket assigned to hash $\hx_λ$. By \cref{obs:derivative-and-bucket-sizes} the buckets of expected size at least $\mu$ are precisely the buckets preceding bucket $i_λ$ and, possibly, bucket $i_λ$ itself. Let $\mu'$ be the expected size of $i_λ$. If $\mu' ≥ \mu$ then the entire range $[0,\hx_λ]$ of hashes is assigned to buckets of size at least $\mu$, hence $x_λ ≥ \hx_λ$. If $\mu' ≤ \mu$ then only a subset of the range $[0,\hx_λ]$ of hashes is assigned to buckets of size at least $\mu$, hence $x_λ ≤ \hx_λ$. The difference between the two cases is the $\mu'/n$, the length of the range assigned to bucket $i_λ$. Since $\mu' ≤ \lambda/ε$ by \cref{obs:size-bound}, the claim follows.
\end{proof}
The quantities $α_s$ and $\ha_s$ are probabilistically related as follows.
\begin{lemma}
	\label{lem:keys-bounded-differences}
	For any $s ∈ ℕ$:
	$\Pr[n·α_s - n·\ha_s ≥ δ] ≤ \exp(\frac{-2δ²}{ns²})$.
\end{lemma}
\begin{proof}
	We can directly apply the method of bounded differences \cite{D:McDiarmid:1989}. The $n$ hash values of keys are independent variables that determine $n·α_s$, the number of keys in buckets of size at least $s$. The expectation of $n·α_s$ is $n·\ha_s$ by definition. Changing a single hash can affect $n·α_s$ by at most $\pm s$ with “$+s$” corresponding to moving a key from a bucket of size less than $s$ to a bucket of size $s-1$ and “$-s$” to the reverse change.
\end{proof}
By choosing $δ = n^{2/3}$ we get
\begin{corollary}
	\label{cor:keys-bounded-differences}
	For $n ≥ n₀(s)$ large enough we have $|\ha_s - α_s| ≤ n^{-1/3}$ whp, hence in most contexts $\ha_s ≈ α_s$.
\end{corollary}
\begin{proof}
	Apply \cref{lem:keys-bounded-differences} with $δ = n^{2/3}$ and assume that $n ≥ s^{12}$.
\end{proof}

Two further claims do not relate to our setting in particular.
\begin{claim}
	\label{obs:poisson-mean}
	There exists a constant $c₀ > 0$  such that for any $\mu ≥ 3$ and $X \sim \Po(\mu)$ we have $\Pr[X ≤ \mu-3] ≥ c₀$.
\end{claim}
\begin{proof}
	The median $\med(\mu)$ of $\Po(\mu)$ satisfies $\med(\mu) ∈(\mu-1,\mu+1)$ \cite{C:Poisson-Median:1994}. For large $\mu$, no single outcome has high probability. This implies for $X \sim \Po(\mu)$:
	\begin{align*}
		\Pr[X ≤ \mu-3] &≥ \Pr[X ≤ \med(\mu) ∧ |\med(\mu) - X| ≥ 4]\\
		&≥ \Pr[X ≤ \med(\mu)] - \Pr[|\med(\mu) - X| ≤ 4] = 1/2 - o_{\mu → ∞}(1).
	\end{align*}
	In other words, $\Pr[X ≤ \mu-3]$ is bounded away from $0$ for large $\mu$. For small $\mu ≥ 3$ clearly $\Pr[X ≤ \mu-3] > 0$. Hence, the desired $c₀ > 0$ exists.
\end{proof}

\begin{claim}
	\label{claim:custom-poisson-tail}
	Let $ε ∈ (0,1)$, $\mu ≥ \frac{1}{2πε²}$, $X \sim \Po(\mu)$ and $s = \mu/(1-ε)$. Then
	\[\Pr[X ≥ s] ≤ (e^ε(1-ε))^s.\]
	Note that $e^ε(1-ε) < 1$, which follows from $1-x < e^{-x}$ for $x ∈ ℝ \setminus \{0\}$.
\end{claim}
\begin{proof}
	Consider the sum $\Pr[X ≥ s] = e^{-\mu}\sum_{i ≥ s} \frac{\mu^i}{i!}$.
	The ratio between subsequent terms is $\mu/i ≤ \mu/s = 1-ε$. Hence, we can upper bound the sum by its first term and a geometric sum as follows:
	\[ \Pr[X ≥ s] ≤ e^{-\mu}\frac{\mu^s}{s!} ·\sum_{i ≥ 0}(1-ε)^i = e^{-\mu}\frac{\mu^s}{s!}·\frac{1}{ε}. \]
	Using Stirling's approximation for $s!$ and using $s ≥ \mu ≥ \frac{1}{2πε²}$ gives:
	\begin{gather*}
		\Pr[X ≥ s] ≤ e^{-\mu}\frac{\mu^se^s}{s^s\sqrt{2πs}}·\frac{1}{ε}
		≤ e^{-\mu}\frac{\mu^se^s}{s^s}
		= e^{-(1-ε)s}\frac{((1-ε)s)^se^s}{s^s}
		= (e^{ε}(1-ε))^{s}.\qedhere
	\end{gather*}
\end{proof}

We will now consider the cases of \cref{sec:deadline-lemma} one after the other.
\myparagraph{(i) Buckets of size $2 ≤ s ≤ \lambda^{1/3}$}
We have
\begin{align*}
	1 - x_s ≈ 1 - \hx_s = \exp(-\tfrac{\lambda/s-ε}{1-ε})
	\qquad \text{and} \qquad
	1 - x_{s+1} ≈ 1 - \hx_{s+1} = \exp(-\tfrac{\lambda/(s+1)-ε}{1-ε})
\end{align*}
Using $s ≤ \lambda^{1/3}$ we find $\lambda/s-\lambda/(s+1) = \frac{\lambda}{s(s+1)} = Ω(\lambda^{1/3})$. This implies $\frac{\lambda/s-ε}{1-ε} - \frac{\lambda/(s+1)-ε}{1-ε} = Ω(\lambda^{1/3})$ and for $\lambda$ large enough we have $\exp(-\frac{\lambda/(s+1)-ε}{1-ε}) ≥ (1 + \frac{1}{c₀})\exp(-\frac{\lambda/s-ε}{1-ε})$ where $c₀$ is the constant from \cref{obs:poisson-mean}. Hence
\begin{align*}
	x_s - x_{s+1} &≈ \hx_s - \hx_{s+1}
	= \exp(-\tfrac{\lambda/(s+1)-ε}{1-ε}) - \exp(-\tfrac{\lambda/s-ε}{1-ε})\\
	&≥ \tfrac{1}{c₀}\exp(-\tfrac{\lambda/s-ε}{1-ε}) ≥ \tfrac{1}{c₀}\exp(-\tfrac{\lambda/s}{1-ε}).
\end{align*}
We now bound the probability $1-\hat{α}_s$ that a key $k$ ends up in a bucket of size at most $s-1$. Let $x \sim 𝒰([0,1])$ be its hash, $\mu ∈ (0,\lambda/ε)$ the \emph{random variable} denoting the \emph{expected} size of the bucket that $k$ ends up in, and $s'$ the number of other keys sharing the bucket with~$k$. Conditioned on $\mu$, we have $s' \sim \Bin(n-1,\mu/n)$, a distribution that is well known to be approximately $\Po(\mu)$. In the following computation we use that $\Pr[s' ≤ s-2 \mid \mu]$ is decreasing in $\mu$, and we use that $\Pr[\mu ∈ [s,s+1)] = \Pr[x ∈ (x_{s+1},x_s]] = x_s - x_{s+1}$.
\begin{align*}
	1-\hat{α}_s &= \Pr[s' ≤ s-2] ≥ \Pr[s' ≤ s-2 ∧ \mu ∈ [s,s+1)]\\
	&= \Pr[s' ≤ s-2  \mid \mu ∈ [s,s+1)]·\Pr[\mu ∈ [s,s+1)]\\
	&≥ {\Pr}_{s'' \sim \Bin(n-1,\frac{s+1}{n})}[s'' ≤ s-2]·(x_s-x_{s+1})\\
	&≈ {\Pr}_{s'' \sim \Po(s+1)}[s'' ≤ s-2]·(x_s-x_{s+1})
	\refRel{Obs.}{obs:poisson-mean}{≥} c₀(x_s-x_{s+1}) \gtrsim \exp(-\tfrac{\lambda/s}{1-ε}).
\end{align*}

This gives use the bound we desired, closing this case:
\begin{align*}
	\Big(\frac{1-α_s}{1-d_s}\Big)^s
	≈ \Big(\frac{1-\hat{α}_s}{1-d_s}\Big)^s
	≥ \frac{\exp(-\frac{\lambda/s}{1-ε})^s}{\exp(-\lambda/s)^s}
	= \frac{\exp(-\frac{\lambda}{1-ε})}{\exp(-\lambda)} = e^{-\lambda/(1-ε)+\lambda} = e^{-𝒪(εA)}.
\end{align*}

\myparagraph{(ii) Buckets of size $\lambda^{1/3} ≤ s ≤ 2e²A/ε$}
We proceed in a similar way as in case \lipicsLabel{(i)} to bound the probability $\ha_s$ that a key $k$ ends up in a bucket $b$ of size at least $s$. Let again $\mu ∈ (0,\lambda/ε)$ be a random variable denoting the \emph{expected} size of $b$ and $s'$ the number of other keys in $b$. We consider two overlapping events that cover all cases where $b$ has size at least $s$. On the one hand, $b$ might have large expected size. On the other hand, $b$ might have small expected size and still have \emph{actual} size at least $s$.
\begin{align*}
	\ha_s &= \Pr[s' ≥ s-1] ≤ \Pr[\mu ≥ (1-ε)s ∨ (\mu < (1-ε)s ∧ s' ≥ s-1)]\\
	&≤ \Pr[\mu ≥ (1-ε)s] + \Pr[\mu < (1-ε)s]·\Pr[s' ≥ s-1 \mid \mu < (1-ε)s]\\
	&≤ x_{(1-ε)s} + (1-x_{(1-ε)s})·{\Pr}_{s'' \sim \Bin(n-1,\frac{(1-ε)s}{n})}[s'' ≥ s-1]\\
	&\gtrsim x_{(1-ε)s} + (1-x_{(1-ε)s})·{\Pr}_{s'' \sim \Po((1-ε)s)}[s'' ≥ s-1]\\
	&≤ x_{(1-ε)s} + (1-x_{(1-ε)s})·2((1-ε)e^ε)^s
\end{align*}
The last step uses \cref{claim:custom-poisson-tail}, very conservatively accounting for a “$-1$” discrepancy with a factor of $2$.

We now make another minor case distinction, first assuming that $(1-ε)s ≤ \lambda/ε$. In that case $1-x_{(1-ε)s} = \exp(-\frac{\lambda/(s·(1-ε))-ε}{1-ε}) ≥ \exp(-\frac{\lambda/s}{(1-ε)²})$.
We can now turn to $1-\ha_s$ and quotient with $1-d_s$. We assume $ε$ is small enough such that $1/(1-ε)² ≤ 1+3ε$ and that $\lambda ≥ \lambda_0(ε)$ is large enough such that $2s((1-ε)e^ε)^s ≤ ε$ for all $s ≥ \lambda^{1/3}$.

\begin{align*}
	&& 1-\ha_s &≥ 1 - x_{(1-ε)s} - (1-x_{(1-ε)s})·2((1-ε)e^ε)^s\\
	&&&= (1 - x_{(1-ε)s})(1-2((1-ε)e^ε)^s)\\
	&⇒& \frac{1-\ha_s}{1-d_s}
	&≥ \frac{(1 - x_{(1-ε)s})}{1-d_s} (1-2((1-ε)e^ε)^s)
	 ≥ \frac{\exp(-\frac{\lambda/s}{(1-ε)²})}{\exp(-\lambda/s)} (1-2((1-ε)e^ε)^s)\\
	&&&≥ \exp(-3εA/s)(1-2((1-ε)e^ε)^s)\\
	&⇒& \Big(\frac{1-\ha_s}{1-d_s}\Big)^s& ≥
	\exp(-3εA)(1-2((1-ε)e^ε)^s)^s\\
	&&&≥ \exp(-3εA)(1-2s((1-ε)e^ε)^s) ≥ \exp(-3εA)·(1-ε) ≥ \exp(-4εA).
\end{align*}
The last step again assumes that $\lambda$ is large enough.

We still have to consider the case where $(1-ε)s > \lambda/ε$. By \cref{obs:size-bound} there are no buckets of expected size $(1-ε)s$ or larger, hence $x_{(1-ε)s} = 0$. This gives $1-\ha_s ≥ 1-2((1-ε)e^ε)^s$, meaning the above derivation only involves the less critical term that comes out as $1-ε ≥ \exp(-εA)$.

\myparagraph{(iii) Buckets of size $2e²A/ε < s ≤ 2\frac{\log n}{\log \log n}$}
Our use of $β_ε$ rather than $β_*$ renders this case quite easy. However, now that the bound on $s$ grows with $n$, shortcuts involving “for $n$ large enough” and the associated notation $≈$ would now be suspect and have to be replaced with careful arguments.

We wish to bound $\hat{α}_s$, the probability that at least $s-1$ further keys join a given key in its bucket. Any bucket has expected size at most $\lambda/ε$ by \cref{obs:size-bound}. We can hence argue
\begin{align*}
	\hat{α}_s ≤ \binom{n-1}{s-1} \Big(\frac{\lambda/ε}{n}\Big)^{s-1} ≤ \Big(\frac{ne}{s-1}\Big)^{s-1}\Big(\frac{\lambda/ε}{n}\Big)^{s-1}
	= \Big(\frac{eA/ε}{s-1}\Big)^{s-1} ≤ e^{-s+1}
\end{align*}
where the last step used $s ≥ e²A/ε+1$.
Applying \cref{lem:keys-bounded-differences} with $δ = n·e^{-s}$ gives $\Pr[n·α_s - n·\ha_s ≥ n·e^{-s}] ≤ \exp(-ne^{-2s}/(2s²))$. For $s ≤ 2\log n / \log \log n$ this probability is negligible. In particular, we have whp
\[ α_s = (α_s - \ha_s) + \ha_s \stackrel{\text{whp}}{≤} e^{-s} + \ha_s ≤ e^{-s} + e^{-s+1} ≤ e^{-s+2}.\]
In particular, $α_s$ is \emph{much} smaller than the deadline $d_s$, we can argue for instance by using that $1-e^{-x} ≥ x/2$ for $x ∈ [0,\frac 12]$ to get:
\[ d_s = 1-e^{-\lambda/s} ≥ \lambda/(2s) ≥ 1/s ≥ e^{-s+2} \stackrel{\text{whp}}{≥} α_s.\]

\myparagraph{(iv) Buckets of size at least $s = 2\frac{\log n}{\log \log n}$}
Again, the largest expected bucket size is at most $\lambda/ε$ by \cref{obs:derivative-and-bucket-sizes} and the probability that a key hashes to a specific bucket hence at most $\frac{\lambda/ε}{n}$. The probability $p_{≥s}$ that the bucket has size at least $s$ is therefore bounded by
\[ p_{≥s} ≤ \binom{n}{s} \Big(\frac{\lambda/ε}{n}\Big)^s ≤ \Big(\frac{ne}{s}\Big)^s·\Big(\frac{eA/ε}{n}\Big)^s = \Big(\frac{eA/ε}{s}\Big)^s = 2^{s·(\log(eA/ε)-\log s)}. \]
Plugging in $s = 2\frac{\log n}{\log \log n}$ gives $p_{≥s} = n^{-2+o(1)}$.  By a union bound, the probability that at least one bucket has size at least $s$ is at most $B·p_{≥s} ≤ n^{-1+o(1)}$. Hence, whp no such bucket exists and we have $α_s = 0$ as claimed.

%% file: sections/primary_bucket_ordering.tex
\label{app:primary-bucket-ordering}

In this section we prove \cref{prop:largest-to-smallest}, restated here for convenience.
\begin{proposition}
    \label{prop:largest-to-smallest}
    To minimize expected construction time in a PTHash context, buckets should be processed in order from largest to smallest.
\end{proposition}
Even though the claim is very intuitive and CHD \cite{belazzougui2009hash} and PTHash \cite{pibiri2021pthash} implicitly assume its truth, the argument is surprisingly subtle and involves the following random process.

Let $k ∈ ℕ$, $p₁,…,pₖ ∈ (0,1)$ and let $\{0,…,k\}$ be a set of \emph{states}. When taking a \emph{step} in state $0 ≤ i < k$, the successor state is state $i+1$ with probability $p_{i+1}$ and state $0$ with probability $1-p_{i+1}$. Let $w(p₁,…,p_k)$ be the expected number steps needed to reach state $k$ from state $0$.

\begin{lemma}
    \label{lem:juggler}
    Assume $1 > p₁ > p₂ > … > pₖ > 0$ and $1 ≤ i < \frac{k}{2}$. Then
    \[
        w(p₁,…,p_{k-i}) + w(p_{k-i+1},…,pₖ) < w(p₁,…,p_i) + w(p_{i+1},…,pₖ).
    \]
\end{lemma}

Before proving \cref{lem:juggler}, let us check that it implies \cref{prop:largest-to-smallest}.

\begin{proof}[Proof of \cref{prop:largest-to-smallest}.]
    Assume that the sequence of bucket sizes in processing order is $b₁,…,b_B$ and that $b_i < b_{i+1}$ for some $1 ≤ i < B$. We will show that by switching the order of these two buckets, the expected work (number of hash function evaluations) decreases. For this, let $ℓ = \sum_{j = 1}^{i-1} b_j$ be the number of keys that are handled before the $i$th bucket. The work for buckets $i$ and $i+1$ is then
    \[w(\tfrac{n-ℓ}{n},\tfrac{n-ℓ-1}{n},…,\tfrac{n-ℓ-b_i+1}{n}) + w(\tfrac{n-ℓ-b_i}{n},\tfrac{n-ℓ-b_i-1}{n},…,\tfrac{n-ℓ-b_i-b_{i+1}+1}{n}).\]
    When processing the two buckets in swapped order the expected work becomes
    \[w(\tfrac{n-ℓ}{n},\tfrac{n-ℓ-1}{n},…,\tfrac{n-ℓ-b_{i+1}+1}{n}) + w(\tfrac{n-ℓ-b_{i+1}}{n},\tfrac{n-ℓ-b_{i+1}-1}{n},…,\tfrac{n-ℓ-b_i-b_{i+1}+1}{n}).\]
    By \cref{lem:juggler} this is less. Note that the expected work needed for other buckets is unchanged because their sizes and the table load when they are processed is unchanged.
\end{proof}

\myparagraph{Intuition for \cref{lem:juggler}} Assume a juggler knows, for some $a,b∈ℕ$, a routine $S$ of $a$ simple throws, a routine $D$ of $a$ difficult throws, and a routine $M$ of $b$ throws of intermediate difficulty. Performing a routine means attempting all throws in sequence, starting over after the first unsuccessful throw, and repeating until all throws succeed. Let $t(R)$ denote the expected number of throws when performing a routine $R$.
It is intuitively plausible that
\[
    t(S∘M)+t(D) < t(S) + t(D∘M) < t(S)+t(M∘D)
\]
where “$∘$” denotes concatenation of routines. The first inequality means that appending $M$ to a difficult routine rather than the simple routine makes things more costly overall. The second inequality means that delaying difficult throws until the end of a routine increases its cost because failures tend to happen \emph{after} the easier throws have already taken place.

Proof of \cref{lem:juggler} formalizes these insights.

\myparagraph{Simple Observations} We make four observations about the function $w(p₁,…,pₖ)$, each time justifying them briefly after stating them. The first serves as an alternative definition of $w(p₁,…,pₖ)$.
\begin{equation}
    w(p₁,…,pₖ) = \sum_{i = 1}^{k} \frac{1}{p_i·…·p_k}.
    \label{obs:w-formula}
\end{equation}
This holds because the number of visits to state $i$ for $1 ≤ i < k$ has a geometric distribution with parameter $p_i·…·p_k$ and hence has expectation $\frac{1}{p_i·…·p_k}$.
\begin{equation}
    w(p₁,…,pₖ) \text{is monotonically decreasing in all of its parameters}.
    \label{obs:w-monotonic}
\end{equation}
This follows from \cref{obs:w-formula} because each summand is non-increasing in all parameters and the first summand is decreasing in all parameters.
\begin{equation}
    \text{if $p_j < p_{j+1}$ for $1 ≤ j < k$ then } w(p₁,…,pₖ) < w(p₁,…,p_{j-1},p_{j+1},p_{j},p_{j+2},…,pₖ).
    \label{obs:w-swapping-adjacent}
\end{equation}
In other words, when swapping two adjacent parameters (at indices $j$ and $j+1$), then having the larger parameter further to the left yields a larger value of $w$. This follows from \cref{obs:w-formula} because the summand with $i = j+1$ is then larger, and all other summands are the same.
\begin{gather}
    \text{$w(p₁,…,pₖ)$ is maximized if the parameters appear in non-increasing order,}\label{obs:w-monotonic-is-extremal}\\
    \text{and any other ordering of the same set of parameters yields a smaller value.}\notag
\end{gather}
This follows by iterating \cref{obs:w-swapping-adjacent} until the parameters are sorted in non-increasing order.

\begin{equation}
    \text{for any $1 ≤ j < k$: } w(p₁,…,pₖ) = \tfrac{1}{p_{j+1}·…·pₖ}w(p₁,…,p_j)+ w(p_{j+1},…,pₖ).
    \label{obs:decomposing}
\end{equation}
This follows from \cref{obs:w-formula} by separating the sum into the summands for $i > j$ and those for $i ≤ j$.
We are now ready to proof \cref{lem:juggler}.

\begin{proof}[\cref{lem:juggler}.]
    Let $1 > p₁ > p₂ > … > pₖ > 0$ and $1 ≤ i < \frac{k}{2}$. Then
    \begin{align*}
        w(p₁,…,&p_{k-i}) + w(p_{k-i+1},…,pₖ)\\
        &\refRelObs{obs:decomposing}{=}
        \tfrac{1}{p_{i+1}·…·p_{k-i}}w(p₁,…,p_i)+ w(p_{i+1},…,p_{k-i}) + w(p_{k-i+1},…,pₖ)\\
        &=\quad w(p₁,…,p_i) + \big(\tfrac{1}{p_{i+1}·…·p_{k-i}}-1\big)w(p₁,…,p_i)\\
        &\phantom{spaaaaaace} + w(p_{i+1},…,p_{k-i}) + w(p_{k-i+1},…,pₖ)\\
        &\refRelObs{obs:w-monotonic}{<}
        w(p₁,…,p_i) + \big(\tfrac{1}{p_{i+1}·…·p_{k-i}}-1\big)w(p_{k-i+1},…,pₖ)\\
        &\phantom{spaaaaaace} + w(p_{i+1},…,p_{k-i})+ w(p_{k-i+1},…,pₖ)\\
        &=\quad w(p₁,…,p_i) + \tfrac{1}{p_{i+1}·…·p_{k-i}} w(p_{k-i+1},…,pₖ)  + w(p_{i+1},…,p_{k-i})\\
        &\refRelObs{obs:decomposing}{=}
        w(p₁,…,p_i) + w(p_{k-i+1},…,pₖ,p_{i+1},…,p_{k-i})\\
        &\refRelObs{obs:w-monotonic-is-extremal}{<} w(p₁,…,p_i) + w(p_{i+1},…,pₖ).\qedhere
    \end{align*}
\end{proof}

%% file: sections/bucket-placement-bounds.tex
\label{sec:bucket-placement-bounds}

In this section we consider perfect hashing through bucket placement in general and without any partitioning (see introduction). 
\begin{proposition}
	\label{prop:bucket-placement-general-bounds}
	Consider perfect hashing through bucket placement with $n$ keys and $B := n/λ$ buckets for some $2 ≤ λ = o(n/\log n)$. Let $T$ be the sum of resulting seeds values and $S$ the number of bits when encoding the seeds using Elias $δ$-encoding. Then
	\begin{enumerate}[(i)]
		• $𝔼[T] = Ω(ne^{λ}/λ)$
		• $𝔼[S] = n(\log₂ e +𝒪(\frac{1}{λ}\log λ))$.
	\end{enumerate}
\end{proposition}
Note: \lipicsLabel{(i)} implies that the expected construction time per key of $e^{λ(1+ε)}$ achieved using $β_ε$ in \cref{s:bucketSizes} (when ignoring buckets of size $1$) is almost optimal. \lipicsLabel{(ii)} means that the space consumption of perfect hashing through bucket placement approaches the lower bound $n\log₂(e)$ for large $λ$, in the sense that it is $n(\log₂ e+o_{λ→∞}(1))$, \emph{regardless} of the bucket assignment function that is used.

\begin{proof}
	The first step of perfect hashing through bucket placement is to map the $n$ keys into $B$ buckets. The sizes $s₁,…,s_B$ of these buckets are random variables. However, this arguments makes no use of their distributions and works for arbitrary bucket sizes. We focus on the second step where the buckets are placed one after the other using brute force. Assume the buckets are indexed in placement order. The $i$th bucket is then associated with the probability $p_i$ that randomly hashing a set of $s_i$ keys into a table does not cause any collision assuming $s₁+…+s_{i-1}$ out of $n$ positions are already occupied.
	
	A useful observation is that the product of all $p_i$ is the probability that a random function of $n$ keys to $n$ positions is a bijection, hence
	\[ \prod_{i=1}^B p_i = \frac{n!}{n^n} \text{ or equivalently } \prod_{i=1}^B 1/p_i = \frac{n^n}{n!}. \]
	If the $i$th seed value is $σ_i$ then its distribution is $σ_i \sim \mathrm{Geom}(p_i)$ implying $𝔼[σ_i] = 1/p_i$. Hence
	\[
		𝔼[T] = 𝔼\big[\sum_{i = 1}^B σ_i\big] = \sum_{i = 1}^B 𝔼[σ_i] = \sum_{i = 1}^B 1/p_i ≥ B·\Big(\frac{n^n}{n!}\Big)^{1/B}.
	\]
	The last inequality uses that the sum of the side lengths of a $B$-dimensional cuboid with volume $V = \frac{n^n}{n!}$ is minimized by the $B$-dimensional cube where all $B$ side lengths are equal to $V^{1/B}$. Continuing with Stirling's approximation yields
	\[
		𝔼[T] ≥ B·\Big(\frac{n^n}{n!}\Big)^{1/B} = B\Big(\frac{e^n}{𝒪(\sqrt{n})}\Big)^{1/B} = \tfrac{n}{λ}e^{λ}n^{-𝒪(λ/n)} = Ω(ne^{λ}/λ),
	\]
	where the last step used $λ = o(n/\log n)$. This proves \lipicsLabel{(i)}.
	
	To encode a number $x ∈ ℕ$, Elias $δ$-coding requires $⌊\log₂(x)⌋+2⌊\log(⌊\log₂(x)⌋+1)⌋+1 ≤ \log₂(x) + 2\log(2+\log x)+1$ bits. Using that $𝔼[\log₂ X] ≤ \log₂ 𝔼[X]$ for any non-negative random variable $X$ by Jensen's inequality allows us to bound
	\begin{align*}
		𝔼[S] &= 𝔼\Big[\sum_{i = 1}^B \big(\log₂(σ_i) + 2\log₂(2+\log₂ σ_i)+1\big)\Big]\\
		&≤ \sum_{i = 1}^B \log₂ 𝔼[σ_i] + 2\sum_{i = 1}^B \log₂(2+\log₂ 𝔼[σ_i]) + B\\
		&= \sum_{i = 1}^B \log₂(1/p_i) + 2\sum_{i = 1}^B \log₂(2+\log₂ 1/p_i) + B.
	\end{align*}
	This first sum can be bounded by $\sum_{i = 1}^B \log₂(1/p_i) = \log₂ \frac{n^n}{n!} ≤ n\log₂ e$ using Stirling's approximation. The only thing left to prove is that the second sum is subsumed by the error term $𝒪(n\frac{\log₂ λ}{λ})$ in our stated bound. We again using Jensen's inequality and the fact that $\sum_{i = 1}^B \log₂ 1/p_i = \log₂ \frac{n^n}{n!} ≤ n \log₂ e$.
	\begin{align*}
		\sum_{i = 1}^B \log₂(2+\log₂ 1/p_i) &≤ B·\log₂\Big(2+\tfrac{1}{B} \log₂ \tfrac{n^n}{n!} \Big)\\
		& ≤ B·\log₂\big(2+ \tfrac{n}{B} \log₂ e \big)
		  = \tfrac{n}{λ}·\log₂\big(2+λ \log₂ e \big)
		  = 𝒪(n\tfrac{\log₂λ}{λ}).
	\end{align*}
	The last step uses $λ ≥ 2$.
\end{proof}

%% file: sections/non-minimal-bucket-assignment.tex
\section{Bucket assignment for non-minimal PHF}
\label{s:nonminimal}

Assume we wish to construct a \emph{non-minimal} perfect hash function using the “perfect hashing through bucket placement” framework. Hence, assume the range of the hash function is $[m]$ and $n = αm$ for $α ∈ (0,1)$.

In \cref{s:bucketSizes} we only discussed the case $α = 1$. There is, however, a natural way to obtain from a bucket assignment function $γ$ for $α = 1$ a bucket assignment function $γ_α$ for the non-minimal case. Namely, we can define $γ_α(x) = \frac{1}{γ(α)}γ(αx)$, which rescales whatever $γ$ does on its subdomain $[0,α]$. The idea is that, in order to fill a table to load factor $1$, you first need to fill it to load factor $α$, and if $γ$ is optimal overall then it should implicitly contain an optimal strategy for this first phase. Since $γ$ defines on $[0,α]$ how an $α$-fraction of the keys is assigned to buckets of largest expected size (presumably processed first), this is where we should look. This argument remains heuristic because the buckets of largest \emph{expected} size may not turn out to be the buckets of largest \emph{actual} size.

%% file: sections/bucketerimpl.tex
Our implementation of bucket assignment functions differs from our theoretical results in two minor ways.

\myparagraph{Perturbation} Recall that $β_*:[0,1]→[0,1]$ was, modulo certain qualifications, identified as the optimized bucket assignment function in \cref{s:bucketSizes}. One of the qualifications was that we actually analyze a slightly perturbed function $β_ε(x) := εx+(1-ε)β_*(x)$. This limits expected bucket sizes to $λ/ε$, which helped with bounding construction times of large buckets. Limiting the bucket sizes is also useful in practice to reduce self-collisions inside the small partitions. Concretely we choose $ε=\frac{\lambda}{5\sqrt{P}}$, for an expected partition size of $P$. Without this perturbation ($ε=0$), running times are noticeably worse.

\myparagraph{Tabulating values} The functions $β_*$ and $γ_P$ involve expensive arithmetic operations such as a logarithm.
We achieve a significant speedup by tabulating $γ_P$ for $2048$ discrete values of $x$ and interpolating linearly.

%% file: sections/additional_experiments.tex
\label{app:additionalExperiments}
In this section, we give additional experiments that further illustrate the data in \cref{s:experiments}.

\begin{figure}[t]
	\centering
	\begin{subfigure}[b]{0.49\textwidth}
		\input{fig/partitionsStackedPlot.tex}
		\caption{Time for different construction steps when varying the partition size. }
		\label{fig:partitionsStackedPlot}
	\end{subfigure}
	\hfill
	\begin{subfigure}[b]{0.49\textwidth}
		\input{fig/partitionPlot.tex}
		\caption{Space consumption for placing all seeds into a single encoder (mono) and using interleaved coding when varying the expected partition size.}
		\label{fig:partitionPlot}
	\end{subfigure}
	\caption{Varying partition count using the original implementation and $\lambda = 4$ and $n=100$ million keys. We use the original PTHash bucket assignment function and their hash function.}
	\label{fig:partitioning}
\end{figure}
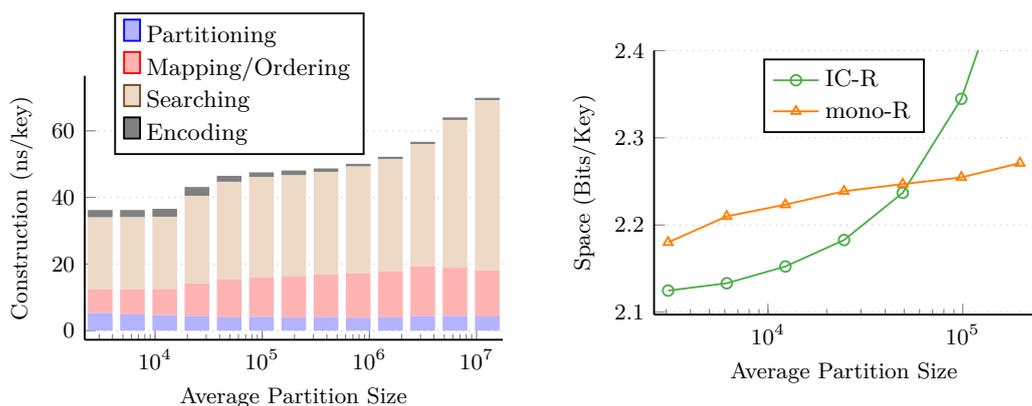

\begin{figure}[t]
	\input{fig/paretoGpu}
	\caption{Comparison between the GPU implementations of PTHash and RecSplit, for $100$ million keys.
		The time is measured for the entire construction, including data transfers to/from the GPU.
	}
	\label{fig:paretoGpu}
\end{figure}
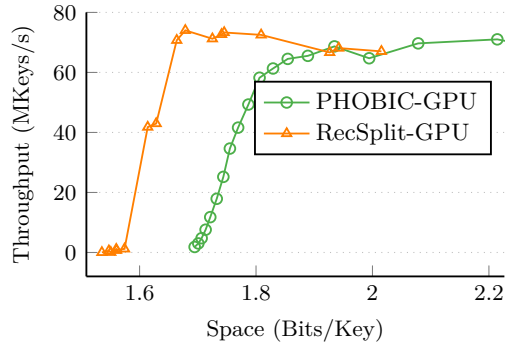

\subsection{From PTHash to PHOBIC}
In \cref{fig:partitioning}, we give additional details on the effect of fine-grained partitioning.
As stated in the main body of the paper, using very small partitions in itself has only a marginal effect on reducing construction time.
The search step does not get much faster when using very small partitions instead of medium sized ones, while there is a small additional overhead in the partitioning step.
\Cref{fig:partitionsStackedPlot} shows this effect.
However, using small partitions enables the using of interleaved coding (see \cref{ss:multi}), whose effect is shown in \cref{fig:partitionPlot}.
For small partitions, we can achieve significant space improvements;
for large partitions instead, we need too many encoders and this leads to
some constant overhead.

\begin{table}[t]
	\caption{Performance of various methods on $100$ million keys on a machine with 64 cores.}
	\label{fig:overviewTableLargeMachine}
	\input{fig/overviewTableLargeMachine}

\end{table}

\subsection{Comparison to Other Methods}\label{s:competitorComparisonAppendix}
In \cref{s:competitorComparison}, we have compared PHOBIC to state-of-the-art MPHFs using multithreading and the GPU.
In particular, \cref{tab:overviewTable} illustrates the overall comparison.
\Cref{fig:overviewTableLargeMachine} illustrates the result of the same experiment
but using a machine with more cores.
This machine is equipped with an AMD EPYC Rome 7702P processor with 64 cores (with hyper-threading) and a clock frequency of 2.0 GHz.
While this machine does not have a GPU, it shows how well PHOBIC scales with a large number of threads.
We also note that the query times are higher compared to those measured with the machine
used for \cref{tab:overviewTable} because of the different clock frequency.

Lastly, \cref{fig:paretoGpu} shows the construction throughput of PHOBIC and RecSplit on the GPU, by varying the number of bits/key.
As it is clear, PHOBIC construction is as fast as RecSplit -- the fastest construction method up to date. However, we remark that PHOBIC is considerably faster to query than
RecSplit.

%% file: fig/partitionsStackedPlot.tex
\centering
\begin{tikzpicture}
  \begin{axis}[
      yticklabel style={
        /pgf/number format/fixed,
        /pgf/number format/precision=5
      },
      ybar stacked,
      bar width=1.7,
      scaled y ticks=false,
      ylabel={Construction (ns/key)},
      xlabel={Average Partition Size},
      width=5.5cm,
      height=3.5cm,
      xmode=log,
      legend style={at={(0.37,0.7)}, anchor=south},
      every axis plot post/.append style={draw=none},
    ]
    \addplot coordinates { (3072,5.26232) (6144,4.86939) (12288,4.64035) (24576,4.38877) (49152,3.98453) (98304,4.15533) (196608,3.90612) (393216,4.06853) (786432,3.86302) (1572864,3.99217) (3145728,4.33797) (6291456,4.42901) (12582912,4.25987) };
    \addlegendentry{Partitioning};

    \addplot coordinates { (3072,7.19778) (6144,7.51555) (12288,7.8895) (24576,9.75061) (49152,11.31589) (98304,11.77864) (196608,12.30357) (393216,12.75175) (786432,13.2938) (1572864,13.89804) (3145728,14.95732) (6291456,14.47989) (12582912,13.84202) };
    \addlegendentry{Mapping/Ordering};

    \addplot coordinates { (3072,21.53378) (6144,21.65779) (12288,21.65687) (24576,26.30323) (49152,29.39722) (98304,30.19319) (196608,30.42855) (393216,30.81407) (786432,32.12753) (1572864,33.57049) (3145728,36.67234) (6291456,44.34853) (12582912,51.12865) };
    \addlegendentry{Searching};

    \addplot coordinates { (3072,2.23649) (6144,2.18562) (12288,2.39071) (24576,2.68862) (49152,1.7901) (98304,1.40432) (196608,1.44967) (393216,1.07865) (786432,0.8062) (1572864,0.73834) (3145728,0.72759) (6291456,0.78938) (12582912,0.67885) };
    \addlegendentry{Encoding};
  \end{axis}
\end{tikzpicture}

%% file: fig/partitionPlot.tex
\centering
\begin{tikzpicture}

\begin{axis}[
	xlabel={Average Partition Size},
	ylabel={Space (Bits/Key)},
	ytick distance={0.1},
	width=5cm,
	height=3.5cm,
	xmode=log,
	ymax=2.4,
	legend style={at={(0.5,0.97)},anchor=north},
	]
 \addplot coordinates { (3072,2.124726) (6144,2.133137) (12288,2.152382) (24576,2.182734) (49152,2.236858) (98304,2.344705) (196608,2.539653) };
 \addlegendentry{IC-R};
 \addplot coordinates { (3072,2.180024) (6144,2.209925) (12288,2.223311) (24576,2.238719) (49152,2.246916) (98304,2.254558) (196608,2.270957) };
 \addlegendentry{mono-R};
\end{axis}

\end{tikzpicture}

%% file: fig/paretoGpu.tex
\centering
    \begin{tikzpicture}
        \begin{axis}[
            xlabel={Space (Bits/Key)},
            width=5.5cm,
            height=3.5cm,
            xmax=2.2,
            ymax=45,
            ylabel={Throughput (MKeys/s)},
            legend style={at={(0.97,0.6)},anchor=east},
            ymax=80,
            xtick distance={0.2},
          ]
          \addplot coordinates { (1.69457,1.83884) (1.70087,3.02307) (1.70655,4.76508) (1.71343,7.59359) (1.72139,11.7786) (1.7325,17.8987) (1.74385,25.2016) (1.7551,34.638) (1.76924,41.6146) (1.78689,49.3097) (1.80604,58.2411) (1.82919,61.3121) (1.85447,64.5161) (1.88903,65.5308) (1.93462,68.6813) (1.99431,64.7249) (2.07851,69.6864) (2.21443,71.0227) (2.40269,63.8978) (2.6717,63.3714) (3.03025,68.6813) };
          \addlegendentry{PHOBIC-GPU};
          \addplot coordinates { (1.53522,0.0269271) (1.54726,0.572787) (1.55053,0.0661941) (1.55992,1.08155) (1.56039,0.721844) (1.57417,1.20589) (1.61398,41.7885) (1.62854,42.9738) (1.66382,70.7214) (1.67854,74.0192) (1.72507,71.2251) (1.74084,72.6744) (1.74533,73.3138) (1.80877,72.5163) (1.92657,66.6667) (1.94203,68.1199) (2.01538,67.0241) };
          \addlegendentry{RecSplit-GPU};
        \end{axis}
    \end{tikzpicture}

%% file: fig/overviewTableLargeMachine.tex
\addtolength\tabcolsep{-1.9pt}
\small
\centering
	\begin{tabular}[t]{l rrrrr}
		\toprule
        Method & Space & Query & \multicolumn{3}{c}{Construction (ns/key)} \\
               \cmidrule(lr){4-6}
    & (bits/key) & (ns/query) & 1 Thread & 64 Threads & Speedup \\ \midrule
		
                   Bip. SH-RS, $n$=$64$, $b$=$2000$ & 1.52 & 425 &   9\,365 &     356 & 26.3 \\ \midrule
                                 CHD, $\lambda$=$3$ & 2.27 & 357 &      756 &       - &    - \\
                                 CHD, $\lambda$=$5$ & 2.07 & 320 &   5\,483 &       - &    - \\ \midrule
                               FMPH, $\gamma$=$2.0$ & 3.40 & 216 &      126 &       9 & 12.9 \\
                               FMPH, $\gamma$=$1.0$ & 2.80 & 269 &      190 &      15 & 12.1 \\ \midrule
                   SIMDRecSplit, $n$=$8$, $b$=$100$ & 1.81 & 315 &      303 &      13 & 22.1 \\
                 SIMDRecSplit, $n$=$14$, $b$=$2000$ & 1.59 & 368 &  40\,168 &     763 & 52.6 \\ \midrule
    SicHash, $\alpha$=$0.9$, $p_1$=$21$, $p_2$=$78$ & 2.41 & 159 &      191 &       7 & 24.1 \\
   SicHash, $\alpha$=$0.97$, $p_1$=$45$, $p_2$=$31$ & 2.08 & 143 &      260 &      17 & 15.1 \\ \midrule
      PTHash, $\lambda$=$4.0$, $\alpha$=$0.99$, C-C & 3.19 &  79 &      504 &     304 &  1.7 \\
       PTHash, $\lambda$=$5.0$, $\alpha$=$0.99$, EF & 2.11 & 160 &   1\,016 &     307 &  3.3 \\
      PTHash, $\lambda$=$10.5$, $\alpha$=$0.99$, EF & 1.86 & 150 & 139\,184 & 10\,312 & 13.5 \\ \midrule
  PTHash-HEM, $\lambda$=$4.0$, $\alpha$=$0.99$, C-C & 3.19 &  88 &      489 &       9 & 49.1 \\
   PTHash-HEM, $\lambda$=$5.0$, $\alpha$=$0.99$, EF & 2.11 & 166 &      930 &      13 & 66.9 \\ \midrule
      PHOBIC, $\lambda$=$3.9$, $\alpha$=$1.0$, IC-C & 3.18 &  94 &      299 &       9 & 31.3 \\
      PHOBIC, $\lambda$=$4.5$, $\alpha$=$1.0$, IC-R & 2.11 & 240 &      369 &      11 & 33.5 \\
      PHOBIC, $\lambda$=$6.5$, $\alpha$=$1.0$, IC-R & 1.85 & 217 &   1\,341 &      29 & 46.0 \\
      PHOBIC, $\lambda$=$9.0$, $\alpha$=$1.0$, IC-R & 1.74 & 205 &  12\,138 &     251 & 48.3 \\
		\bottomrule
	\end{tabular}